\newif\iftr
\trtrue
\newif\ifproofs
\iftr
\proofstrue
\else
\proofsfalse
\fi

\newif\ifappendix
\iftr
\appendixtrue
\else
\appendixfalse
\fi

\iftr
\documentclass[12pt]{article}
\usepackage{fullpage}
\usepackage{authblk}
\usepackage[bookmarks]{hyperref}
\else
 \documentclass[acmsmall, 10pt, screen]{acmart}
 \fi
\usepackage{xspace}
\usepackage{verbatim}
\usepackage{latexsym,amsthm,amsmath,amsfonts,amssymb,stmaryrd,mathtools}
\usepackage{flushend}
\usepackage{subcaption}
\usepackage{multirow}
\usepackage{ifthen}
\newcounter{pnum}
\newcounter{cnum}
\newcommand{\printnum}
    {\ifthenelse{\value{cnum}>0}
                {(\alph{cnum})\addtocounter{cnum}{1}}
                {\addtocounter{pnum}{1}(\thepnum)}}
\newcommand{\have}[1]{\printnum & #1}
\newcommand{\just}[1]{& \text{(#1)}}
\newcommand{\cjust}[1]{\quad\hfill \text{(#1)}}
\newcommand{\ljust}[1]{\\ & \hfill \text{(#1)}}

\newenvironment{plines}
               {\setcounter{pnum}{0}
               \setcounter{cnum}{0}
                 \begin{array}{l l r l}
               }
               {
                 \end{array}
               }
\newcommand{\pcase}[1]
           {\setcounter{cnum}{1}
            \\
            & \hspace{-3em}\text{Subcase:}~#1\\
           }

\usepackage{listings}
\usepackage{wrapfig}

\usepackage{booktabs}
\usepackage{scalerel}
\usepackage{accents}

\newcommand{\altdiv}{\mathbin{\sim\hspace{-1.0ex}:}}

\newcommand{\langname}{\ensuremath{\mathsf{PriML}}\xspace}
\newcommand{\calcname}{\ensuremath{\lambda^{4}}\xspace}

\newcommand{\rulename}[1]{\textsc{#1}}

\newcommand{\ctx}{\Gamma}

\newtheorem{thm}{Theorem}
\newtheorem{lem}{Lemma}
\newtheorem{cor}{Corollary}
\theoremstyle{definition}
\newtheorem{defn}{Definition}

\newcommand{\ectx}{\cdot}

\newcommand{\dom}[1]{\mathit{dom}(#1)}

\newcommand{\fc}{C}

\newcommand{\defeq}{\triangleq}

\newcommand{\delay}{\delta}

\newcommand{\priowork}[1]{W_{#1}}
\newcommand{\prioworkof}[2]{\priowork{#2}(#1)}
\newcommand{\psnle}[1]{\not\prec #1}
\newcommand{\longp}[2]{S_{#2}(#1)}

\newcommand{\anc}[2]{#1 \sqsupseteq #2}
\newcommand{\nanc}[2]{#1 \not\sqsupseteq #2}
\newcommand{\neqanc}[2]{#1 \sqsupset #2}

\newcommand{\noancs}[2][\graph]{\ensuremath{\not\uparrow\hspace{-0.3em} #2}}
\newcommand{\nodecs}[2][\graph]{\ensuremath{\not\downarrow\hspace{-0.3em} #2}}

\newcommand{\compwork}[2]{\ensuremath{\not\uparrow\hspace{-0.5em}\downarrow\hspace{-0.3em} #2}}

\newcommand{\exec}[1]{\mathit{Exec}(#1)}

\newcommand{\uprio}[2]{\mathit{Prio}_{#1}(#2)}

\newcommand{\uthread}{\vec{u}}
\newcommand{\tscomp}[1]{\cdot}

\newcommand{\sthread}[1]{#1}
\newcommand{\cthread}[3]{#1 \underset{#2}{\hookrightarrow} #3}
\newcommand{\tgraph}[3]{[\cthread{#1}{#2}{#3}]}

\newcommand{\dthreads}{\mathcal{T}}
\newcommand{\spawns}{E_s}
\newcommand{\syncs}{E_j}
\newcommand{\polls}{E_p}
\newcommand{\dagq}[4]{\ensuremath{(#1, #2, #3)}}

\newcommand{\resptimeof}[1]{\ensuremath{T(#1)}}

\newcommand{\kw}[1]{\mbox{\texttt{#1}}}
\newcommand{\cdparens}[1]{({#1})}
\newcommand{\cdsqbracks}[1]{[#1]}
\newcommand{\cd}[1]{{\lstinline!#1!}}

\newcommand{\kwunit}{\kw{unit}}
\newcommand{\kwnat}{\kw{nat}}
\newcommand{\prodsym}{\ensuremath{\times}}
\newcommand{\kwprod}[2]{\ensuremath{{#1} \prodsym {#2}}}
\newcommand{\sumsym}{\ensuremath{+}}
\newcommand{\kwsum}[2]{\ensuremath{{#1} \sumsym {#2}}}
\newcommand{\arrsym}{\ensuremath{\to}}
\newcommand{\kwarr}[2]{\ensuremath{{#1} \arrsym {#2}}}
\newcommand{\cmdsym}{\ensuremath{\kw{cmd}}}
\newcommand{\kwcmdt}[2]{\ensuremath{#1~\cmdsym \cdsqbracks{#2}}}
\newcommand{\atsym}{\ensuremath{\mathop{\kw{thread}}}}
\newcommand{\kwat}[2]{\ensuremath{#1 \atsym \cdsqbracks{#2}}}

\newcommand{\fasym}{\ensuremath{\forall}}
\newcommand{\kwforall}[3]{\ensuremath{\fasym #1: #2. #3}}

\newcommand{\cons}{C}
\newcommand{\cconj}[2]{\ensuremath{#1 \land #2}}
\newcommand{\prio}{\rho}
\newcommand{\prioc}{\overline{\prio}}
\newcommand{\vprio}{\pi}
\newcommand{\ple}[2]{#1 \preceq #2}
\newcommand{\plt}[2]{#1 \prec #2}
\newcommand{\prel}{\ensuremath{\prec}}
\newcommand{\nple}[2]{#1 \not\preceq #2}
\newcommand{\isprio}{~\kw{prio}}

\newcommand{\worlds}{R}
\newcommand{\prios}{R}

\newcommand{\kwnumeral}[1]{\overline{#1}}
\newcommand{\kwn}{\kwnumeral{n}}
\newcommand{\kwtriv}{\langle \rangle}
\newcommand{\kwfun}[2]{\ensuremath{\lambda #1{.}#2}}
\newcommand{\kwfix}[3]{\ensuremath{\mathop{\kw{fix}}#1{:}#2\mathbin{\kw{is}}#3}}

\newcommand{\kwifz}[4]{\ensuremath{\kw{ifz}~#1~\{#2; #3. #4\}}}
\newcommand{\kwpair}[2]{\ensuremath{\langle{#1},{#2}}\rangle}
\newcommand{\kwepair}[2]{\ensuremath{\cdparens{{#1},{#2}}}}
\newcommand{\kweinl}[1]{\ensuremath{\kw{inl}~#1}}
\newcommand{\kwinl}[1]{\ensuremath{\kw{l}\cdot#1}}
\newcommand{\kweinr}[1]{\ensuremath{\kw{inr}~#1}}
\newcommand{\kwinr}[1]{\ensuremath{\kw{r}\cdot#1}}
\newcommand{\kwapply}[2]{\ensuremath{{#1}~{#2}}}
\newcommand{\kwfst}[1]{\ensuremath{\kw{fst}~#1}}
\newcommand{\kwsnd}[1]{\ensuremath{\kw{snd}~#1}}
\newcommand{\kwcase}[5]{\ensuremath{\kw{case}~#1~\{#2.#3; #4.#5\}}}

\newcommand{\kwbspawn}[3]{\ensuremath{\kw{spawn}\cdsqbracks{#1; #2}~\{#3\}}}

\newcommand{\kwbsync}[1]{\ensuremath{\kw{sync}~#1}}

\newcommand{\kwtid}[1]{\ensuremath{\kw{tid}\cdsqbracks{#1}}}

\newcommand{\kwwapp}[2]{\ensuremath{#1\cdsqbracks{#2}}}

\newcommand{\kwcmd}[2]{\ensuremath{\kw{cmd}\cdsqbracks{#1}~\{#2\}}}
\newcommand{\kwwlam}[3]{\ensuremath{\Lambda #1: #2. #3}}
\newcommand{\kwlet}[3]{\ensuremath{\kw{let}~#1 = #2~\kw{in}~#3}}

\newcommand{\cmd}{m}

\newcommand{\kwbind}[3]{\ensuremath{#2 \leftarrow #1; #3}}
\newcommand{\kwoutput}[1]{\ensuremath{\kw{output}~#1}}
\newcommand{\kwinput}{\ensuremath{\kw{input}}}

\newcommand{\kwret}[1]{\ensuremath{\kw{ret}~#1}}

\newcommand{\fresh}{~\kw{fresh}}
\newcommand{\dassn}{\Delta}
\newcommand{\ceval}[4][\dassn]{#2 \downarrow #3; #4}
\newcommand{\cceval}[9]{#1; #2; #5 \downarrow_{(#3, #4)} #6; #7; #8; #9}
\newcommand{\mceval}[6][\dassn]{#2; #3; #4 \downarrow #6; #5}

\newcommand{\graph}{g}
\newcommand{\ethread}{{[]}}
\newcommand{\sgraph}[1]{[#1]}
\newcommand{\scompsym}{\oplus}
\newcommand{\scomp}[1]{\scompsym_{#1}}
\newcommand{\gscomp}[1]{\overline{\scompsym}_{#1}}
\newcommand{\egraph}{\emptyset}

\newcommand{\dthread}[3]{\ensuremath{#1 \xhookrightarrow[#2]{} #3}}
\newcommand{\mthread}[4][]{\ensuremath{#2 \xhookrightarrow[#3]{} #4}}
\newcommand{\mcp}{\mathop{\uplus}}
\newcommand{\mem}{\mu}
\newcommand{\sig}{\Sigma}
\newcommand{\esig}{\cdot}
\newcommand{\tsig}{\sigma}
\newcommand{\emem}{\emptyset}

\newcommand{\meq}{\ensuremath{\equiv}}

\newcommand{\mbconfig}[2]{\ensuremath{\nu #1 \{ #2 \}}}

\newcommand{\rpconfig}[3]{\ensuremath{(#1, #2, #3)}}
\newcommand{\execstosymbol}[2]{\ensuremath{\xRightarrow[#1]{#2}}}
\newcommand{\gstep}[2]{\mathrel{\execstosymbol{#1}{#2}}}
\newcommand{\pgstep}[1]{\mathrel{\execstosymbol{P}{#1}}}
\newcommand{\estep}[1]{\ensuremath{\to}}
\newcommand{\tstep}[1][\prel]{\ensuremath{\Rightarrow}}
\newcommand{\mstep}[1][\prel]{\ensuremath{\Rrightarrow}}

\newcommand{\val}[1]{\ensuremath{~\kw{val}_{#1}}}
\newcommand{\final}[1]{\ensuremath{~\kw{final}}}

\newcommand{\atyped}[3][\worlds]{\vdash^{#1}_{#2} #3~\mathsf{action}}
\newcommand{\llblstepstosymbol}[2]{\ensuremath{\xmapsto[#1]{#2}}}
\newcommand{\lstep}[3][\prio]{\mathrel{\llblstepstosymbol{#2}{#3}}}
\newcommand{\asil}{\epsilon}
\newcommand{\act}{\alpha}
\newcommand{\acts}{\ensuremath{A}}

\newcommand{\async}[2]{#1 \mathbin{?}#2}
\newcommand{\tact}[2]{#1/#2}
\newcommand{\asend}[2]{\mathbin{!}#2}

\newcommand{\hastype}[2]{#1 \mathop{:} #2}
\newcommand{\sigtype}[3]{#1 \mathord{\sim} #2 \mathord{@} #3}

\newcommand{\sigcent}[3]{#1\hookrightarrow (#2, #3)}
\newcommand{\tsigent}[3]{#1\hookrightarrow (#2, #3)}
\newcommand{\etyped}[5][\worlds]{#3 \vdash^{#1}_{#2} #4 : #5}
\newcommand{\cmdtyped}[6][\worlds]{#3 \vdash^{#1}_{#2} #4 \altdiv #5 \mathbin{@} #6}
\newcommand{\tsigtyped}[3][\worlds]{\vdash^{#1}_{#2} #3}
\newcommand{\mtyped}[4]{\vdash_{#2}^{#1} #3 : #4}
\newcommand{\meetc}[3][\worlds]{#2 \vdash^{#1} #3}

\newcommand{\secref}[1]{Section~\ref{sec:#1}}

\newcommand{\figref}[1]{Figure~\ref{fig:#1}}

\renewcommand{\eqref}[1]{Equation~(\ref{eq:#1})}

\newcounter{remark}[section]

\newcommand{\code}[1]{\lstinline!#1!}

\usepackage{math-cmds}
\usepackage{mathpartir}
\usepackage{code}

\iftr
\else
\setcopyright{rightsretained}
\acmPrice{}
\acmDOI{10.1145/3236790}
\acmYear{2018}
\copyrightyear{2018}
\acmJournal{PACMPL}
\acmVolume{2}
\acmNumber{ICFP}
\acmArticle{95}
\acmMonth{9}


\fi

\iftr
\usepackage{natbib}
\usepackage{url}
\newtheorem{lemma}{Lemma}
\else
\fi

\usepackage{microtype}

\begin{document}

%\iftr
%\author{Stefan K. Muller \and Umut A. Acar \and Robert Harper}
%\date{April 2017}
%\trnumber{CMU-CS-17-107}

%\citationinfo{A version of this work appears in the proceedings of the
%ACM SIGPLAN Conference on Programming Language Design and Implementation
%(PLDI 2017).}

%\support{This research is partially supported
%grants from the National Science Foundation (CCF-1320563, CCF-1408940,
%CCF-1629444) and European Research Council (grant
%ERC-2012-StG-308246), and by a gift from Microsoft Research.}

%\else
%\toappear{}

%\special{papersize=8.5in,11in}
%\setlength{\pdfpageheight}{\paperheight}
%\setlength{\pdfpagewidth}{\paperwidth}

%\conferenceinfo{CONF 'yy}{Month d--d, 20yy, City, ST, Country}
%\copyrightyear{2018}
%\copyrightdata{978-1-nnnn-nnnn-n/yy/mm}
%\copyrightdoi{nnnnnnn.nnnnnnn}
\iftr
\author[1]{Stefan K. Muller\thanks{smuller@cs.cmu.edu}}
%\affil[1]{Carnegie Mellon University}
\author[2]{Umut A. Acar\thanks{umut@cs.cmu.edu}}
\author[3]{Robert Harper\thanks{rwh@cs.cmu.edu}}
\affil[1,2,3]{Carnegie Mellon University}
\affil[2]{Inria}
\date{July 2018}
\else
\author{Stefan K. Muller}
\affiliation{\institution{Carnegie Mellon University}\country{USA}}
\email{smuller@cs.cmu.edu}

\author{Umut A. Acar}
\affiliation{\institution{Carnegie Mellon University}\country{USA}}
\affiliation{\institution{Inria}\country{France}}
\email{umut@cs.cmu.edu}

\author{Robert Harper}
\affiliation{\institution{Carnegie Mellon University}\country{USA}}
\email{rwh@cs.cmu.edu}
\fi

%\title{Cooperative Multithreading with Priorities}
\title{Competitive Parallelism: Getting Your Priorities Right}
%\title{Responsive Parallelism with Partially Ordered Priorities}

%\ifproofs
%\subtitle{Extended Supplementary Version}
%\fi

\begin{abstract}
  Multi-threaded programs have traditionally fallen into one of two domains:
  {\em cooperative} and {\em competitive}.
  %threading, which maximizes the throughput of a large number
  %of equally important threads, and competitive threading, which focuses on
  %responsiveness of threads which may have different requirements and
  %priorities.
%
These two domains have traditionally remained mostly disjoint, with
cooperative threading used for increasing {\em throughput} in
compute-intensive applications such as scientific workloads and
cooperative threading used for increasing {\em responsiveness} in
interactive applications such as GUIs and games.
% usually using a single
%processors.
%
As multicore hardware becomes increasingly mainstream, there is a need
for bridging these two disjoint worlds, because many applications mix
interaction and computation and would benefit from both cooperative
and competitive threading.

In this paper, we present techniques for programming and reasoning
about {\em parallel interactive} applications that can use both
cooperative and competitive threading.
Our techniques enable the programmer to write rich parallel
interactive programs by creating and synchronizing with threads as
needed, and by assigning threads user-defined and partially ordered
priorities.
To ensure important responsiveness properties, we present a modal type
system analogous to S4 modal logic that precludes low-priority threads from
delaying high-priority threads, thereby statically preventing a
crucial set of {\em priority-inversion} bugs.
%, which can have serious
%and even disastrous consequences.
%
We then present a cost model that allows reasoning about
responsiveness and completion time of well-typed programs.
The cost model extends the traditional work-span model for
cooperative threading to account for competitive scheduling decisions
needed to ensure responsiveness.
Finally, we show that our proposed techniques are realistic by
implementing them as an extension to the Standard ML language.
\end{abstract}

\maketitle

%\category{D.3.3}{Language Constructs and Features}{Input/output}
%\category{F.1.2}{Modes of Computation}{Parallelism and concurrency}

%\iftr
%\renewcommand*{\thepage}{title-\arabic{page}}
%\maketitle
%\renewcommand*{\thepage}{\arabic{page}}
%\fi

\section{Introduction}
The increasing proliferation of multicore hardware
% is bringing parallel
%programming to the masses, and sparking
has sparked a renewed interest in programming-language support for
{\em cooperative threading}.
In cooperative threading, threads correspond to pieces of a job and
are scheduled with the goal of completing the job as quickly as
possible---or to maximize {\em throughput}. 
Cooperative thread scheduling algorithms are therefore usually
non-preemptive: once a thread starts executing, it is allowed to
continue executing until it completes.

Cooperatively threaded languages such as NESL~\citep{nesl-94},
Cilk~\citep{FrigoLeRa98}, parallel
Haskell~\citep{haskell-dp-2007,Keller+2010} and parallel
ML~\citep{manticore-implicit-11,JagannnathanNaSiZi10,RMAB-mm-2016},
have at least two important features:
\begin{itemize}
\item The programmers can express opportunities for parallelism at a
  high level with relatively simple programming abstractions such as
  {\em fork/join} and {\em async/finish}.
The run-time system of the language then handles the creation and
scheduling of the threads.

\item The efficiency and performance of parallel programs written at
  this high level can be analyzed by using cost models based on {\em
    work} and {\em span} (e.g.~\citep{BlellochGr95,BlellochGr96,EagerZaLa89,
    SpoonhowerBlHaGi08}), which can guide efficient implementations.%.
  %The schedulers used by these
  %languages can usually closely match these cost models.
\end{itemize}

Cooperative threading is elegant and expressive but it mostly excludes the
important class of {\em interactive applications}, which require
communication with the external world, including users and other
programs.
Such interactive applications typically require {\em responsiveness},
such as the requirement to process user input as soon as possible.
Ensuring responsiveness usually requires {\em competitive threading},
where threads are scheduled pre-emptively, usually based on priorities.
To guarantee responsiveness, most competitive threading libraries in use
today expose a fixed range of numerical priorities which may be assigned
to threads.
Regardless of the threading primitives used, this greatly complicates the
task of writing programs:
\begin{itemize}
\item Writing effective competitively threaded programs requires
  assigning priorities to threads.  While this can be simple for
  simple programs, using priorities at scale is a big challenge
  because most current approaches to priorities are inherently
  anti-modular.  Because priorities are totally ordered, writing
  responsive programs might require reasoning about whether a thread
  should be given a higher or lower priority than a thread introduced
  in another part of the program, or possibly even in a library
  function.

\item To compensate for this lack of modularity, many systems expose large
  numbers of priorities: the POSIX threads (pthreads) API exposes scheduling
  policies with as many as 100 levels. Without clean guidelines governing
  their use, however, programmers must still reason globally about how to
  assign these numbers to threads. Studies have shown that programmers
  struggle to use systems with even 7 priorities~\citep{hauserjathwewe93}.

\item Reasoning about performance is much more difficult: the clean
  work-span model of cooperative threading does not apply to
  competitive threading, because of the impact of priorities on
  run-time. Furthermore, in competitive threading, {\em priority
    inversions}, where a low-priority thread delays a high-priority
  one, can have harmful and even disastrous consequences. For example, ``Mars
    Pathfinder'', which landed on Mars on 4 July
    1997, suffered from a software bug, traced to a priority inversion,
    that caused the craft to reset
    itself periodically. The bug
    had to be patched remotely so the mission could continue.

\end{itemize}

In this paper, we develop language techniques and a cost model for
writing parallel interactive programs that use a rich set of
cooperative and competitive threading primitives.
This problem is motivated by the fact that as shared-memory hardware
becomes widely used, competitively threaded, interactive applications
will need to take advantage of the benefits of this parallel hardware,
and not just cooperatively threaded, compute-intensive
applications.

We present a programming language with features for {\em spawning} and
{\em syncing} with asynchronous threads, which may be assigned priorities
by the programmer.
Aside from priorities, these threads are equivalent to
futures, a powerful general-purpose cooperative threading mechanism.
Like futures, threads are first-class values in the language.
To enable modular programming with priorities, we allow the programmer
to declare any number of priorities and define a partial order between
them.
The resulting language is sufficiently powerful to enable both
cooperative and competitive threading.
For example, the programmer can write a purely compute intensive
program (e.g., parallel Quicksort), a purely interactive program
(e.g. a simple graphical user interface), and anything that combines
the two (e.g. an email client that sorts tens of thousands of emails in
parallel in the background while remaining responsive to user interaction
events).

To reason about the efficiency and responsiveness of the programs
written in this language, we present a cost model that bounds
both the total computation time of the program and the
response time of individual threads.
Our cost semantics extends prior cost models of cooperative parallel
programs %~\citep{BlellochGr96,SpoonhowerBlHaGi08, MullerAcHa17}
to enable reasoning about the response time of threads with
partially-ordered priorities.
The main theoretical result of the paper shows that the response time of
a thread does not depend on the amount of computation performed at lower
priorities for any program in which threads do not sync on threads
of lower priority.
Such a sync clearly allows the response time of a high-priority thread
to depend on low-priority work and is an example of the classic problem
of priority inversions described above.

Our prior work on extending cooperative threading with
priorities~\citep{MullerAcHa17} also observed that priority inversions
prevent responsiveness guarantees and presented static mechanisms for
avoiding them.
That work, however, considers only two priorities (high and low).
Research in languages such as Ada~\citep{cornhillsh87, levine88}
also discusses the importance of preventing priority inversion in a
general setting with rich priorities, but we are aware of no prior
static language mechanisms for doing so.

To guarantee appropriate bounds on responsiveness, we specify a type
system that statically identifies and prevents priority inversions that
would render such an analysis impossible.
The type system enforces a monadic separation between {\em commands}, which
are restricted to run at a certain priority, and {\em expressions}, which are
priority-invariant.
The type system then tracks the priorities of threads and
rejects programs in which a high-priority thread may synchronize with a
lower-priority one.
In developing this system, we draw inspiration from modal logics, where
the ``possible worlds'' of the modal logic correspond to priorities in our
programs.
More specifically, our type system is analogous to S4 modal logic,
where the accessibility relation between worlds is assumed to be
reflexive and transitive.
This accessibility relation reflects the fact that the ways in which
priorities are intended to interact is inherently asymmetric.
Modal logic has proved to be effective in many problems of computer
science. For example, \citet{murphy04}, and~\citet{JiaWa04}
use the modal logic S5, where the accessibility relation between
worlds is assumed to be symmetric (as well as reflexive and transitive), to
model distributed computing.

The dynamic semantics of our language is a transition system that
simulates, at an abstract level, the execution of a program
on a parallel machine.
We show that, for well-typed programs, our cost model accurately predicts
the response time of threads in such an execution.
Finally, we show that the proposed techniques can be incorporated into
a practical language by implementing a compiler which typechecks
prioritized programs and compiles them to a parallel version of
Standard ML.
We also provide a runtime system which schedules threads
according to their priorities. 

%% 
%This paper shows that a reasonably broad class of multithreaded
%programs can safely use priorities.
%
The specific contributions of this paper include the following.
\begin{itemize}
\item An extension of the Parallel ML language, called {\langname},
  with language constructs for user-defined, partially ordered priorities.

\item A core calculus {\calcname} that captures the essential ideas of
  {\langname} and a type system that guarantees inversion-free use of
  threads.

\item A cost semantics for {\calcname} which can be used to make
  predictions about both overall computation time and responsiveness, and a
  proof that these predictions are accurately reflected by the
  dynamic semantics.

\item
An implementation of the compiler and the runtime system for
{\langname} as an extension of the Parallel MLton compiler.

\item
  Example benchmarks written in our implementation that give
  preliminary qualitative evidence for the practicality of the
  proposed techniques.
\end{itemize}

%%% Local Variables:
%%% mode: latex
%%% TeX-master: "main"
%%% End:

\section{Overview}\label{sec:overview}
We present an overview of our approach to multithreaded programming
with priorities by using a language called {\langname} that extends
Standard ML with facilities for prioritized multithreaded
programming.
As a running example, we consider an email client
which interacts with a user while performing other necessary
tasks in the background.
The purpose of this section is to highlight the main
ideas.
The presentation is therefore high-level and sometimes informal.
The rest of the paper formalizes these ideas (\secref{lang}),
expands on them to place performance bounds on {\langname} programs
(\secref{cost}) and
describes how they may be realized in practice (\secref{impl}).

\paragraph{Priorities.}

{\langname} enables the programmer to define priorities as needed and
specify the relationships between them.
For example, in our mail client, we sometimes wish to alert the user
to certain situations (such as an incoming email) and
we also wish to compress old emails in the background when the system is idle.
To express this in {\langname}, we define two priorities \cd{alert}
and \cd{background} and order them accordingly as follows.
\begin{lstlisting}[numbers=none]
priority alert
priority background
order background < alert
\end{lstlisting}
The ordering constraint specifies that \cd{background} is lower
priority than \cd{alert}.
Programmers are free to specify as many, or as few, ordering constraints
between priorities as desired.
{\langname} therefore provides support for a set of partially ordered
priorities.
Partially ordered priorities suffice to capture the intuitive notion
of priorities, and to give the programmer flexibility to
express any desired priority behavior, but without the burden of having to
reason about a total order over all priorities.
Consider two priorities \cd{p} and \cd{q}.
If they are ordered, e.g., \cd{p < q}, then the system is instructed to run
threads with priority~\cd{q} over threads with priority~\cd{p}.
If no ordering is specified (i.e.~\cd{p} and~\cd{q} are incomparable in the
partial order), then the system is free to choose arbitrarily between
a thread with priority \cd{p} and another with priority \cd{q}.

\paragraph{Modal type system.}
To ensure responsive use of priorities, {\langname} provides a modal type
system that tracks priorities. The types of
{\langname} include the standard types of functional
programming languages as well as
a type of thread handles, by which computations can refer to, and
synchronize with, running threads.

To support computations that can operate at multiple priorities, the
type system supports {\em priority polymorphism} through
a polymorphic type of the form $\forall \pi: C.\tau$, where $\pi$ is
a newly bound priority variable, and $C$ is a set of constraints of the form
$\ple{\prio_1}{\prio_2}$ (where~$\prio_1$ and~$\prio_2$ are priority constants
or variables, one of which will in general be~$\pi$), which bounds the
allowable instantiations of~$\pi$.

To support the tracking of priorities, the
syntax and type system of {\langname} distinguish between commands
and expressions.
{\em Commands} provide the constructs for spawning and synchronizing
with threads.
{\em Expressions} consist of an ML-style functional language, with some
extensions. Expressions cannot directly
execute commands or interact with threads, and can thus be evaluated without
regard to priority. Expressions can, however, pass around encapsulated
commands (which have a distinguished type)
and abstract over priorities to introduce priority-polymorphic
expressions.

\paragraph{Threads}
Once declared, priorities can be used to specify the priority of
threads.
For example, in response to a request from the user, the mail client
can spawn a thread to sort emails for background compression, and spawn
another thread to alert the user about an incoming email. Spawned threads
are annotated with a priority and run asynchronously with the rest of the
program.

\begin{lstlisting}[numbers=none]
  spawn[background] { ret (sort ...) };
  spawn[alert] { ret (display ``Incoming mail!'') }
\end{lstlisting}

The~\cd{spawn} command takes a command to run in the new thread and
returns a handle to the spawned thread. In the above
code, this handle is ignored, but it can also be bound to a variable
using the notation \texttt{x <- m;} and
used later to synchronize with the thread (wait for it to complete).

\begin{lstlisting}[numbers=none]
  spawn[background] { ret (sort ...) };
  alert_thread <- spawn[alert] { ret (display ``New mail received'') };
  sync alert_thread
\end{lstlisting}

\paragraph{Example: priority-polymorphic multithreaded quicksort.}
Priority polymorphism allows prioritized code to be compositional.
For example, several parts of our email client might wish to use
a library function \cd{qsort} for sorting (e.g., the background thread
sorts emails by date to decide which ones to compress and a higher-priority
thread sorts emails by subject when the user clicks a column header.)
Quicksort is easily parallelized, and so the library code spawns threads
to perform recursive calls in parallel. The use of threads, however,
means that the code must involve priorities and cannot be purely an
expression.
Because sorting is a basic function and may be used at many priorities,
We would want the code for \cd{qsort} to be polymorphic over
priorities. %, allowing it to be executed at all suitable priorities.
This is possible in {\langname} by defining \cd{qsort} to operate
at a priority defined by an unrestricted priority variable.

\begin{figure}
\begin{lstlisting}
fun[p] qsort (compare: 'a * 'a -> bool) (s: 'a seq) : 'a seq cmd[p]  =
  if Seq.isEmpty s then
    cmd[p] {ret Seq.empty}
  else
    let val pivot = Seq.sub(s, (Seq.length s) / 2)
        val (s_l, s_e, s_g) = Seq.partition (compare pivot) s
    in
      cmd[p]
      {
        quicksort_l <- spawn[p] {do ([p]qsort compare s_l)};
        quicksort_g <- spawn[p] {do ([p]qsort compare s_g)};
        ss_l <- sync quicksort_l;
        ss_g <- sync quicksort_g;
        ret (Seq.append [ss_l, s_e, ss_g])
      }
    end
\end{lstlisting}
\caption{Code for multithreaded quicksort, which is priority
  polymorphic.}
\label{fig:quicksort-code}
\end{figure}

\figref{quicksort-code} illustrates the code for a multithreaded
implementation of Quicksort in {\langname}.
The code uses a module called \cd{Seq} which implements some basic
operations on sequences.
In addition to a comparison function on the elements of the sequence
that will be sorted and the sequence to sort, the function takes as
an argument a priority $p$, to which the body of the function may refer
(e.g. to spawn threads at that priority)\footnote{Note that, unlike type-level
parametric polymorphism in languages such as ML, which can be left implicit
and inferred during type checking, priority parameters in
{\langname} must be specified in the function declaration.}.
The implementation of \cd{qsort} follows a standard implementation
of the algorithm but is structured according to the type system of
{\langname}.
This can be seen in the return type of the function, which is an
encapsulated command
at priority~\cd{p}.

The function starts by checking if the sequence is empty. If so, it
returns a command that returns an empty sequence.
If the sequence is not empty, it partitions the sequence into sub-sequences
consisting of elements less than, equal to and greater than, a pivot,
chosen to be the middle element of the sequence.
It then returns a command that sorts the sub-sequences in parallel,
and concatenates the sorted sequences to produce the result.
To perform the two recursive calls in parallel, the function \cd{spawn}s two
threads, specifying that the threads operate at priority~\cd{p}.

This code also highlights the interplay between expressions and commands in
{\langname}. The expression \cd{cmd[p] {m}} introduces an encapsulated command,
and the command \texttt{do e} evaluates \texttt{e} to an
encapsulated command, and then runs the command.

\paragraph{Priority Inversions.}
The purpose of the modal type system is to prevent priority inversions, that is,
situations in which a thread synchronizes with a thread of a lower priority.
An illustration of such a situation appears in Figure~\ref{fig:prio1-illtyped}.
This code shows a portion of the main event loop of the email client, which
processes and responds to input from the user. The event loop runs at a high
priority. If the user sorts the emails by date, the loop spawns a new thread,
which calls the priority-polymorphic sorting function. The code instantiates
this function at a lower priority~\code{sort_p}, reflecting the programmer's
intention that the sorting, which might take a significant fraction of a second
for a large number of emails, should not delay the handling of new events.
Because syncing with that thread immediately afterward
(line~\ref{line:prio1-sync})
causes the remainder of the event loop (high-priority) to wait
on the sorting thread (lower priority), this code will be correctly rejected
by the type system.
The programmer could instead write the code as shown in
Figure~\ref{fig:prio1-welltyped}, which displays the sorted list in the new
thread, allowing the event loop to continue processing events. This code does
not have a priority inversion and is accepted by the type system.

\begin{figure}
\begin{subfigure}[t]{0.45\textwidth}
\begin{lstlisting}[escapeinside={@}{@}]
priority loop_p
priority sort_p
order sort_p < loop_p

fun loop emails : unit cmd[loop_p] =
  case next_event () of
  SORT_BY_DATE =>
    cmd[loop_p] {
      t <- spawn[sort_p] {
        do ([sort_p]qsort
                  date emails)};
      l <- sync t;@\label{line:prio1-sync}@
      ret (display_ordered l)
    }
    | ...
\end{lstlisting}
\caption{Ill-typed event loop code}
\label{fig:prio1-illtyped}
\end{subfigure}
\hfill
\begin{subfigure}[t]{0.48\textwidth}
\begin{lstlisting}[escapeinside={@}{@}]
priority loop_p
priority sort_p
order sort_p < loop_p

fun loop emails : unit cmd[loop_p] =
  case next_event () of
  SORT_BY_DATE =>
    cmd[loop_p] {
      spawn[sort_p] {
        l <- do ([sort_p]qsort
                   date emails);
        ret (display_ordered l)
      }
    }
    | ...
\end{lstlisting}
\caption{Well-typed event loop code}
\label{fig:prio1-welltyped}
\end{subfigure}
\caption{Two implementations of the event loop, one of which displays a
priority inversion.}
\label{fig:prio1}
\end{figure}

Although the priority inversion of Figure~\ref{fig:prio1-illtyped} could easily be
noticed by a programmer, the type system also rules out
more subtle priority inversions.
Consider the ill-typed code in Figure~\ref{fig:prio2}, which shows
another way in which a programmer might choose to implement the event loop.
In this implementation, the event loop spawns two threads.
The first (at priority~\code{sort\_p}) sorts the emails, and the second
(at priority~\code{display\_p}) calls a
priority-polymorphic function~\code{[p]disp}, which takes a
sorting thread at priority~\cd{p}, waits for it to complete, and displays
the result. This type of ``chaining'' is a common idiom in programming with
futures, but this attempt has gone awry because the thread at
priority~\code{display\_p} is waiting on the lower-priority sorting thread.
Because of priority polymorphism, it may not be immediately clear where
exactly the priority inversion occurs, and yet this code will still be
correctly rejected by the type system. The type error is on
line~\ref{line:prio2}:
\begin{verbatim}
constraint violated at 9.10-9.15: display_p <= p_1
\end{verbatim}
This~\code{sync} operation is passed a thread of
priority~\code{p} (note from the function signature that the types of thread
handles explicitly track their priorities), and there is no guarantee
that~\cd{p} is higher-priority
than~\cd{display_p} (and, in fact, the instantiation on
line~\ref{line:prio2-call} would violate this constraint). We may
correct the type error in the~\code{disp} function by adding this constraint
to the signature:

\begin{lstlisting}[numbers=none]
fun[p : display_p <= p] disp (t: email seq thread[p]) : unit cmd[display_p] =
\end{lstlisting}

With this change, the instantiation on line~\ref{line:prio2-call} would become
ill-typed, as it should because this way of structuring the code inherently has
a priority inversion. The event loop code should be written as in
Figure~\ref{fig:prio1-welltyped} to avoid a priority inversion. However, the
revised~\code{disp} function could still be called on a higher-priority thread
(e.g. one that checks for new mail).

\begin{figure}
\begin{lstlisting}[escapeinside={@}{@}]
priority loop_p
priority display_p
priority sort_p
order sort_p < loop_p
order sort_p < display_p

fun[p] disp (t : email seq thread[p]) : unit cmd[display_p] =
  cmd[display_p] {
    l <- sync t;@\label{line:prio2}@
    ret (display_ordered l)
  }

fun loop emails : unit cmd[loop_p] =
  case next_event () of
  SORT_BY_DATE =>
    cmd[loop_p] {
      t <- spawn[sort_p] { do ([sort_p]qsort date emails) };
      spawn[display_p] { do ([sort_p]disp t) } @\label{line:prio2-call}@
    }
    | ...
\end{lstlisting}
\caption{An ill-typed attempt at chaining threads together.
%  The high-priority
  %display thread is waiting on the lower-priority sorting thread.
}
\label{fig:prio2}
\end{figure}

Note that the programmer could also fix the type error in both
versions of the code by spawning the sorting thread at a higher priority.
This change, however, betrays the
programmer's intention (clearly stated in the priority annotations)
that the sorting should be lower priority. The purpose of
the type system, as with all such programming language mechanisms, is
not to relieve programmers entirely of the burden of thinking about
the desired behavior of their code, but rather to ensure that the code
adheres to this behavior if it is properly specified.

\section{The {\calcname} calculus}\label{sec:lang}
\begin{figure}
\[
\begin{array}{llll}

  \mathit{Types} & \tau & \bnfdef &
  \kwunit \bnfalt
  \kwnat \bnfalt
  \kwarr{\tau}{\tau} \bnfalt
  \kwprod{\tau}{\tau} \bnfalt
  \kwsum{\tau}{\tau} \bnfalt
  \kwat{\tau}{\prio} \bnfalt
  \kwcmdt{\tau}{\prio} \bnfalt
  \kwforall{\vprio}{\cons}{\tau}
  \\

  \mathit{Priorities} & \prio & \bnfdef &
  \prioc \bnfalt
  \vprio
  \\

  \mathit{Constrs.} & \cons & \bnfdef &
  \ple{\prio}{\prio} \bnfalt
  \cconj{\cons}{\cons}
  \\

  \mathit{Values} & v & \bnfdef & x \bnfalt
  \kwtriv \bnfalt
  \kwnumeral{n} \bnfalt
  \kwfun{x}{e} \bnfalt
  \kwpair{v}{v} \bnfalt
  \kwinl{v} \bnfalt
  \kwinr{v} \bnfalt
  \kwtid{a} \bnfalt
  \kwcmd{\prio}{\cmd} \bnfalt
  \kwwlam{\vprio}{\cons}{e}
  \\

  \mathit{Exprs.} & e & \bnfdef & v \bnfalt
  \kwlet{x}{e}{e} \bnfalt
  \kwifz{v}{e}{x}{e} \bnfalt
  \kwapply{v}{v} \\ & & & \bnfalt
  \kwepair{v}{v} \bnfalt
  \kwfst{v} \bnfalt
  \kwsnd{v} \bnfalt
  \kweinl{v} \bnfalt
  \kweinr{v} \bnfalt
  \kwcase{v}{x}{e}{y}{e} \\ & & & \bnfalt
  \kwoutput{v} \bnfalt
  \kwinput \bnfalt
  \kwwapp{v}{\prio} \bnfalt
  \kwfix{x}{\tau}{e}
  \\

  \mathit{Commands} & \cmd & \bnfdef &
  \kwbind{e}{x}{\cmd} \bnfalt
  \kwbspawn{\prio}{\tau}{\cmd} \bnfalt
  \kwbsync{e} \bnfalt
  \kwret{e}\\
\end{array}
\]
\caption{Syntax of {\calcname}}
\label{fig:syn}
\end{figure}

In this section, we define a core calculus {\calcname} which captures the key
ideas of a language with an ML-style expression layer and a modal layer of
prioritized asynchronous threads.
\ifproofs
\else
Some straightforward rules and proof details which are omitted from this
section for space reasons are available in the extended version~\citep{partial-prio-tr}.
\fi
Figure~\ref{fig:syn} presents the abstract
syntax of {\calcname}. In addition to the unit type, a type of natural numbers,
functions, product types and sum types, {\calcname} has three special types.
The type~$\kwat{\tau}{\prio}$ is used for a handle to an asynchronous
thread running at priority~$\prio$ and returning a value of type~$\tau$.
The type~$\kwcmdt{\tau}{\prio}$ is used for an encapsulated command.
The calculus also has a type~$\kwforall{\vprio}{\cons}{\tau}$ of
priority-polymorphic expressions. These types are annotated with
a constraint~$\cons$ which restricts the instantiation of the bound priority
variable. For example, the
abstraction~$\kwwlam{\vprio}{\ple{\vprio}{\prioc}}{e}$
can only be instantiated with priorities~$\prioc'$ for
which~$\ple{\prioc'}{\prioc}$.

A priority~$\prio$ can be either a
priority constant, written~$\prioc$, or a priority variable~$\vprio$. Priority
constants will be drawn from a pre-defined set, in much the same way that
numerals~$\kwnumeral{n}$ are drawn from the set of natural numbers. The set
of priority constants (and the partial order over them)
will be determined statically and is a parameter to the
static and dynamic semantics. This is a key difference between the calculus
{\calcname} and {\langname}, in which the program can define new priority
constants
(we discuss in Section~\ref{sec:impl} how a compiler can hoist priority
definitions out of the program).
%We will discuss in Section~\ref{sec:elab} how the priority
%definitions of {\langname} may be hoisted out of the program to produce
%{\calcname} programs.

%Expressions and commands are distinguished syntactically.
As in {\langname}, the syntax is separated into expressions, which do not
involve priorities, and commands which do.
For simplicity, the expression language is in ``2/3-cps'' form: we distinguish
between expressions and values, and expressions take only values as arguments
when this would not interfere with evaluation order. An expression with
unevaluated subexpressions, e.g.~$\kwepair{e_1}{e_2}$ can be expressed using
let bindings as~$\kwlet{x}{e_1}{\kwlet{y}{e_2}{\kwepair{x}{y}}}$.
Values consist of
the unit value~$\kwtriv$, numerals~$\kwnumeral{n}$, anonymous
functions~$\kwfun{x}{e}$, pairs of values, left- and right-injection of values,
thread identifiers,
encapsulated commands~$\kwcmd{\prio}{\cmd}$ and priority-level
abstractions~$\kwwlam{\vprio}{\cons}{e}$.

Expressions include values, let binding, the if-zero
conditional~$\kwifz{e}{e_1}{x}{e_2}$ and function application. There are
also additional expression forms for pair introduction and left- and
right-injection. These are~$\kwepair{v_1}{v_2}$,~$\kweinl{v}$ and~$\kweinr{v}$,
respectively. One may think of these forms as the source-level instructions
to allocate the pair or tag, and the corresponding value forms as the actual
runtime representation of the pair or tagged value (separating the two
will allow us to account for the cost of performing the allocation).
Finally, expressions include the case construct~$\kwcase{e}{x}{e_1}{y}{e_2}$,
output, input, priority instantiation~$\kwwapp{v}{\prio}$ and fixed points.

 Commands are combined using the binding
construct~$\kwbind{e}{x}{\cmd}$, which evaluates~$e$ to an encapsulated
command, which it executes, binding its return
value to~$x$, before continuing with command~$\cmd$. Spawning a thread
and synchronizing with a thread are also commands.
The spawn command~$\kwbspawn{\prio}{\tau}{\cmd}$ is parametrized by both
a priority~$\prio$ and the type~$\tau$ of the return value of~$\cmd$ for
convenience in defining the dynamic semantics.

\subsection{Static Semantics}\label{sec:statics}
The type system of {\calcname} carefully tracks the priorities of threads as
they wait for each other and enforces that a
program is free of priority inversions.
This static guarantee will ensure that we can derive cost guarantees from
well-typed programs.
%Even if such a strong static guarantee
%is not desired, tracking the priorities of threads statically gives programmers
%the ability to reason about priority inversions and where they might occur.

\begin{figure}
\[
\begin{array}{c}

\infer[var]
{
  \strut
}
{
  \etyped{\sig}{\ctx, \hastype{x}{\tau}}{x}{\tau}
}

\qquad

\infer[\kwunit I]
{
\strut
}
{
\etyped{\sig}{\ctx}{\kwtriv}{\kwunit}
}

\qquad

\infer[Tid]
      {
        \strut
  %\etyped{\sig}{\ctx}{v}{\tau}\\
  %\ple{\prio}{\prio'}\\
  %\tau\mobile
}
{
  \etyped{\sig, \sigtype{a}{\tau}{\prio'}}{\ctx}{\kwtid{a}}
        {\kwat{\tau}{\prio'}}
}

\\[4ex]

\infer[\kwnat I]
{
\strut
}
{
\etyped{\sig}{\ctx}{\kwn}{\kwnat}
}

\qquad

\infer[\kwnat E]
{
  \etyped{\sig}{\ctx}{v}{\kwnat}\\
  \etyped{\sig}{\ctx}{e_1}{\tau}\\
  \etyped{\sig}{\ctx,\hastype{x}{\kwnat}}{e_2}{\tau}
}
{
  \etyped{\sig}{\ctx}{\kwifz{v}{e_1}{x}{e_2}}{\tau}
}

\\[4ex]

\infer[\arrsym I]
{
  \etyped{\sig}
        {\ctx, \hastype{x}{\tau_1}}{e}{\tau_2}
}
{
  \etyped{\sig}{\ctx}{\kwfun{x}{e}}{\kwarr{\tau_1}{\tau_2}}
}

\qquad

\infer[\arrsym E]
{
\etyped{\sig}{\ctx}{v_1}{\kwarr{\tau_1}{\tau_2}}\\
\etyped{\sig}{\ctx}{v_2}{\tau_1}
}
{
\etyped{\sig}{\ctx}{\kwapply{v_1}{v_2}}{\tau_2}
}

\\[4ex]

\infer[\prodsym $I_1$]
{
  \etyped{\sig}{\ctx}{v_1}{\tau_1}\\
  \etyped{\sig}{\ctx}{v_2}{\tau_2}
}
{
  \etyped{\sig}{\ctx}{\kwepair{v_1}{v_2}}{\kwprod{\tau_1}{\tau_2}}
}

\qquad

\ifproofs
\infer[\prodsym $I_2$]
{
  \etyped{\sig}{\ctx}{v_1}{\tau_1}\\
  \etyped{\sig}{\ctx}{v_2}{\tau_2}
}
{
  \etyped{\sig}{\ctx}{\kwpair{v_1}{v_2}}{\kwprod{\tau_1}{\tau_2}}
}

\\[4ex]
\fi

\infer[\prodsym $E_1$]
{
  \etyped{\sig}{\ctx}{v}{\kwprod{\tau_1}{\tau_2}}
}
{
  \etyped{\sig}{\ctx}{\kwfst{v}}{\tau_1}
}

\qquad

\ifproofs
\infer[\prodsym $E_2$]
{
  \etyped{\sig}{\ctx}{v}{\kwprod{\tau_1}{\tau_2}}
}
{
  \etyped{\sig}{\ctx}{\kwsnd{v}}{\tau_2}
}

\qquad
\fi

\infer[\sumsym $I_1$]
{
  \etyped{\sig}{\ctx}{v}{\tau_1}
}
{
  \etyped{\sig}{\ctx}{\kweinl{v}}{\kwsum{\tau_1}{\tau_2}}
}

\\[4ex]

\ifproofs
\infer[\sumsym $I_2$]
{
  \etyped{\sig}{\ctx}{v}{\tau_2}
}
{
  \etyped{\sig}{\ctx}{\kweinr{v}}{\kwsum{\tau_1}{\tau_2}}
}

\qquad

\infer[\sumsym $I_3$]
{
  \etyped{\sig}{\ctx}{v}{\tau_1}
}
{
  \etyped{\sig}{\ctx}{\kwinl{v}}{\kwsum{\tau_1}{\tau_2}}
}

\qquad

\infer[\sumsym $I_4$]
{
  \etyped{\sig}{\ctx}{v}{\tau_2}
}
{
  \etyped{\sig}{\ctx}{\kwinr{v}}{\kwsum{\tau_1}{\tau_2}}
}

\\[4ex]
\fi

\infer[\sumsym E]
{
  \etyped{\sig}{\ctx}{v}{\kwsum{\tau_1}{\tau_2}}\\
  \etyped{\sig}{\ctx,\hastype{x}{\tau_1}}{e_1}{\tau'}\\
  \etyped{\sig}{\ctx,\hastype{y}{\tau_2}}{e_2}{\tau'}
}
{
  \etyped{\sig}{\ctx}{\kwcase{v}{x}{e_1}{y}{e_2}}{\tau'}
}

\\[4ex]

\infer[Output]
{
  \etyped{\sig}{\ctx}{v}{\kwnat}
}
{
  \etyped{\sig}{\ctx}{\kwoutput{v}}{\kwunit}
}

\qquad

\infer[Input]
{
  \strut
}
{
  \etyped{\sig}{\ctx}{\kwinput}{\kwnat}
}

\qquad

\infer[\cmdsym I]
{
  \cmdtyped{\sig}{\ctx}{\cmd}{\tau}{\prio}
}
{
  \etyped{\sig}{\ctx}{\kwcmd{\prio}{\cmd}}{\kwcmdt{\tau}{\prio}}
}

\\[4ex]

\infer[\fasym I]
{
  \etyped{\sig}{\ctx, \vprio \isprio, \cons}{e}{\tau}
}
{
  \etyped{\sig}{\ctx}{\kwwlam{\vprio}{\cons}{e}}{\kwforall{\vprio}{\cons}{\tau}}
}

\quad

\infer[\fasym E]
{
  \etyped{\sig}{\ctx}{v}{\kwforall{\vprio}{\cons}{\tau}}\\
  \meetc{\ctx}{[\prio'/\vprio]\cons}
}
{
  \etyped{\sig}{\ctx}{\kwwapp{v}{\prio'}}{[\prio'/\vprio]\tau}
}

\\[4ex]

\infer[fix]
{
  \etyped{\sig}{\ctx, \hastype{x}{\tau}}{e}{\tau}
}
{
  \etyped{\sig}{\ctx}{\kwfix{x}{\tau}{e}}{\tau}
}

\qquad

\infer[let]
{
  \etyped{\sig}{\ctx}{e_1}{\tau_1}\\
  \etyped{\sig}{\ctx, \hastype{x}{\tau_1}}{e_2}{\tau_2}
}
{
  \etyped{\sig}{\ctx}{\kwlet{x}{e_1}{e_2}}{\tau_2}
}

%% \infer[\boxsym I]
%% {
%%   \etyped{\sig}{\ctx,\vprio \isprio}{e}{\kwcmdt{\tau}{\vprio}}
%% }
%% {
%%   \etyped{\sig}{\ctx}{\kwbox{\vprio}{e}}{\kwboxt{\tau}}
%% }

%% \qquad

%% \infer[\diasym I]
%% {
%%   \etyped{\sig}{\ctx}{e}{\kwcmdt{\tau}{\prio}}
%% }
%% {
%%   \etyped{\sig}{\ctx}{\kwthere{\prio}{e}}{\kwdia{\tau}}
%% }
\end{array}
\]
\caption{\ifproofs Expression typing rules.\else Selected expression typing rules.\fi}
\label{fig:statics}
\end{figure}

\begin{figure}
\[
\begin{array}{c}
\infer[Bind]
{
  \etyped{\sig}{\ctx}{e}{\kwcmdt{\tau}{\prio}}\\
  \cmdtyped{\sig}{\ctx, \hastype{x}{\tau}}{\cmd}{\tau'}{\prio}
}
{
  \cmdtyped{\sig}{\ctx}{\kwbind{e}{x}{\cmd}}{\tau'}{\prio}
}

\quad

\infer[Spawn]
{
  \cmdtyped{\sig}{\ctx}{\cmd}{\tau}{\prio'}
}
{
  \cmdtyped{\sig}{\ctx}{\kwbspawn{\prio'}{\tau}{\cmd}}{\kwat{\tau}{\prio'}}{\prio}
}

\\[4ex]

\infer[Sync]
{
  \etyped{\sig}{\ctx}{e}{\kwat{\tau}{\prio'}}\\
  \meetc[\worlds]{\ctx}{\ple{\prio}{\prio'}}
}
{
  \cmdtyped{\sig}{\ctx}{\kwbsync{e}}{\tau}{\prio}
}

\quad

%% \infer[Poll]
%% {
%%   \etyped{\sig}{\ctx}{e}{\kwat{\tau}{\prio'}}
%% }
%% {
%%   \cmdtyped{\sig}{\ctx}{\kwpoll{e}}{\tau}{\prio}
%% }

%% \quad

%% \infer[\boxsym E]
%% {
%%   \etyped{\sig}{\ctx}{e}{\kwboxt{\tau}}\\
%%   \cmdtyped{\sig}{\ctx, \hastype{x}{\kwcmdt{\tau}{\prio}}}{\cmd}{\tau'}{\prio}
%% }
%% {
%%   \cmdtyped{\sig}{\ctx}{\kwletbox{x}{e}{\cmd}}{\tau'}{\prio}
%% }

%% \qquad

%% \infer[\diasym E]
%% {
%%   \etyped{\sig}{\ctx}{e}{\kwdia{\tau}}\\
%%   \cmdtyped{\sig}{\ctx, \vprio \isprio, \hastype{x}{\kwcmdt{\tau}{\vprio}}}
%%            {\cmd}{\tau'}{\prio}
%% }
%% {
%%   \cmdtyped{\sig}{\ctx}{\kwletdia{\vprio}{x}{e}{\cmd}}{\tau'}{\prio}
%% }

\qquad

\infer[Ret]
{
  \etyped{\sig}{\ctx}{e}{\tau}
}
{
  \cmdtyped{\sig}{\ctx}{\kwret{e}}{\tau}{\prio}
}

%% \qquad

%% \infer[\cmdsym E]
%% {
%%   \etyped{\sig}{\ctx}{e}{\kwcmdt{\tau}{\prio}}
%% }
%% {
%%   \cmdtyped{\sig}{\ctx}{\kwdo{e}}{\tau}{\prio}
%% }

%% \qquad

%% \infer[Post]
%% {
%%   \etyped{\sig}{\ctx}{e}{\tau}\\
%%   \cmdtyped{\sig}{\ctx}{\cmd}{\tau}{\prio}
%% }
%% {
%%   \cmdtyped{\sig}{\ctx}{\kwpost{e}{\cmd}}{\tau}{\prio}
%% }

%\qquad
%
%\infer[letm]
%{
%  \etyped{\sig}{\ctx}{e_1}{\tau_1}{\prio}\\
%  \tau_1 \mobile\\
%  \etyped{\sig,\hastype{x}{\tau_1}{\prio}}{\ctx}{e_2}{\tau_2}{\prio}
%}
%{
%  \etyped{\sig}{\ctx}{\kwletm{x}{e_1}{e_2}}{\tau_2}{\prio}
%}

%% \\[4ex]

%% \infer
%% {
%% }
%% {
%%   \mtyped{(\worlds, \prel)}{\esig}{\emem}
%% }

%% \qquad

%% \infer
%% {
%%   \cmdtyped{\sig}{\ctxify{(\worlds, \prel)}}{\cmd}{\tau}{\prio}\\
%%   \mtyped{(\worlds, \prel)}{\sig}{\mem}
%% }
%% {
%%   \mtyped{(\worlds, \prel)}{\sig, \sigtype{a}{\tau}{\prio}}
%%          {\dthread{a}{\cmd} \mcp \mem}
%% }

\end{array}
\]

\caption{Command typing rules.}
\label{fig:thread-statics}
\end{figure}

\begin{figure}
\[
\begin{array}{c}
\infer[hyp]
{
  \strut
}
{
  \meetc{\ctx, \ple{\prio_1}{\prio_2}}{\ple{\prio_1}{\prio_2}}
}

\quad

\infer[assume]
{
  \plt{\prioc_1}{\prioc_2} \in \worlds
}
{
  \meetc[\worlds]{\ctx}{\ple{\prioc_1}{\prioc_2}}
}

\quad

\infer[refl]
{
  \strut
}
{
  \meetc{\ctx}{\ple{\prio}{\prio}}
}

\quad

\infer[trans]
{
  \meetc{\ctx}{\ple{\prio_1}{\prio_2}}\\\\
  \meetc{\ctx}{\ple{\prio_2}{\prio_3}}
}
{
  \meetc{\ctx}{\ple{\prio_1}{\prio_3}}
}

\quad

\infer[conj]
{
  \meetc{\ctx}{\cons_1}\\\\
  \meetc{\ctx}{\cons_2}
}
{
  \meetc{\ctx}{\cconj{\cons_1}{\cons_2}}
}

\end{array}
\]

\caption{Constraint entailment}
\label{fig:constraint-statics}
\end{figure}

As with the syntax, the static semantics are separated into the expression layer
and the command layer.
Because expressions do not depend on priorities, the static semantics for
expressions is fairly standard. The main unusual feature is that the typing
judgment is parametrized by a signature~$\sig$ containing the types and
priorities of running threads. A signature has entries of the
form~$\sigtype{a}{\tau}{\prio}$ indicating that thread~$a$ is running at
priority~$\prio$ and will return a value of type~$\tau$. The signature is
needed to check the types of thread handles.

The expression typing judgment is~$\etyped{\sig}{\ctx}{e}{\tau}$, indicating
that under signature~$\sig$, a partial order~$\worlds$ of priority constants
and context~$\ctx$, expression~$e$ has type~$\tau$.
As usual, the variable context~$\ctx$ maps variables to their types.
\ifproofs The rules
\else Most rules
\fi
for this judgment are shown in Figure~\ref{fig:statics}
\ifproofs\else (omitted rules are similar to others% and are available
%in the technical report~\citep{partial-prio-tr}
)\fi. The
variable rule~\rulename{var}, the rule for fixed points
and the introduction and elimination rules for
unit, natural numbers, functions, products and sums, are straightforward.
The rule for thread handles~$\kwtid{a}$ looks up the thread~$a$ in the
signature. The rule for encapsulated commands~$\kwcmd{\prio}{\cmd}$ requires that
the command~$\cmd$ be well-typed and runnable at priority~$\prio$, using the
typing judgment for commands, which will be defined below.
Rule~\rulename{\fasym I} extends the
context with both the priority variable~$\vprio$ and the constraint~$\cons$.
Rule~\rulename{\fasym E} handles priority instantiation.
When instantiating the variable~$\vprio$ with priority~$\prio'$, the
rule requires that the constraints hold with~$\prio'$ substituted for~$\vprio$
(the constraint entailment judgment~$\meetc{\ctx}{\cons}$ will be discussed
below).
The rule also performs the corresponding substitution in the return type.

The command typing judgment is~$\cmdtyped{\sig}{\ctx}{\cmd}{\tau}{\prio}$ and
includes both the return type~$\tau$ and the priority~$\prio$ at which~$\cmd$
is runnable. The rules are shown in Figure~\ref{fig:thread-statics}.
The rule for bind requires that~$e$ return a command of the
current priority and return type~$\tau$, and then extends the context with a
variable~$x$ of type~$\tau$ in order to type the remaining command.
The rule for~$\kwbspawn{\prio'}{\tau}{\cmd}$ requires that~$\cmd$ be runnable
at priority~$\prio'$ and return a value of type~$\tau$. The spawn command
returns a thread handle of type~$\kwat{\tau}{\prio'}$, and may do so at any
priority. The~$\kwbsync{e}$ command requires that~$e$ have the type of a thread
handle of type~$\tau$, and returns a value of type~$\tau$. The rule also
checks the priority annotation on the thread's type and requires that this
priority be at least the current priority. This is the condition that
rules out~\cd{sync} commands that would cause priority inversions.
Finally, if~$e$ has type~$\tau$, then the command~$\kwret{e}$ returns a value
of type~$\tau$, at any priority.

The constraint checking judgment is defined in
Figure~\ref{fig:constraint-statics}.
We can conclude that a constraint
holds if it appears directly in the context (rule~\rulename{hyp}) or the
partial order (rule~\rulename{assume})
or if it can be concluded from reflexivity or transitivity
(rules~\rulename{refl} and \rulename{trans}, respectively).
Finally, the conjunction~$\cconj{\cons_1}{\cons_2}$ requires that both conjuncts
hold.

We use several forms of substitution in both the static and dynamic semantics.
All use the standard definition of capture-avoiding substitution. We can
substitute expressions for variables in expressions ($[e_2/x]e_1$) or in
commands ($[e/x]\cmd$), and we can substitute priorities for priority variables
in expressions ($[\prio/\vprio]e$), commands ($[\prio/\vprio]\cmd$),
constraints ($[\prio/\vprio]\cons$), contexts ($[\prio/\vprio]\ctx$), types
and priorities. For each of these substitutions, we prove
the principle that substitution preserves typing. These substitution principles
are collected in Lemma~\ref{lem:subst}.

\begin{lem}[Substitution]\label{lem:subst}\strut\\[-1em]
  \begin{enumerate}
\item If $\etyped{\sig}{\ctx, \hastype{x}{\tau}}{e_1}{\tau'}$ and
  $\etyped{\sig}{\ctx}{e_2}{\tau}$, then
  $\etyped{\sig}{\ctx}{[e_2/x]e_1}{\tau'}$.
\item If $\cmdtyped{\sig}{\ctx, \hastype{x}{\tau}}{\cmd}{\tau'}{\prio}$ and
  $\etyped{\sig}{\ctx}{e}{\tau}$, then
  $\cmdtyped{\sig}{\ctx}{[e/x]\cmd}{\tau'}{\prio}$.
\item If $\etyped{\sig}{\ctx, \vprio \isprio}{e}{\tau}$,
  then $\etyped{\sig}{[\prio/\vprio]\ctx}{[\prio/\vprio]e}{[\prio/\vprio]\tau}$.
\item If $\cmdtyped{\sig}{\ctx, \vprio \isprio}{\cmd}{\tau}{\prio}$,
  then $\cmdtyped{\sig}{[\prio'/\vprio]\ctx}{[\prio'/\vprio]\cmd}
  {[\prio'/\vprio]\tau}{[\prio'/\vprio]\prio}$.
\item If $\meetc{\ctx, \vprio \isprio}{\cons}$, then
  $\meetc{[\prio/\vprio]\ctx}{[\prio/\vprio]\cons}$.
\end{enumerate}
\end{lem}
\ifproofs
\begin{proof}
\begin{enumerate}
\item By induction on the derivation of
  $\etyped{\sig}{\ctx, \hastype{x}{\tau}}{e_1}{\tau'}$. Consider one
  representative case.
  \begin{itemize}
  \item \rulename{\fasym E} Then $e = \kwwapp{e_0}{\prio'}$.
    By inversion,
    $\etyped{\sig}{\ctx, \hastype{x}{\tau}}{e_0}{\kwforall{\vprio}{\cons}{\tau_0}}$
    and $\tau' = [\prio'/\vprio]\tau_0$.
    By induction, $\etyped{\sig}{\ctx}{[e/x]e_0}{\kwforall{\vprio}{\cons}{\tau_0}}$.
    Apply \rulename{\fasym E}.
  \end{itemize}
\item By induction on the derivation of
  $\cmdtyped{\sig}{\ctx, \hastype{x}{\tau}}{\cmd}{\tau'}{\prio}$.
  \begin{itemize}
  \item \rulename{Bind}. Then $\cmd = \kwbind{e_0}{y}{\cmd_0}$. By inversion,
    $\etyped{\sig}{\ctx, \hastype{x}{\tau}}{e_0}{\kwcmdt{\tau''}{\prio}}$
    and\\
    $\cmdtyped{\sig}{\ctx, \hastype{x}{\tau}, \hastype{y}{\tau''}}{\cmd_0}
    {\tau'}{\prio}$.
    By weakening, $\etyped{\sig}{\ctx, \hastype{y}{\tau''}}{e}{\tau}$.
    By induction, $\etyped{\sig}{\ctx}{[e/x]e_0}{\kwcmdt{\tau''}{\prio}}$ and
    $\cmdtyped{\sig}{\ctx, \hastype{y}{\tau''}}{[e/x]\cmd_0}{\tau'}{\prio}$.
    Apply \rulename{Bind}.
  \item \rulename{Spawn}. Then $\cmd = \kwbspawn{\prio'}{\tau''}{\cmd_0}$.
    By inversion,
    $\cmdtyped{\sig}{\ctx, \hastype{x}{\tau}}{\cmd_0}{\tau''}{\prio'}$.
    By induction, $\cmdtyped{\sig}{\ctx}{[e/x]\cmd_0}{\tau''}{\prio'}$.
    Apply \rulename{Spawn}.
  \item \rulename{Sync}. Then $\cmd = \kwbsync{e_0}$. By inversion,
    $\etyped{\sig}{\ctx, \hastype{x}{\tau}}{e_0}{\kwat{\tau'}{\prio'}}$.
    By induction, $\etyped{\sig}{\ctx}{[e/x]e_0}{\kwat{\tau'}{\prio'}}$.
    Apply \rulename{Sync}.

  %% \item \rulename{\diasym E}. Then $\cmd = \kwletdia{\vprio}{y}{e_0}{\cmd_0}$.
  %%   By inversion, $\typed{\sig}{\ctx, \hastype{x}{\tau}}{e_0}{\kwdia{\tau''}}$ and
  %%   $\cmdtyped{\sig}{\ctx, \hastype{x}{\tau}, \vprio \isprio,
  %%     \hastype{y}{\kwcmdt{\tau''}{\vprio}}}{\cmd_0}{\tau'}{\prio}$.
  %%   By weakening,
  %%   $\typed{\sig}{\ctx, \vprio \isprio, \hastype{y}{\kwcmdt{\tau''}{\vprio}}}
  %%   {e}{\tau}$.
  %%   By induction, $\typed{\sig}{\ctx}{[e/x]e_0}{\kwdia{\tau''}}$ and
  %%   $\cmdtyped{\sig}{\ctx, \vprio \isprio, \hastype{y}{\kwcmdt{\tau''}{\vprio}}}
  %%   {[e/x]\cmd_0}{\tau'}{\prio}$. Apply \rulename{\diasym E}.
  \item \rulename{Ret}. Then $\cmd = \kwret{e_0}$. By inversion,
    $\etyped{\sig}{\ctx, \hastype{x}{\tau}}{e_0}{\tau'}$. By induction,
    $\etyped{\sig}{\ctx}{[e/x]e_0}{\tau'}$. Apply \rulename{Ret}.
  \end{itemize}
\item By induction on the derivation of
  $\etyped{\sig}{\ctx, \vprio \isprio}{e}{\tau}$.
  \begin{itemize}
  \item \rulename{\fasym E} Then $e = \kwwapp{e_0}{\prio''}$
    and $[\prio'/\vprio]e = \kwwapp{[\prio'/\vprio]e_0}{[\prio'/\vprio]\prio''}$.
    By inversion,\\
    $\etyped{\sig}{\ctx, \vprio \isprio}{e_0}{\kwforall{\vprio'}{\cons}{\tau'}}$
    and $\tau = [\prio''/\vprio']\tau'$ and
    $\meetc{\ctx, \vprio \isprio}{[\prio''/\vprio']\cons}$.\\
    By induction, $\etyped{\sig}{[\prio'/\vprio]\ctx}{[\prio'/\vprio]e_0}
    {[\prio'/\vprio](\kwforall{\vprio'}{\cons}{\tau'}) =
    \kwforall{\vprio'}{[\prio'/\vprio]\cons}{[\prio'/\vprio]\tau'}}$ and
    $\meetc{\ctx}{[\prio'/\vprio][\prio''/\vprio']\cons
      = [[\prio'/\vprio]\prio''/\vprio'][\prio'/\vprio]\cons}$.\\
    By \rulename{\fasym E},
    $\cmdtyped{\sig}{\ctx}{\kwwapp{e_0}{\prio''}}
    {[[\prio'/\vprio]\prio''/\vprio'][\prio'/\vprio]\tau'}{\prio}$.\\
    Because $[[\prio'/\vprio]\prio''/\vprio'][\prio'/\vprio]\tau' =
    [\prio'/\vprio][\prio''/\vprio]\tau' = [\prio'/\vprio]\tau$,
    this completes the case.
  \end{itemize}
\item By induction on the derivation of
  $\cmdtyped{\sig}{\ctx, \vprio \isprio}{\cmd}{\tau}{\prio}.$
  \begin{itemize}
  \item \rulename{Bind}. Then $\cmd = \kwbind{e}{x}{\cmd_0}$. By inversion,
    $\etyped{\sig}{\ctx, \vprio \isprio}{e}{\kwcmdt{\tau''}{\prio}}$
    and
    $\cmdtyped{\sig}{\ctx, \vprio \isprio, \hastype{x}{\tau''}}{\cmd_0}
    {\tau'}{\prio}$.
    By induction,
    $\etyped{\sig}{[\prio'/\vprio]\ctx}{[\prio'/\vprio]e}
    {[\prio'/\vprio](\kwcmdt{\tau''}{\prio})}$
    and
    $\cmdtyped{\sig}{[\prio'/\vprio]\ctx, \hastype{x}{[\prio'/\vprio]\tau''}}
    {[\prio'/\vprio]\cmd_0}
    {[\prio'/\vprio]\tau'}{[\prio'/\vprio]\prio}$.
    Apply \rulename{Bind}.
  \item \rulename{Spawn}. Then $\cmd = \kwbspawn{\prio''}{\tau'}{\cmd_0}$\\
    and $[\prio'/\vprio]\cmd = \kwbspawn{[\prio'/\vprio]\prio''}
    {[\prio'/\vprio]\tau'}{[\prio'/\vprio]\cmd_0}$.
    By inversion,\\
    $\cmdtyped{\sig}{\ctx, \vprio \isprio}{\cmd_0}{\tau'}{\prio''}$.
    By induction,
    $\cmdtyped{\sig}{[\prio'/\vprio]\ctx}{[\prio'/\vprio]\cmd_0}
    {[\prio'/\vprio]\tau'}
    {[\prio'/\vprio]\prio''}$.
    By \rulename{Spawn},
    $\cmdtyped{\sig}{[\prio'/\vprio]\ctx}{[\prio'/\vprio]\cmd}
    {[\prio'/\vprio](\kwat{\tau'}{\prio''})}{[\prio'/\vprio]\prio}.$
  \item \rulename{Sync}. Then $\cmd = \kwbsync{e}$. By inversion,
    $\etyped{\sig}{\ctx, \vprio \isprio}{e}{\kwat{\tau}{\prio''}}$.
    By induction,
    $\etyped{\sig}{[\prio'/\vprio]\ctx}{[\prio'/\vprio]e}
    {[\prio'/\vprio](\kwat{\tau'}{\prio''})}$.
    Apply \rulename{Sync}.

  %% \item \rulename{\diasym E}. Then $\cmd = \kwletdia{\vprio'}{x}{e_0}{\cmd_0}$.
  %%   By inversion, $\typed{\sig}{\ctx, \vprio \isprio}{e_0}{\kwdia{\tau'}}$ and
  %%   $\cmdtyped{\sig}{\ctx, \hastype{x}{\tau}, \vprio \isprio, \vprio' \isprio,
  %%     \hastype{x}{\kwcmdt{\tau'}{\vprio'}}}{\cmd_0}{\tau}{\prio}$.
  %%   By induction, $\typed{\sig}{\ctx}{[\prio'/\vprio]e_0}{\kwdia{\tau''}}$ and
  %%   $\cmdtyped{\sig}{\ctx, \vprio' \isprio, \hastype{y}{\kwcmdt{\tau''}{\vprio}}}
  %%   {[\prio'/\vprio]\cmd_0}{\tau'}{\prio}$. Apply \rulename{\diasym E}.
  \item \rulename{Ret}. Then $\cmd = \kwret{e}$. By inversion,
    $\etyped{\sig}{\ctx, \vprio \isprio}{e}{\tau'}$.\\ By induction,
    $\etyped{\sig}{[\prio'/\vprio]\ctx}{[\prio'/\vprio]e}{[\prio'/\vprio]\tau'}$.
    Apply \rulename{Ret}.
  \end{itemize}
\item By induction on the derivation of $\meetc{\ctx, \vprio\isprio}{\cons}$.
  We consider the non-trivial cases.
  \begin{itemize}
%  \item \rulename{hyp}. Apply rule \rulename{hyp}.
%  \item \rulename{refl}. Apply rule \rulename{refl}.
  \item \rulename{trans}.
    By inversion, $\meetc{\ctx, \vprio \isprio}{\ple{\prio_1}{\prio_2}}$ and
    $\meetc{\ctx, \vprio \isprio}{\ple{\prio_2}{\prio_3}}$.
    By induction,
    $\meetc{[\prio/\vprio]\ctx}{\ple{[\prio/\vprio]\prio_1}{[\prio/\vprio]\prio_2}}$
    and
    $\meetc{[\prio/\vprio]\ctx}{\ple{[\prio/\vprio]\prio_2}{[\prio/\vprio]\prio_3}}$
    Apply rule \rulename{trans}.
  \item \rulename{conj}
    By inversion, $\meetc{\ctx, \vprio \isprio}{\cons_1}$ and
    $\meetc{\ctx, \vprio \isprio}{\cons_2}$. By induction,
    $\meetc{[\prio/\vprio]\ctx}{[\prio/\vprio]\cons_1}$ and
    $\meetc{[\prio/\vprio]\ctx}{[\prio/\vprio]\cons_2}$.
    Apply rule \rulename{conj}.
  \end{itemize}
\end{enumerate}
\end{proof}
\else
%See the technical report~\citep{partial-prio-tr} for the straightforward proof.
The proof is straightforward.
\fi

\subsection{Dynamic Semantics}\label{sec:dyn}
We define a transition semantics for {\calcname}.
%As before, we will separate the semantics into two components.
Because the operational behavior (as distinct from run-time or responsiveness,
which will be the focus of Section~\ref{sec:cost}) of expressions does not
depend on the
priority at which they run or what other threads are running,
their semantics can be defined without
regard to other running threads. The semantics for commands will be
more complex, because it must include other threads. We will also define a
syntax and dynamic semantics for {\em thread pools}, which are collections of all
of the currently running threads.

The dynamic semantics for expressions
consists of two judgments. The judgment~$v \val{\sig}$ states that~$v$
is a well-formed value and refers only to thread names in the signature~$\sig$.
The rules for this judgment are omitted.
%For clarity, we will often use metavariables~$v$ and variants
%to range over values, even though they are in the same syntactic class
%as expressions. The values are the unit value, thread handles
%contained in~$\sig$, numerals, functions, pairs of values, injections of
%values, encapsulated commands (because the command is encapsulated and not run,
%these are values regardless of the contained command) and
%priority abstractions.
The transition relation for expressions~$e \estep{\sig} e'$
is fairly straightforward for a left-to-right, call-by-value
lambda calculus and is shown in Figure~\ref{fig:exp-dyn}.
The signature~$\sig$ does not
change during expression evaluation and is
used solely to determine whether thread IDs are well-formed values.
%The~$\kw{ifz}$ and~$\kw{case}$ constructs evaluate the argument~$e$ to a value
%in order to examine it.
The~$\kw{ifz}$ construct conditions on the value of the
numeral~$\kwnumeral{n}$.
If~$n = 0$, it steps to~$e_1$. If not, it steps to~$e_2$,
substituting~$n-1$ for~$x$.
The~$\kw{case}$ construct conditions on whether~$e$
is a left or right injection, and steps to~$e_1$ (resp.~$e_2$),
substituting the injected value for~$x$ (resp.~$y$).
Function applications and priority instantiations simply perform the
appropriate substitution.
%Because the semantics is
%call-by-name, function applications~$\kwapply{e_1}{e_2}$ are not performed
%until~$e_1$ is evaluated to a lambda abstraction and~$e_2$ is evaluated to a
%value, which is then substituted for the bound variable.
%Pairs are evaluated left-to-right, with the
%first component being fully evaluated before the second, and both are evaluated
%eagerly before a projection is performed.
%Priority instantiation~$\kwwapp{e}{\prio'}$ evaluates~$e$ to
%an abstraction, and then performs the substitution.

\begin{figure}
%% \[
%% \begin{array}{c}

%% \infer
%% {
%%   \strut
%% }
%% {
%%   \kwtriv \val{\sig}
%% }

%% \qquad

%% \infer
%% {
%%   \strut
%% }
%% {
%%   \kwtid{a} \val{\sig, \sigtype{a}{\tau}{\prio}}
%% }

%% \qquad

%% \infer
%% {
%%   \strut
%% }
%% {
%%   \kwn \val{\sig}
%% }

%% \qquad

%% \infer
%% {
%%   \strut
%% }
%% {
%%   \kwfun{x}{e} \val{\sig}
%% }

%% \qquad

%% \infer
%% {
%%   v_1 \val{\sig}\\
%%   v_2 \val{\sig}
%% }
%% {
%%   \kwpair{v_1}{v_2} \val{\sig}
%% }

%% \\[4ex]

%% \infer
%% {
%%   v \val{\sig}
%% }
%% {
%%   \kwinl{v} \val{\sig}
%% }

%% \qquad

%% \infer
%% {
%%   v \val{\sig}
%% }
%% {
%%   \kwinr{v} \val{\sig}
%% }

%% \qquad

%% \infer
%% {
%%   \strut
%% }
%% {
%%   \kwcmd{\prio}{\cmd} \val{\sig}
%% }

%% \qquad

%% \infer
%% {
%%   \strut
%% }
%% {
%%   \kwwlam{\prio}{\cons}{e} \val{\sig}
%% }

%% %% \\[4ex]

%% %% \infer
%% %% {
%% %%   \strut
%% %% }
%% %% {
%% %%   \kwbox{\vprio}{\kwcmd{\vprio}{\cmd}} \val{\sig}
%% %% }

%% %% \qquad

%% %% \infer
%% %% {
%% %%   \strut
%% %% }
%% %% {
%% %%   \kwthere{\prio}{\kwcmd{\prio}{\cmd}} \val{\sig}
%% %% }

%% \end{array}
%% \]

\[
\begin{array}{c}
  \infer[D-Let-Step]
      {
        e_1 \estep{\sig} e_1'
      }
      {
        \kwlet{x}{e_1}{e_2} \estep{\sig} \kwlet{x}{e_1'}{e_2}
      }

\qquad

\infer[D-Let]
      {
        v \val{\sig}
      }
      {
        \kwlet{x}{v}{e} \estep{\sig} [v/x]e
      }

      \\[4ex]

\infer[D-Ifz-Z]
{
  \strut
}
{
  \kwifz{0}{e_1}{x}{e_2} \estep{\sig} e_1
}

\qquad

\infer[D-Ifz-NZ]
{
  \strut
}
{
  \kwifz{\kwnumeral{n+1}}{e_1}{x}{e_2} \estep{\sig} [\kwnumeral{n}/x]e_2
}

\qquad

\infer[D-App]
{
  v \val{\sig}
}
{
  \kwapply{(\kwfun{x}{e})}{v} \estep{\sig} [v/x]e
}

\\[4ex]

\infer[D-Pair]
{
  v_1 \val{\sig}\\
  v_2 \val{\sig}\\
}
{
  \kwepair{v_1}{v_2} \estep{\sig} \kwpair{v_1}{v_2}
}

\qquad

\infer[D-Fst]
{
  v_1 \val{\sig}\\
  v_2 \val{\sig}
}
{
  \kwfst{\kwpair{v_1}{v_2}} \estep{\sig} v_1
}

\quad

\infer[D-Snd]
{
  v_1 \val{\sig}\\
  v_2 \val{\sig}
}
{
  \kwsnd{\kwpair{v_1}{v_2}} \estep{\sig} v_2
}

\\[4ex]

\infer[D-InL]
{
  v \val{\sig}
}
{
  \kweinl{v} \estep{\sig} \kwinl{v}
}

\qquad

\infer[D-InR]
{
  v \val{\sig}
}
{
  \kweinr{v} \estep{\sig} \kwinr{v}
}

\qquad

\infer[D-Case-L]
{
  v \val{\sig}
}
{
  \kwcase{\kwinl{v}}{x}{e_1}{y}{e_2} \estep{\sig} [v/x]e_1
}

\\[4ex]

\infer[D-Case-R]
{
  v \val{\sig}
}
{
  \kwcase{\kwinr{v}}{x}{e_1}{y}{e_2} \estep{\sig} [v/y]e_2
}

\qquad

\infer[D-Output]
{
  \strut
}
{
  \kwoutput{\kwn} \estep{\sig} \kwtriv
}

\qquad

\infer[D-Input]
{
 n \in \mathbb{N}
}
{
  \kwinput \estep{\sig} \kwnumeral{n}
}

\\[4ex]

\infer[D-PrApp]
{
  \strut
}
{
  \kwwapp{(\kwwlam{\vprio}{\cons}{e})}{\prio'} \estep{\sig}
  [\prio'/\vprio]e
}

\qquad

\infer[D-Fix]
{
  \strut
}
{
  \kwfix{x}{\tau}{e} \estep{\sig} [\kwfix{x}{\tau}{e}/x] e
}

%% \qquad

%% \infer
%% {
%%   e \estep{\sig} e'
%% }
%% {
%%   \kwwlam{\vprio}{\cons}{e} \estep{\sig} \kwwlam{\vprio}{\cons}{e'}
%% }

%% \qquad

%% \infer
%% {
%%   e \estep{\sig} e'
%% }
%% {
%%   \kwthere{\prio}{e} \estep{\sig} \kwthere{\prio}{e'}
%% }

\end{array}
\]
\caption{Dynamic semantics for expressions.}
\label{fig:exp-dyn}
\end{figure}

\begin{figure}
\[
\begin{array}{c}

\infer[Empty]
{
}
{
  \mtyped{\worlds}{\sig}{\emem}{\esig}
}

\qquad

\infer[OneThread]
{
  \cmdtyped[\worlds]{\sig}{\cdot}{\cmd}
           {\tau}{\prio}
}
{
  \mtyped{\worlds}{\sig}{\dthread{a}{\prio}{\cmd}}{\sigtype{a}{\tau}{\prio}}
}

\qquad

\infer[Concat]
{
  \mtyped{\worlds}{\sig,\sig_2}{\mem_1}{\sig_1}\\
  \mtyped{\worlds}{\sig,\sig_1}{\mem_2}{\sig_2}
}
{
  \mtyped{\worlds}{\sig}{\mem_1 \mcp \mem_2}{\sig_1,\sig_2}
}

\qquad

\infer[Extend]
{
  \mtyped{\worlds}{\sig}{\mem}{\sig',\sig''}
}
{
  \mtyped{\worlds}{\sig}{\mbconfig{\sig'}{\mem}}{\sig''}
}

\end{array}
\]
\caption{Typing rules for thread pools}
\label{fig:statics-tp}
\end{figure}

\begin{figure}
\[
\begin{array}{c}

\infer
    {
}
{
  \mbconfig{\sig}{\mem_1} \mcp \mem_2 \meq
  \mbconfig{\sig}{\mem_1 \mcp \mem_2}
}

\qquad

\infer
    {
}
{
  \mbconfig{\sig}{\mbconfig{\sig'}{\mem}} \meq \mbconfig{\sig,\sig'}{\mem}
}

\qquad

\infer
    {
}
{
  \mbconfig{\esig}{\mem} \meq \mem
}

\end{array}
\]
\caption{Congruence rules for thread pools.}
\label{fig:tp-cong}
\end{figure}

Define a thread pool~$\mem$ to be a mapping of thread symbols to
threads:~$\dthread{a}{\prio}{\cmd}$ indicates a thread~$a$ at priority~$\prio$
running~$\cmd$. The concatenation of two thread pools is
written~$\mem_1 \mcp \mem_2$. Thread pools can also introduce new thread names:
the thread pool~$\mbconfig{\sig}{\mem}$ allows the thread pool~$\mem$ to use
thread names bound in the signature~$\sig$. Thread pools are not ordered;
we identify thread pools up to commutativity and associativity
of~$\mcp$\footnote{Because threads cannot refer to threads that
  (transitively) spawned them, we could order the thread pool, which would allow
  us to prove that deadlock is not possible in {\calcname}. This is outside the
  scope of this paper.}.
We also introduce the additional congruence rules of
Figure~\ref{fig:tp-cong}, which allow for thread name bindings to freely change
scope within a thread pool.
%We will use this congruence relation to establish
%a normal form for thread pools, in which all thread names have been moved
%to the outside: for every thread pool~$\mem$, there exists
%a thread pool~$\mbconfig{\sig}{\mem'} \meq \mem$, where~$\mem'$ has no nested
%thread name bindings.

Figure~\ref{fig:statics-tp} gives the typing rules for thread pools.
The typing judgment~$\mtyped{\worlds}{\sig}{\mem}{\sig'}$
indicates that all threads
of~$\mem$ are well-typed assuming an ambient environment that includes the
threads mentioned in~$\sig$, and that~$\sig'$ includes the threads introduced
in~$\mem$, minus any bound in a~$\mbconfig{\sig''}{\mem''}$ form.
The rules are straightforward:
the empty thread pool~$\emem$ is always well-typed and introduces no threads,
individual threads are well-typed if their commands are,
and concatenations are well-typed if their components are.
In a concatenation~$\mem_1 \mcp \mem_2$, if~$\mem_1$ introduces the
threads~$\sig_1$ and~$\mem_2$ introduces the threads~$\sig_2$, then~$\mem_1$
may refer to threads in~$\sig_2$ and vice versa.
If a thread pool~$\mem$ is well-typed and introduces the threads
in~$\sig',\sig''$, then~$\mbconfig{\sig'}{\mem}$ introduces the threads
in~$\sig''$ (subtracting off the threads explicitly introduced by the binding).

\begin{figure}
\[
\begin{array}{c}

\infer[D-Bind1]
{
  e \estep{\sig} e'
}
{
  \kwbind{e}{x}{\cmd}
  \lstep{\sig}{\asil}
  \rpconfig{\esig}{\kwbind{e'}{x}{\cmd}}{\emem}
}

\qquad

\infer[D-Bind2]
{
  \cmd_1 \lstep{\sig}{\act} \rpconfig{\sig'}{\cmd_1'}{\mem'}
}
{
  \kwbind{\kwcmd{\prio}{\cmd_1}}{x}{\cmd_2}
  \lstep{\sig}{\act}
  \rpconfig{\sig'}{\kwbind{\kwcmd{\prio}{\cmd_1'}}{x}{\cmd_2}}{\mem'}
}

\\[4ex]

\infer[D-Bind3]
{
  e \val{\sig}
}
{
  \kwbind{\kwcmd{\prio}{\kwret{e}}}{x}{\cmd}
  \lstep{\sig}{\asil}
  \rpconfig{\esig}{[e/x]\cmd}{\emem}
}

\qquad

\infer[D-Spawn]
{
  b \fresh
}
{
  \kwbspawn{\prio}{\tau}{\cmd}
  \lstep{\sig}{\asil}
  \rpconfig{\sigtype{b}{\tau}{\prio}}{\kwret{\kwtid{b}}}{\dthread{b}{\prio}{\cmd}}
}

\\[4ex]

\infer[D-Sync1]
{
  e \estep{\sig} e'
}
{
  \kwbsync{e}
  \lstep{\sig}{\asil}
  \rpconfig{\esig}{\kwbsync{e'}}{\emem}
}

\qquad

\infer[D-Sync2]
{
  v \val{\sig}%\\
  %\ple{\prio}{\prio'}
}
{
  \kwbsync{(\kwtid{b})}
  \lstep[\prio]{\sig}{\async{b}{v}}
  \rpconfig{\esig}{\kwret{v}}{\emem}
}

%% \\[4ex]

%% \infer
%% {
%%   e \estep{\sig} e'
%% }
%% {
%%   \lconfig{\sig}{\dthread{a}{\kwletdia{x}{\vprio}{e}{\cmd}}}{\mem} \tstep
%%   \rconfig{\sig}{\dthread{a}{\kwletdia{x}{\vprio}{e'}{\cmd}}}
%% }

%% \\[4ex]

%% \infer
%% {
%%   \strut
%% }
%% {
%%   \lconfig{\sig}
%%           {\dthread{a}{\kwletdia{x}{\vprio}
%%               {\kwthere{\prio}{\kwcmd{\prio}{\cmd_1}}}{\cmd_2}}}
%%           {\mem}
%%   \tstep
%%   \rconfig{\sig}{\dthread{a}{[\kwcmd{\prio}{\cmd_1}/x][\prio/\vprio]\cmd_2}}
%% }

\qquad

\infer[D-Ret]
{
  e \estep{\sig} e'
}
{
  \kwret{e}
  \lstep{\sig}{\asil}
  \rpconfig{\esig}{\kwret{e'}}{\emem}
}

\end{array}
\]
\caption{Dynamic rules for commands.}
\label{fig:thread-dyn}

\end{figure}

\begin{figure}
\[
\begin{array}{c}

\infer[DT-Thread]
{
  \cmd \lstep[\prio]{\sig}{\act}
  \rpconfig{\sig'}{\cmd'}{\mem'}
}
{
  \dthread{a}{\prio}{\cmd}
  \gstep{\sigtype{a}{\tau}{\prio}, \sig}{\tact{a}{\act}}
  \mbconfig{\sig'}{\dthread{a}{\prio}{\cmd'} \mcp \mem'}
}

\qquad

\infer[DT-Ret]
{
  v \val{\sigtype{a}{\tau}{\prio}, \sig}
}
{
  \dthread{a}{\prio}{\kwret{v}}
  \gstep{\sigtype{a}{\tau}{\prio}, \sig}{\tact{a}{\asend{a}{v}}}
  \dthread{a}{\prio}{\kwret{v}}
}

\\[4ex]

\infer[DT-Sync]
{
  \sig = \sig', \sigtype{a}{\tau_a}{\prio_a}, \sigtype{b}{\tau_b}{\prio_b}\\
  %\ple{\prio_a}{\prio_b}\\
  \mem_1 \gstep{\sig}{\tact{a}{\async{b}{v}}} \mem_1'\\
  \mem_2 \gstep{\sig}{\tact{b}{\asend{b}{v}}} \mem_2
}
{
  \mem_1 \mcp \mem_2
  \gstep{\sig}{\tact{a}{\asil}}
  \mem_1' \mcp \mem_2
}

\\[4ex]

\infer[DT-Concat]
{
  \mem_1 \gstep{\sig}{\tact{a}{\act}} \mem_1'
}
{
  \mem_1 \mcp \mem_2
  \gstep{\sig}{\tact{a}{\act}}
  \mem_1' \mcp \mem_2
}

\qquad

%% \infer[DT-Concat-Both]
%% {
%%   \mem_1 \uplus \mem_2 \gstep{\sig}{\acts_1} \mem_1' \uplus \mem\\
%%   \mem_2 \uplus \mem \gstep{\sig}{\acts_2} \mem_2' \uplus \mem
%% }
%% {
%%   \mem_1 \mcp \mem_2' \mcp \mem
%%   \gstep{\sig}{\acts_1 \cup \acts_2}
%%   \mem_1' \mcp \mem_2' \mcp \mem
%% }

%%\qquad

%% \infer
%% {
%%   \mem_2 \gstep{\sig}{\async{b}{v}} \mem_2'\\
%%   \mem_4 \gstep{\sig}{\asend{b}{v}} \mem_4'
%% }
%% {
%%   \mem_1 \mcp \mem_2 \mcp \mem_3 \mcp \mem_4 \mcp \mem_5
%%   \gstep{\sig}{\asil}
%%   \mem_1 \mcp \mem_2' \mcp \mem_3 \mcp \mem_4' \mcp \mem_5
%% }

%%\quad

\infer[DT-Extend]
{
  \mem \gstep{\sig, \sigtype{a}{\tau}{\prio}}{\tact{b}{\act}} \mem'
}
{
  \mbconfig{\sigtype{a}{\tau}{\prio}}{\mem}
  \gstep{\sig}{\tact{b}{\act}}
  \mbconfig{\sigtype{a}{\tau}{\prio}}{\mem'}
}

\\[4ex]

\infer[DT-Par]
      { (\forall 1 \leq i \leq n)
        \mbconfig{\sig}{\mem \mcp \mem_1 \mcp \dots \mcp \mem_n}
        \gstep{\esig}{\tact{a_i}{\asil}}
        \mbconfig{\sig}{\mem \mcp \mem_1 \mcp \dots \mcp \mem_i' \mcp \dots \mcp \mem_n}
      }
      {
        \mbconfig{\sig}{\mem \mcp \mem_1 \mcp \dots \mcp \mem_n}
        \pgstep{\{a_1, \dots, a_n\}}
        \mbconfig{\sig}{\mem \mcp \mem_1' \mcp \dots \mcp \mem_i' \mcp \dots \mcp \mem_n'}
      }

\end{array}
\]
\caption{Dynamic rules for thread pools.}
\label{fig:mem-dyn}
\end{figure}

The transition judgment for commands is
$\cmd \lstep{\sig}{\act} \rpconfig{\sig'}{\cmd'}{\mem'}$,
indicating that under signature~$\sig$, command~$\cmd$ steps
to~$\cmd'$. The transition relation carries a label~$\act$, indicating
the ``action'' taken by this step. At this point, actions can be the silent
action~$\asil$ or the sync action~$\async{b}{v}$, indicating that the transition
receives a value~$v$ by synchronizing on thread~$b$.
This step may also spawn new threads, and so the judgment
includes extensions to the thread pool~($\mem'$) and
the signature~($\sig'$). Both extensions may be empty.
%We do not define a judgment corresponding to the~$e \val{\sig}$ judgment for
%expressions, because the only well-formed irreducible command is~$\kwret{v}$
%where~$v$ is a value.

The rules for the transition judgment are shown in Figure~\ref{fig:thread-dyn}.
The rules for the bind construct~$\kwbind{e}{x}{\cmd_2}$ evaluate~$e$ to an
encapsulated command~$\kwcmd{\prio}{\cmd_1}$, then evaluate this command to
a return value~$\kwret{v}$ before substituting~$v$ for~$x$ in~$\cmd_2$.
The spawn command~$\kwbspawn{\prio}{\tau}{\cmd}$ does {\em not} evaluate~$\cmd$,
but simply spawns a fresh thread~$b$ to execute it, and returns a
thread handle~$\kwtid{b}$. The sync command~$\kwbsync{e}$ evaluates~$e$ to a
thread handle~$\kwtid{b}$, and then takes a step to~$\kwret{v}$ labeled
with the action~$\async{b}{v}$. Note that, because the thread~$b$ is not
available to the rule, the return value~$v$ is ``guessed''. It will be the
job of the thread pool semantics to connect this thread to the thread~$b$
and provide the appropriate return value.
Finally, $\kwret{e}$
evaluates~$e$ to a value.

We define an additional transition judgment for thread pools, which
nondeterministically allows a thread to step.
The judgment~$\mem \gstep{\sig}{\tact{a}{\act}} \mem'$ is again annotated with
an action.
In this judgment, because it is not clear what thread is taking
the step, the action is labeled with the thread~$a$.
Actions now also include the ``return'' action~$\asend{a}{v}$, indicating
that the thread returns the value~$v$.
Rule~\rulename{DT-Sync} matches this with a corresponding sync action and
performs the synchronization.
If a thread in~$\mem_1$ wishes to
sync with~$b$ and a thread~$b$ in~$\mem_2$ wishes to return its
value, then the thread pool~$\mem_1\mcp \mem_2$ can step
silently, performing the synchronization.
Without loss of generality,~$\mem_1$ can come first because thread pools are
identified up to ordering.
The last two rules allow threads to
step when concatenated with other threads and under bindings.

We will show as part of the type safety theorem that any thread pool may
be, through the congruence rules, placed in a normal
form~$\mbconfig{\sig}{\dthread{a_1}{\prio_1}{\cmd_1} \mcp \dots \mcp
  \dthread{a_n}{\prio_n}{\cmd_n}}$ and that
stepping one of these threads does not affect the rest of the thread pool
other than by spawning new threads.
This property, that transitions of separate threads do not impact each other,
is key to parallel functional programs and allows us to cleanly talk about
taking multiple steps of separate threads in parallel.
This is expressed by the judgment~$\mem \pgstep{\acts} \mem'$,
which allows all of the threads in the set~$\acts$ to step silently in
parallel.
The only rule for this judgment is~\rulename{DT-Par}, which
steps any number of threads in a nondeterministic fashion.
We do not impose any sort of scheduling algorithm in the semantics, nor even
a maximum number of threads.
When discussing cost bounds, we will quantify over executions which choose
threads in certain ways.

%We are now ready to state and prove the progress theorem for our semantics.
We prove a version of the standard progress theorem for each syntactic class.
Progress for expressions is standard: a well-typed
expression is either a value or can take a step. The progress statement for
commands is similar, because commands can step (with a sync action) even if
they are waiting for other threads. The statement for
thread pools is somewhat counter-intuitive. One might expect it
to state that if a thread pool is well-typed, then either all threads are
complete or the thread pool can take a step. This statement is true but too
weak to be useful; because of the non-determinism in our semantics, such
a theorem would allow
for one thread to enter a ``stuck'' state as
long as any other thread is still able to make progress (for example, if it is
in an infinite loop).
Instead, we state that, in a well-typed thread pool, {\em every}
thread is either complete or is {\em active}, that is, able to take a step.

\begin{figure}
%\[
%\begin{array}{l l l l}
%  Actions & \alpha & \bnfdef & \asil \bnfalt
%  \async{b}{v} \bnfalt
%  \asend{b}{v}
%\end{array}
%\]

\[
\begin{array}{c}

\infer
{
}
{
  \atyped{\sig}{\asil}
}

\qquad

\infer
{
  \etyped{\sig, \sigtype{b}{\tau}{\prio}}{\cdot}{v}{\tau}
}
{
  \atyped{\sig, \sigtype{b}{\tau}{\prio}}{\async{b}{v}}
}

\qquad

\infer
{
  \etyped{\sig, \sigtype{b}{\tau}{\prio}}{\cdot}{v}{\tau}
}
{
  \atyped{\sig, \sigtype{b}{\tau}{\prio}}{\asend{b}{v}}
}
\end{array}
\]

\caption{Static semantics for actions.}
\label{fig:act}
\end{figure}

The progress theorems for commands and thread pools also state that, if the
command or thread pool can take a step, the action performed by that step is
well-typed. The typing rules for actions are shown in
Figure~\ref{fig:act} and require that the value returned or received
match the type of the thread.

\begin{thm}[Progress]\label{thm:prog}
\begin{enumerate}
\item If $\etyped{\sig}{\cdot}{e}{\tau}$, then
  either~$e \val{\sig}$ or $e \estep{\sig} e'$.
\item If $\cmdtyped[\worlds]{\sig}{\cdot}{\cmd}{\tau}{\prio}$,
  then either~$\cmd = \kwret{e}$ where $e \val{\sig}$
  or $\cmd \lstep[\prio]{\sig}{\act} \rpconfig{\sig'}{\cmd'}{\mem}$ where
  $\atyped{\sig}{\act}$.
%  $\waitingon{\cmd}{b}$ and $\dthread{b}{\cmd_b} \in \mem$ or
%  $\lconfig{\sig}{\dthread{a}{\cmd}}{\mem} \tstep{\prel}
%  \rconfig{\sig'}{\dthread{a}{\cmd'}}$.
\item If $\mtyped{\worlds}{\sig}{\mem}{\sig'}$ and
  $\sig',\sig'' = \sigtype{a_1}{\tau_1}{\prio_1}, \dots,
  \sigtype{a_n}{\tau_n}{\prio_n}$,
  then $\mem \meq \mbconfig{\sig''}{\dthread{a_1}{\prio_1}{\cmd_1} \mcp \dots \mcp
  \dthread{a_n}{\prio_n}{\cmd_n}}$ and for all
  $i \in [1, n]$, we
  have $\mem \gstep{\sig,\sig'}{\tact{a_i}{\act}} \mem'$
  and~$\atyped{\sig, \sig'}{\act}$.
  %\mbconfig{\sig'''}{\dthread{a_1}{\prio_1}{\cmd}\mcp \dots \mcp
  %  \dthread{a_i}{\prio_i}{\cmd_i'}\mcp \mem' \mcp \dots \mcp
  %  \dthread{a_n}{\prio_n}{\cmd_n}}$ where
  %and if~$\act=\asend{a_i}{v}$ then~\cmd_i'=\cmd_i$.
\end{enumerate}
\end{thm}

\ifproofs
\begin{proof}
\begin{enumerate}
\item By induction on the derivation of
  $\etyped{\sig}{\cdot}{e}{\tau}$. Consider three representative cases.
  \begin{itemize}
  \item \rulename{\kwnat E}. Then $e = \kwifz{v}{e_1}{x}{e_2}$. By inversion,
    $\etyped{\sig}{\cdot}{v}{\kwnat}$. By canonical forms,
    $v = \kwnumeral{n}$ and either
    $e \estep{\sig} e_1$ or $e \estep{\sig} [\kwnumeral{n-1}/x]e_2$.
  \item \rulename{\arrsym E}. Then $e = \kwapply{v_1}{v_2}$. By inversion,
    $\etyped{\sig}{\cdot}{v_1}{\kwarr{\tau_1}{\tau}}$ and
    $\etyped{\sig}{\cdot}{v_2}{\tau_1}$.
    By canonical forms, $e_1 = \kwfun{x}{e_0}$ and
    $e \estep{\sig} [e_2/x]e_0$.
  \item \rulename{\fasym E}. Then $e = \kwwapp{v}{\prio'}$. By inversion,
    $\etyped{\sig}{\cdot}{v}{\kwforall{\vprio}{\cons}{\tau'}}$,
    where $\tau = [\prio/\vprio]\tau'$.
    By canonical forms, $v = \kwwlam{\vprio}{\cons}{e_0}$
    and~$e$ steps by the transition rules.
  \end{itemize}

\item By induction on the derivation of
  $\cmdtyped[\worlds]{\sig}{\cdot}{\cmd}{\tau}{\prio}$.
  \begin{itemize}
  \item \rulename{Bind}. Then $\cmd = \kwbind{e}{x}{\cmd_2}$. By inversion,
    $\etyped{\sig}{\cdot}{e}{\kwcmdt{\tau'}{\prio}}$
    and $\cmdtyped[\worlds]{\sig}{\hastype{x}{\tau}}
    {\cmd_2}{\tau}{\prio}$.
    By induction, either $e \val{\sig}$ or $e \estep{\sig} e'$.
    In the second case,~$\cmd$ steps by rule~\rulename{Bind1}.
    In the first case, by canonical forms, $e = \kwcmd{\prio}{\cmd_1}$ and,
    by inversion on the expression typing rules,
    $\cmdtyped[\worlds]{\sig}{\cdot}{\cmd_1}{\tau'}{\prio}$.
    By induction, either $\cmd_1 = \kwret{e}$ where $e \val{\sig}$ or
    $\cmd_1$ takes a step. In both cases,
    $\cmd$ takes a step~(\rulename{Bind3} or~\rulename{Bind2}).
  \item \rulename{Spawn}. Apply rule~\rulename{Spawn}.
  \item \rulename{Sync}. Then $\cmd = \kwbsync{e}$. By inversion,
    $\etyped{\sig}{\cdot}{e}{\kwat{\tau}{\prio'}}$.
    %and $\ple{\prio}{\prio'}$.
    By induction, either $e \val{\sig}$ or $e \estep{\sig} e'$. In the
    second case,~$\cmd$ steps by rule~\rulename{Sync1}.
    In the first case, by canonical forms, $e = \kwtid{b}$.
    Apply rule~\rulename{Sync2}.
%    By inversion on the typing rules, $\sigtype{b}{\tau'}{\prio'} \in \sig$,
%    so $\cmd$ steps by the transition rules and~$\atyped{\sig}{\async{b}{v}}$.
%    By inversion on the derivation of $\mtyped{\worlds}{\sig}{\mem}$,
%    we have $\dthread{b}{\cmd_b} \in \mem$,
    %% \item \rulename{\boxsym E}. Then $\cmd = \kwletbox{x}{e}{\cmd_0}$. By
  %%   inversion, $\etyped{\sig}{\ctxify{\worlds}}{e}{\kwboxt{\tau'}}$ and
  %%   $\cmdtyped{\sig}{\ctxify{\worlds}, \hastype{x}{\kwcmdt{\tau'}{\prio}}}
  %%   {\cmd_0}{\tau}{\prio}$.
  %%   By induction, either $e \val{\sig}$ or $e \estep{\sig} e'$. In the second
  %%   case, $\cmd$ steps to $\kwletbox{x}{e'}{\cmd_0}$. In the first case,
  %%   by canonical forms, $e = \kwbox{\vprio}{\kwcmd{\vprio}{\cmd_1}}$ and
  %%   $\cmd$ steps to $[\prio/\vprio][\kwcmd{\vprio}{\cmd_1}/x]\cmd_0$.
  %% \item \rulename{\diasym E}. Then $\cmd = \kwletdia{\vprio}{x}{e}{\cmd_0}$. By
  %%   inversion, $\etyped{\sig}{\ctxify{\worlds}}{e}{\kwdia{\tau'}}$ and
  %%   $\cmdtyped{\sig}
  %%   {\ctxify{\worlds}, \vprio \isprio, \hastype{x}{\kwcmdt{\tau'}{\prio}}}
  %%   {\cmd_0}{\tau}{\prio}$.
  %%   By induction, either $e \val{\sig}$ or $e \estep{\sig} e'$. In the second
  %%   case, $\cmd$ steps to $\kwletdia{\vprio}{x}{e'}{\cmd_0}$. In the first case,
  %%   by canonical forms, $e = \kwthere{\prio'}{\kwcmd{\prio'}{\cmd_1}}$ and
  %%   $\cmd$ steps to $[\kwcmd{\prio'}{\cmd_1}/x][\prio'/\vprio]\cmd_0$.
  \item \rulename{Ret}. Then $\cmd = \kwret{e}$. By inversion,
    $\etyped{\sig}{\cdot}{e}{\tau}$.
    By induction, either $e \val{\sig}$ or $e \estep{\sig} e'$. In the
    second case,~$\cmd$ steps by rule~\rulename{Ret}. In the first case,
    the conclusions are trivially satisfied.
  %% \item \rulename{ThreadE}. Then $\cmd = \kwdo{e}$. By inversion,
  %%   $\etyped{\sig}{\ctxify{\worlds}}{e}{\kwcmd{\tau}{\prio}}$.
  %%   By induction, either $e \val{\sig}$ or $e \estep{\sig} e'$. In the second
  %%   case,~$\cmd$ steps to~$\kwdo{e'}$. In the first case, by canonical forms,
  %%   $e = \kwcmd{\prio}{\cmd_0}$ and~$\cmd$ steps to~$\cmd_0$.
  \end{itemize}
\item
  By induction on the derivation of $\mtyped{\worlds}{\sig}{\mem}{\sig'}$.
  We consider the interesting cases.
  \begin{itemize}
  \item \rulename{Concat}.
    By inversion, $\mtyped{\worlds}{\sig,\sig_2}{\mem_1}{\sig_1}$ and
    $\mtyped{\worlds}{\sig,\sig_1}{\mem_2}{\sig_2}$.
    By induction, $\mem_1 \meq \mbconfig{\sig_1'}{\dthread{a_1}{\prio_1}{\cmd_1}
      \mcp \dots \mcp \dthread{a_n}{\prio_n}{\cmd_n}}$ and
    $\mem_2 \meq \mbconfig{\sig_2'}{\dthread{a_{n+1}}{\prio_{n+1}}{\cmd_{n+1}}
      \mcp \dots \mcp \dthread{a_m}{\prio_m}{\cmd_m}}$,
    where
    \iftr\[\sig_1,\sig_1' = \sigtype{a_1}{\tau_1}{\prio_1}, \dots,
    \sigtype{a_n}{\tau_n}{\prio_n}\] and
    \[\sig_2,\sig_2' = \sigtype{a_{n+1}}{\tau_{n+1}}{\prio_{n+1}}, \dots,
    \sigtype{a_m}{\tau_m}{\prio_m}\]%
    \else $\sig_1,\sig_1' = \sigtype{a_1}{\tau_1}{\prio_1}, \dots,
    \sigtype{a_n}{\tau_n}{\prio_n}$
    and~$\sig_2,\sig_2' = \sigtype{a_{n+1}}{\tau_{n+1}}{\prio_{n+1}}, \dots,
    \sigtype{a_m}{\tau_m}{\prio_m}$.
    \fi
    We also have that for all~$i \in [1, n]$,
    $\mem_1 \gstep{\sig,\sig_2,\sig_1}{\tact{a_i}{\act_i}} \mem_1'$ and
    for all~$i \in [n+1,m]$,
    $\mem_2 \gstep{\sig,\sig_1,\sig_2}{\tact{a_i}{\act_i}} \mem_2'$. We have
    $\mem_1 \mcp \mem_2 \meq \mbconfig{\sig_1', \sig_2'}{\mem_1' \mcp \mem_2'}$,
    so the conclusion holds by weakening and~\rulename{DT-Concat-One}.
  \item \rulename{Extend}. Then
    $\mem = \mbconfig{\sig'}{\mem'}$ and~$\mtyped{\sig}{\mem}{\sig''}$.
    By inversion,
    $\mtyped{\worlds}{\sig}{\mem'}{\sig',\sig''}$. By induction,
    $\mem' \meq \mbconfig{\sig'''}{\dthread{a_1}{\prio_1}{\cmd_1}
      \mcp \dots \mcp \dthread{a_n}{\prio_n}{\cmd_n}}$,
    \iftr
    where
    \[\sig',\sig'',\sig'' = \sigtype{a_1}{\tau_1}{\prio_1}, \dots,
    \sigtype{a_n}{\tau_n}{\prio_n}\]
    \else
    where~$\sig',\sig'',\sig'' = \sigtype{a_1}{\tau_1}{\prio_1}, \dots,
    \sigtype{a_n}{\tau_n}{\prio_n}$.
    \fi
    We also have that for all~$i \in [1, n]$,
    $\mem' \gstep{\sig,\sig''}{\act} \mem''$. By the congruence rules,
    $\mem \meq \mbconfig{\sig', \sig'''}{\dthread{a_1}{\prio_1}{\cmd_1}
      \mcp \dots \mcp \dthread{a_n}{\prio_n}{\cmd_n}}$ and the conclusion
    holds by weakening and~\rulename{DT-Extend}.
  %If $\mem = \emem$, the conclusion is vacuously true.
  %If $\sig = \sig, \sigtype{a}{\tau}{\prio}$ and
  %$\mem = \dthread{a}{\cmd} \mcp \mem'$, then by inversion,
  %$\cmdtyped{\sig}{\ctxify{\worlds}}{\cmd}{\tau}{\prio}$ and
  %$\mtyped{\worlds}{\sig'}{\mem'}$. By part (2), the conclusion is true
  %for thread~$a$. By induction, it is true of the remainder of the threads
    %in~$\sig'$.
  \end{itemize}
\end{enumerate}
\end{proof}
\else
%See the technical report~\citep{partial-prio-tr} for the proof.
\fi

The preservation theorem is also split into components for expressions, commands
and thread pools. The theorem for commands requires that any new threads spawned
($\mem'$) meet the extension of the signature ($\sig'$).

\begin{thm}[Preservation]\label{thm:pres}
\begin{enumerate}
\item If $\etyped{\sig}{\cdot}{e}{\tau}$ and
  $e \estep{\sig} e'$, then $\etyped{\sig}{\cdot}{e'}{\tau}$.
\item If $\cmdtyped[\worlds]{\sig}{\cdot}{\cmd}{\tau}{\prio}$
  %and $\mtyped{\sig}{\mem}$ and $\lconfig{\sig, \sigtype{a}{\cmd}{\prio}}
%  {\dthread{a}{\cmd}}{\mem} \tstep{\prel}
  %  \rconfig{\sig, \sig', \sigtype{a}{\tau}{\prio}}{\dthread{a}{\cmd'} \mcp \mem'}$,
  and $\cmd \lstep{\sig}{\act} \rpconfig{\sig'}{\cmd'}{\mem'}$ and
  $\atyped{\sig}{\act}$
  then
  $\cmdtyped[\worlds]{\sig, \sig'}{\cdot}{\cmd'}{\tau}{\prio}$
  and $\mtyped{\worlds}{\sig}{\mem'}{\sig'}$.
\item If $\mtyped{\worlds}{\sig}{\mem}{\sig'}$
  % and~$\wedges$ is compatible
  %with~$\sig$ and $\wedges$ is compatible with $\mem$ and
  %$\mconfig{\sig}{\mem}; \graph \mstep{\prel} \mconfig{\sig, \sig'}{\mem'}; \wedges'$,
  and $\mem \gstep{\sig}{\act} \mem'$
  then $\mtyped{\worlds}{\sig}{\mem'}{\sig'}$% and~$\wedges'$
  %  is compatible with $\sig, \sig'$ and $\graph$ is compatible with $\mem'$.
\item If $\mtyped{\worlds}{\esig}{\mem}{\sig}$ and
  $\mem \pgstep{\acts} \mem'$ then
  $\mtyped{\worlds}{\esig}{\mem'}{\sig}$.
\end{enumerate}
\end{thm}

\ifproofs
\begin{proof}
\begin{enumerate}
\item By induction on the derivation of $e \estep{\sig} e'$.
  %% Consider some representative cases.
  %% \ifproofs
  %% \begin{itemize}
  %% \item $\kwifz{e}{e_1}{x}{e_2} \estep{\sig} \kwifz{e'}{e_1}{x}{e_2}$.
  %%   By inversion on the typing rules, $\etyped{\sig}{\cdot}{e}{\kwnat}$
  %%   and $\etyped{\sig}{\cdot}{e_1}{\tau}$ and
  %%   $\etyped{\sig}{\hastype{x}{\kwnat}}{e_2}{\tau}$.
  %%   By induction, $\etyped{\sig}{\cdot}{e'}{\kwnat}$. Apply rule
  %%   \rulename{\kwnat E}.
  %% \item $\kwifz{e}{e_1}{x}{e_2} \estep{\sig} e_1$.
  %%   By inversion on the typing rules, $\etyped{\sig}{\cdot}{e_1}{\tau}$.
  %% \item $\kwifz{e}{e_1}{x}{e_2} \estep{\sig} [\kwn/x]e_2$.
  %%   By inversion on the typing rules,
  %%   By induction, $\etyped{\sig}{\cdot}{e'}{\kwnat}$. Apply substitution.
  %% \item $\kwwapp{(\kwwlam{\vprio}{\cons}{e})}{\prio'} \estep{\sig}
  %%   [\prio'/\vprio]e$. By inversion on the typing rules,
  %%   $\etyped{\sig}{\cdot}{\kwwlam{\vprio}{\cons}{e}}
  %%   {\kwforall{\vprio}{\cons}{\tau'}}$, where $\tau = [\prio'/\vprio]\tau'$, and
  %%   $\meetc[\worlds]{\cdot}{[\prio'/\vprio]\cons}$.
  %%   By continued inversion,
  %%   $\etyped[\worlds]{\sig}{\vprio \isprio, \cons}{e}{\tau'}$.
  %%   By substitution,
  %%   $\etyped[\worlds]{\sig}{\cdot}%[\prio'/\vprio]\cons}
  %%   {[\prio'/\vprio]e}{[\prio'/\vprio]\tau'}$.
  %% \end{itemize}
  %% \fi
\item By induction on the derivation of $\cmd \lstep{\sig}{\act} \rpconfig{\sig'}{\cmd'}{\mem'}$.
  \begin{itemize}
  \item \rulename{Bind1}. By inversion on the typing rules,
    $\etyped[\worlds]{\sig}{\cdot}{e}{\kwcmdt{\tau'}{\prio'}}$ and
    $\cmdtyped[\worlds]{\sig}{\hastype{x}{\tau'}}
    {\cmd}{\tau}{\prio}$. By induction,
    $\etyped[\worlds]{\sig}{\cdot}{e'}{\kwat{\tau'}{\prio'}}$.
    Apply rule \rulename{Bind}.
  \item \rulename{Bind2}. By inversion on the typing rules,
    $\etyped[\worlds]{\sig}{\cdot}{\kwcmd{\prio}{\cmd_1}}
    {\kwcmdt{\tau'}{\prio}}$ and
    $\cmdtyped[\worlds]{\sig}{\cdot}{\cmd_1}{\tau'}{\prio}$ and
    $\cmdtyped[\worlds]{\sig}{\hastype{x}{\tau'}}
    {\cmd_2}{\tau}{\prio}$. By induction,
    $\cmdtyped[\worlds]{\sig, \sig'}{\cdot}{\cmd_1'}{\tau'}{\prio}$
    and $\mtyped{\worlds}{\sig, \sig'}{\mem'}$. By \rulename{\cmdsym I},
    $\etyped[\worlds]{\sig, \sig'}{\cdot}{\kwcmd{\prio}{\cmd_1'}}
    {\kwcmdt{\tau'}{\prio}}$
    By weakening,
    $\cmdtyped[\worlds]{\sig, \sig'}{\hastype{x}{\tau'}}
    {\cmd_2}{\tau}{\prio}$.
    Apply rule \rulename{Bind}.
  \item \rulename{Bind3}. By inversion on the typing rules,
    $\etyped[\worlds]{\sig}{\cdot}{e}{\tau'}$ and
    $\cmdtyped[\worlds]{\sig}{\hastype{x}{\tau'}}
    {\cmd}{\tau}{\prio}$.
    By substitution, $\cmdtyped[\worlds]{\sig}
    {\cdot}{[e/x]\cmd}{\tau}{\prio}$.
  \item \rulename{Spawn}. By inversion on the typing rules,
    $\cmdtyped[\worlds]{\sig}{\cdot}{\cmd}{\tau'}{\prio'}$.
    By rule~\rulename{OneThread},
    $\mtyped{\worlds}{\sig}{\dthread{b}{\cmd}}{\sigtype{b}{\tau'}{\prio'}}$.
    Apply rules \rulename{TID} and \rulename{Ret}.
  \item \rulename{Sync1}. By inversion on the typing rules,
    $\etyped[\worlds]{\sig}{\cdot}{e}{\kwat{\tau}{\prio'}}$ and
    $\ple{\prio}{\prio'}$.
    By induction,
    $\etyped[\worlds]{\sig}{\cdot}{e'}{\kwat{\tau}{\prio'}}$.
    Apply \rulename{Sync}.
  \item \rulename{Sync2}. By inversion on the typing rules,
    $\etyped[\worlds]{\sig}{\cdot}{\kwtid{b}}{\kwat{\tau}{\prio'}}$
    and $\ple{\prio}{\prio'}$.
    By inversion on the action typing rules,
    $\sigtype{b}{\tau}{\prio'} \in \sig$ and
    $\etyped[\worlds]{\sig}{\cdot}{v}{\tau}$.
    Apply \rulename{Ret}.
    %By inversion on the expression typing rules, 
    %By inversion on the derivation of
    %$\mtyped{\sig}{\mem \mcp \dthread{b}{\kwret{v}} \mcp \mem'}$,
    %we have .
    %
 % \item
  %% \item By inversion on the typing rules,
  %%   $\etyped{\sig}{\ctxify{\worlds}}{\kwbox{\vprio}{\kwcmd{\vprio}{\cmd_1}}}{\kwboxt{\tau'}}$
  %%   and
  %%   $\cmdtyped{\sig}{\ctxify{\worlds}, \hastype{x}{\kwcmdt{\tau'}
  %%       {\prio}}}{\cmd_2}{\tau}{\prio}$.
  %%   By inversion on the expression typing rules,
  %%   $\etyped{\sig}{\ctxify{\worlds}, \vprio \isprio}
  %%   {\kwcmd{\vprio}{\cmd_1}}{\kwcmdt{\tau'}{\vprio}}$.
  %%   By substitution,
  %%   $\cmdtyped{\sig}{\ctxify{\worlds}}{[\prio/\vprio][\kwcmd{\vprio}{\cmd_1}/x]\cmd_2}{\tau}{\prio}$.
  %% \item
  %% \item By inversion on the typing rules,
  %%   $\etyped{\sig}{\ctxify{\worlds}}{\kwthere{\prio'}{\kwcmd{\prio'}{\cmd_1}}}{\kwdia{\tau'}}$
  %%   and $\cmdtyped{\sig}{\ctxify{\worlds}, \vprio \isprio, \hastype{x}{\kwcmdt{\tau'}{\vprio}}}
  %%   {\cmd_2}{\tau}{\prio}$.
  %%   By inversion on the expression typing rules,
  %%   $\etyped{\sig}{\ctxify{\worlds}}{\kwcmd{\prio'}{\cmd_1}}{\kwcmdt{\tau'}{\prio'}}$
  %%   By substitution (XXX TODO),
  %%   $\cmdtyped{\sig}{\ctxify{\worlds}, \hastype{x}{\kwcmdt{\tau'}{\prio'}}}
  %%   {[\prio'/\vprio]\cmd_2}{\tau}{\prio}$. By substitution,
  %%   $\cmdtyped{\sig}{\ctxify{\worlds}}{[\kwcmd{\prio'}{\cmd_1}/x][\prio'/\vprio]\cmd_2}
  %%   {\tau}{\prio}$.
  \end{itemize}
\item By induction on the derivation of $\mem \gstep{\sig}{\act} \mem'$.
  We show the non-trivial cases.
  \begin{itemize}
  \item \rulename{DT-Thread}. By inversion on the typing rules,
    $\cmdtyped[\worlds]{\sig}{\cdot}{\cmd}{\tau}{\prio}$. By induction,
    $\cmdtyped[\worlds]{\sig, \sig'}{\cdot}{\cmd'}{\tau}{\prio}$
    and $\mtyped{\worlds}{\sig}{\mem'}{\sig'}$.
    Apply rules \rulename{OneThread}, \rulename{Concat} and \rulename{Extend}.
  \item \rulename{DT-Sync}. By inversion on the typing rules,
    $\mtyped{\worlds}{\sig,\sig_2}{\mem_1}{\sig_1}$ and
    $\mtyped{\worlds}{\sig,\sig_1}{\mem_2}{\sig_2}$.
    By induction, $\mtyped{\worlds}{\sig,\sig_2}{\mem_1'}{\sig_1}$ and
    $\mtyped{\worlds}{\sig,\sig_1}{\mem_2'}{\sig_2}$.
    Apply rule \rulename{Concat}.
  \item \rulename{DT-Concat}. By inversion on the typing rules,
    $\mtyped{\worlds}{\sig}{\mem_1}{\sig'}$.
    By induction, $\mtyped{\worlds}{\sig}{\mem_1'}{\sig'}$.
    Apply rule \rulename{Concat}.

  \item \rulename{DT-Extend}. By inversion on the typing rules,
    $\mtyped{\worlds}{\sig, \sigtype{a}{\tau}{\prio}}{\mem}{\sig'}$.
    By induction,
    $\mtyped{\worlds}{\sig, \sigtype{a}{\tau}{\prio}}{\mem'}{\sig'}$.
    Apply rule \rulename{Extend}.

  \end{itemize}
  \item By part 3 and~\rulename{Concat}.
\end{enumerate}
\end{proof}
\else
%See the technical report~\citep{partial-prio-tr} for the proof.
The proofs of both theorems can be found in the technical report~\citep{partial-prio-tr}.
\fi

%% In order to show the absence of priority inversions, we also show that a
%% well-typed command can only take a step labeled with a sync action if the
%% sync would not cause a priority inversion.

%% \begin{lem}\label{lem:prio-inv}
%%   If $\cmdtyped{\sig, \sigtype{a}{\tau}{\prio}}{\cdot}{\cmd}{\tau}{\prio}$
%%   and $\dthread{a}{\cmd} \gstep{\sig, \sigtype{a}{\tau}{\prio}}{\async{b}{v}} \mem'$,
%%   then $\sigtype{b}{\tau'}{\prio'} \in \sig$ and $\ple{\prio}{\prio'}$.
%% \end{lem}
%% \begin{proof}
%%   By induction on the derivation of
%%   $\dthread{a}{\cmd} \gstep{\sig, \sigtype{a}{\tau}{\prio}}{\async{b}{v}} \mem'$.
%%   The only thread pool rule that applies is~\rulename{ThreadT}, so continue
%%   by induction on the sub-derivation
%%   $\cmd \lstep[\prio]{\sig}{\act} \rpconfig{\sig'}{\cmd'}{\mem'}$.
%%   The only applicable non-inductive case is~\rulename{Sync2}. In this case,
%%   $\cmd = \kwbsync{\kwtid{b}}$. By inversion on the typing rules,
%%   we have~$\sigtype{b}{\tau'}{\prio'} \in \sig$ and $\ple{\prio}{\prio'}$.
%% \end{proof}

%% We can now combine the results of this section into a theorem that shows
%% both standard type safety and the additional result that no sync step causes
%% a priority inversion.

\begin{thm}[Type Safety]\label{thm:safety}
  If $\mtyped{\worlds}{\esig}{\dthread{a_0}{\prio_0}{\cmd_0}}
  {\sigtype{a_0}{\tau_0}{\prio_0}}$
  and $\dthread{a_0}{\prio_0}{\cmd_0} \pgstep{}^* \mem'$, then
  $\mem' \meq \mbconfig{\sigtype{a_1}{\tau_1}{\prio_1},\dots,
    \sigtype{a_n}{\tau_n}{\prio_n}}{\dthread{a_0}{\prio_0}{\cmd_0'}
      \mcp \dots \mcp
      \dthread{a_n}{\prio_n}{\cmd_n'}}$
  and for all~$i \in [1,n]$, we have
  $\mem' \gstep{\esig}{\tact{a_i}{\act}} \mem''$.
\end{thm}
\begin{proof}
  By inductive application of Theorem~\ref{thm:prog} and Theorem~\ref{thm:pres}.
  %Part (3) follows from Lemma~\ref{lem:prio-inv}.
\end{proof}

\section{Cost Semantics}\label{sec:cost}
So far, we have presented a core calculus for writing parallel programs and
expressing responsiveness requirements using priority annotations.
We have not yet discussed how these requirements are met and what
guarantees can be made.
Doing so is the main theoretical contribution of the remainder of the paper.
We will show how to derive cost bounds (both for computation time and response
time) for {\calcname} programs and then show that, under reasonable
assumptions about scheduling, these bounds hold for the dynamic semantics
of Section~\ref{sec:dyn}.
We first (Section~\ref{sec:dag}) develop a cost model for parallel programs
with partially ordered thread priorities.
The model comes equipped with bounds on computation times and response times.
We then use this model (Sections~\ref{sec:cost-sem} and~\ref{sec:resp-time})
to reason about {\calcname} programs.

\subsection{A Cost Model for Prioritized Parallel Code}\label{sec:dag}
Parallel programs admit an elegant technique for reasoning about their
execution time, in the form of Directed Acyclic Graph, or
DAG models~\citep{BlellochGr95,BlellochGr96}.
Such a model captures the dependences between threads in a program
and, conversely, what portions may be parallelized.
%
%Results going back to~\citet{Brent74} produce from such a DAG an
%execution time bound for the program it represents.
%
%Later work~\citep{MullerAcHa17} extended these models to reason about
%the response time of interactive programs.
%
%In this section, we introduce a further extension of this model where
%threads are assigned priorities from a partial order, as in
%{\calcname}.
%
In DAG models of parallel programs, vertices
represent units of sequential computation and edges represent sequential
dependences.
For example, an edge~$(u_1, u_2)$ indicates that the
computation~$u_1$ must run before~$u_2$.
If there is no directed path
between~$u_1$ and~$u_2$, the two computations may run in parallel.
Without loss of generality, it is typically assumed that each vertex
represents a computation taking a single indivisible unit of time
(perhaps a single processor clock cycle).
These are the units in which we will measure execution time and response time.

%We now discuss how to extend existing models to {\em prioritized DAGs},
%which we will use to model {\calcname} programs.
%
Because threads play such an important role in the design of {\langname}
(and {\calcname}) programs, it will be helpful for us to distinguish in the
DAG model between edges that represent continuations of threads and edges
that represent synchronizations between threads.
In our model, a thread is a sequence of vertices
$\uthread = u_1\tscomp{} u_2 \tscomp{} u_3 \tscomp{} \dots \tscomp{}
u_n$, written~$\ethread$ when $n=0$,
representing a sequence of unit-time operations that are connected by a series of edges
$(u_1, u_2), (u_2, u_3), \dots, (u_{n-1}, u_n)$ representing
sequential dependences.
These are referred to as {\em thread edges} and ensure that the
operations of a thread are performed in the proper sequence.

We then combine threads into a DAG,
$\graph=\dagq{\dthreads}{\spawns}{\syncs}{\polls}$,
in which $\dthreads$ is a mapping from thread names
to a pair consisting of that thread's priority and its sequence of
operations.
We write an element of the mapping as~$\cthread{a}{\prio}{\uthread}$,
and we define~$\uprio{\graph}{u}$
as the priority of the thread to which~$u$ belongs.
The other two components of a DAG are the sets of {\em spawn edges},
$\spawns$, and {\em join edges}, $\syncs$.
A spawn edge~$(u, a)$ indicates that a vertex~$u$ spawned a thread~$a$.
It may be considered an edge from~$u$ to the first vertex of~$a$.
A join edge~$(a, u)$ indicates that vertex~$u$ syncs
(joins) with thread~$a$.
It may be considered an edge from the last vertex of~$a$ to
vertex~$u$.

If there is a path from~$u$ to~$u'$ (using any combination of thread, spawn
and join edges), we say that~$u$ is an {\em ancestor} of~$u'$ (and~$u'$ is
a {\em descendant} of~$u$), and write~$\anc{u}{u'}$.
We will define shorthands for a graph with the (proper) ancestors and
descendants of a vertex~$u$ removed:
\[
\begin{array}{c}
\noancs[\graph]{u} \defeq \graph \setminus \{u' \neq u \mid \anc{u'}{u} \}\\
\nodecs[\graph]{u} \defeq \graph \setminus \{u' \neq u \mid \anc{u}{u'} \}
\end{array}
\]
The {\em competitor work}, $\compwork{\graph}{a}$, of thread~$a$ is
the subgraph formed by the vertices that may be executed in a valid
schedule while~$a$ is active.  More precisely,
if~$\graph = \dagq{\cthread{a}{\prio}{s\tscomp{} \dots \tscomp{}
    t}}{\spawns} {\syncs}{\polls}$, then
  \[\compwork{\graph}{a} \defeq \graph \setminus \{u \neq s \mid \anc{u}{s}\}
  \setminus \{u \neq t \mid \anc{t}{u}\}\]
In these notations the underlying graph, $g$, is left implicit because
it will generally be clear from context.
%When the graph is clear from context (as will usually be the case), we omit it.

\paragraph{The Prompt Scheduling Principle.}
A {\em schedule} of a DAG simulates the execution of a parallel program on
a given number of processors.
The execution proceeds in time steps, each one time unit in length.
At each time step, the schedule designates some number of vertices to
be executed, bounded by the available number of processors, $P$.
A schedule may only execute a vertex if it is {\em ready}, that is, if
all of its ancestors in the DAG have been executed.
%
%One may think of a schedule as a form of pebbling: if~$P$ processors are
%available, at each time step, place up to~$P$ pebbles on ready vertices until
%all vertices have been pebbled.

A {\em greedy} schedule is one in which as many vertices as possible are executed in each time
step, bounded by~$P$ and the number of ready vertices.
%
%The length of, or number of steps in, greedy schedules was bounded by
%\citet{eagerzala89}.
%
Greedy schedules obey provable bounds on computation
time~\citep{EagerZaLa89}, but greediness is insufficient to place
bounds on response time.
To provide such bounds, a schedule must take into account the
thread priorities.
A {\em prompt} schedule~\citep{MullerAcHa17} is a greedy schedule that
prioritizes vertices according to their priority, with high priorities
preferred over low.
Prompt schedules have previously only been used in languages with two
priorities, so more care is required to apply them to an arbitrary
partial order.
%\citet{MullerAcHa17} introduced {\em prompt schedules} for DAGs with two
%priorities (foreground and background), and bounded the response time of
%such schedules. A prompt schedule is greedy, but prioritizes foreground
%vertices.
%
%The extension of the prompt scheduling principle to partially ordered
%priorities is straightforward.
%
At each step, we will assign at most~$P$ vertices to processors and then
execute all of the assigned vertices in parallel.
To begin, assign any ready vertex such that no unassigned vertex has a
higher priority,\footnote{Simply saying ``pick a vertex of the highest
  available priority'' would be correct in a totally ordered setting, but
  might be ambiguous in our partially ordered setting.}
and continue until~$P$ vertices are assigned or no ready vertices remain.
According to this definition, a prompt schedule is necessarily greedy.

\paragraph{Response Time}
Our goal is to show a bound on the response time of threads in
prompt schedules.
In a given schedule, the response time of a thread~$a$,
written~$\resptimeof{a}$, is defined as the number of steps
from when~$s$ is ready (exclusive) to when~$t$ is executed (inclusive).
If our definitions of priority are set up correctly, the response time
of a thread~$a$ at priority~$\prio$ should depend only on the
amount of work at priorities greater than, equal to,
or unrelated to~$\prio$ in the partial order.
%
%Because this set of priorities will come up in many places while analyzing
%response time, we use the notation~$\psnle{\prio}$ to refer to the set of
%priorities not less than~$\prio$.
%
Were the response time of a high-priority thread to depend on the
amount of low-priority work in the computation, there would be a
priority inversion in the schedule, a condition to be avoided.

\paragraph{Well-formed DAGs}
To prove a bound on response time that depends only on work at
priorities not less than~$\prio$, we will need to place an additional
restriction on DAGs.
Consider a DAG with two threads,~$\cthread{a}{\prio_a}{u_1 \tscomp{} \dots
  \tscomp{} u \tscomp{} \dots \tscomp{} u_n}$
and~$\cthread{b}{\prio_b}{\uthread_b}$, where~$\plt{\prio_b}{\prio_a}$.
Suppose there is a join edge~$(b, u)$ from~$b$ to~$a$.
Thread~$a$ will need to wait for~$b$ to complete, so the response
time of~$a$ depends on the length of thread~$b$.
The type system given earlier is designed to rule out such inversions
in programs; we must impose a similar restriction on computation DAGs.
%to rule them out as well.
%

A DAG is {\em well-formed} if no thread depends on lower-priority
work along its critical path.
More precisely, if a thread~$a$
consists of operations~$u_1 \tscomp{} \dots \tscomp{} u_n$, no vertex that
may be executed after~$u_1$ and must be executed before~$u_n$ may have a
priority less than that of~$a$.

\begin{defn}\label{def:wellformed}
A DAG~$\graph = \dagq{\dthreads}{\spawns}{\syncs}{\polls}$ is well-formed if
for all
threads~$\cthread{a}{\prio}{u_1 \tscomp{} \dots \tscomp{} u_n} \in \dthreads$,
if~$\anc{u}{u_n}$ and $\nanc{u}{u_1}$
then~$\ple{\prio}{\uprio{\graph}{u}}$.
\end{defn}

We will show that the well-formedness restriction on DAGs and the type
restrictions imposed on {\calcname} programs coincide in that
well-typed programs give rise only to well-formed DAGs.
In fact, the type system guarantees an even stronger property which
will also be more convenient to prove.
Intuitively, a DAG is {\em strongly well-formed} if 1) all join edges go from
higher-priority threads to lower-priority threads and 2) if a path from~$u$
to~$u'$ starts with a spawn edge and ends with a join edge, there exists
another path from~$u$ to~$u'$ that doesn't go through the spawn edge.
In terms of programs, the second condition means that thread~$a$ can't
sync on thread~$b$ if it doesn't ``know about'' thread~$b$. Because {\calcname}
is purely functional,~$a$ can only know about~$b$ by being descended from
the thread that spawned~$b$.
%
% We now formally define this property and show that it implies
% well-formedness.

\begin{defn}
  A DAG~$\graph= \dagq{\dthreads}{\spawns}{\syncs}{\polls}$ is
  {\em strongly well-formed} if
  for all~$(a, u) \in \syncs$, if~$\cthread{a}{\prio_a}{\uthread},
  \cthread{b}{\prio_b}{\uthread_1 \tscomp{} u \tscomp{} \uthread_2} \in \dthreads$,
  we have that
  \begin{enumerate}
  \item $\ple{\prio_b}{\prio_a}$ and
  \item If~$(u', a) \in \spawns$, then there exists a path from~$u'$ to~$u$
    where the first edge is a thread edge.
  \end{enumerate}
\end{defn}

\begin{lemma}\label{lem:wf-alt}
  If~$\graph$ is strongly well-formed, then~$\graph$ is well-formed.
\end{lemma}
\begin{proof}
  Let $\cthread{a}{\prio}{u_1 \tscomp{} \dots \tscomp{} u_n} \in \dthreads$ and
  let~$\anc{u}{u_n}$.
  We need to show that either~$\anc{u}{u_1}$ or
  $\ple{\prio}{\uprio{\graph}{u}}$.
  Since the graph is finite and acyclic, we can proceed by well-founded
  induction on~$\anc{}{}$.
  If~$u = u_n$, the result is clear.
  Otherwise, assume that
  for all~$u'$ such that~$\neqanc{u}{\anc{u'}{u_n}}$,
  we have~$\anc{u'}{u_1}$ or $\ple{\prio}{\uprio{\graph}{u'}}$.
  If~$\anc{u'}{u_1}$ for any such~$u'$, then~$\anc{u}{u_1}$, so consider the
  case where $\ple{\prio}{\uprio{\graph}{u'}}$ for all such~$u'$.
  Consider the outgoing edges of~$u$ which lead to~$u'$ such
  that~$\anc{u'}{u_n}$. If any is a thread or join edge, then we have
  $\ple{\prio}{\ple{\uprio{\graph}{u'}}{\uprio{\graph}{u}}}$.
  Suppose the only such edge is a spawn edge~$(u, b)$, where~$u'$ is the first
  vertex of thread~$b$.
  If there exists a corresponding join edge~$(b, u'')$ in
  the path, then by assumption there exists a path from~$u$ to~$u''$ where the
  first edge is a thread edge, but this is a contradiction because the
  spawn edge~$(u, b)$ was assumed to be the only outgoing edge from~$u$ to
  an ancestor of~$u_n$.
  If no corresponding join edge~$(b, u'')$ is in the path,
  then~$u_n$ must be in~$b$, so~$u' = u_1$ and~$\anc{u}{u_1}$, also a
  contradiction.
\end{proof}

\paragraph{Bounding Response Time}
We are now ready to bound the response time of threads in prompt schedules using cost metrics that we now
define.

The {\em priority work}~$\prioworkof{\graph}{\psnle{\prio}}$
of a graph~$\graph$ at
a priority~$\prio$ is defined as the number of vertices in the graph at
priorities not less than~$\prio$:
%\[\prioworkof{\graph}{\prio} =
%|\{u \in \graph \mid \uprio{\graph}{u} = \prio \}|\]
%
%We use a convenient shorthand for summing priority work over the set
%$\psnle{\prio}$.
\[
\prioworkof{\graph}{\psnle{\prio}} \defeq
|\{u \in \graph \mid \nple{\uprio{\graph}{u}}{\prio} \}|
%  \fc(\psnle{\prio}) & \defeq &
%  \sum_{\nple{\prio'}{\prio}} \fc(\prio')
%\end{array}
\]
The {\em $a$-span}~$\longp{\graph}{a}$ of a
graph~$\graph \ni \cthread{a}{\prio}{s\tscomp{} \dots \tscomp{} t}$,
is the length of the longest path in~$\graph$ ending at~$t$.

Theorem~\ref{thm:gen-brent} bounds the response time
%(because fairly
%prompt schedules are defined probabilistically, the respones time can
%only be bounded in expectation)
of a thread based on these quantities
which depend only on the work and span of high-priority threads.
Because they deal with scheduling DAGs which are known ahead of time,
results of this form are often known as {\em offline scheduling bounds}.
Later in the section, we will apply this result to executions of the
{\calcname} dynamic semantics as well.

\begin{thm}\label{thm:gen-brent}
  Let~$\graph$ be a well-formed DAG with a thread
  $\cthread{a}{\prio}{\uthread} \in \graph$.
  For any prompt schedule of $\graph$ on~$P$ processors,
  %and any~$\ple{\prio'}{\prio}$,
  \[
  \resptimeof{a} \leq %\frac{1}{\fc(\psnle{\prio'})}
  \frac{\prioworkof{\compwork{}{a}}{\psnle{\prio}}}{P} +
      \longp{\compwork{}{a}}{a}
      \]
\end{thm}
\begin{proof}
  Let~$s$ and~$t$ be the first and last vertices of~$a$, respectively.
  Consider the portion of the schedule from the step in which~$s$ is ready
  (exclusive) to the step in which~$t$ is executed (inclusive).
  %
  %At each step, let~$P_\fc$ be the number of processors attempting to work on
  %vertices of priority in~$\psnle{\prio'}$.
  %
  For each processor at each step, place a token in one of two buckets.
  %
  %If the processor is attempting to work at a priority not less than~$\prio'$,
  %but is unable to, place a token in the ``low'' bucket~$B_l$.
  %
  If the processor is working on a vertex of a priority not less than~$\prio$,
  place a token in the ``high'' bucket~$B_h$;
  otherwise, place a token in the ``low'' bucket~$B_l$.
  %
  %If it is attempting to work at a priority less than~$\prio'$, place a
  %token in the ``fair'' bucket~$B_f$.
  %
  Because~$P$ tokens are placed per step,
  we have $\resptimeof{a} = \frac{1}{P}(B_l + B_h)$,
  where~$B_l$ and $B_h$
  are the number of tokens in the buckets after~$t$ is executed.

  %Let~$\Sigma = \fc(\psnle{\prio'})$.
  %
  %By fairness, we have $B_l + B_h = \Sigma(B_l + B_h + B_f)$.
  %Thus,
  %\[\resptimeof{a} = \frac{1}{P\Sigma}(B_l + B_h)\]
  %
  Each token in~$B_h$ corresponds to work done at priority not less
  than~$\prio$, and thus~$B_h \leq \prioworkof{\graph}{\psnle{\prio}}$, so
  \[\resptimeof{a} \leq %\frac{1}{\Sigma}\left(
  \frac{\prioworkof{\graph}{\psnle{\prio}}}{P}
  + \frac{B_l}{P}\]
  We now need only bound~$B_l$ by~$P\cdot \longp{\compwork{}{a}}{a}$.

  Let step 0 be the step after~$s$ is ready, and let~$\exec{j}$ be the
  set of vertices that have been executed at the start of step~$j$.
  Consider a step~$j$ in which a token is added to~$B_l$.
  For any path ending at~$t$ consisting of vertices
  of~$\graph \setminus \exec{j}$, the path starts at a vertex that is ready
  at the beginning of step~$j$.
  By the definition of well-formedness, this
  vertex must have priority greater than~$\prio$ and is therefore executed
  in step~$j$ by the prompt principle.
  Thus, the length of the path decreases by 1 and so
  $\longp{\graph \setminus \exec{j+1}}{a} =
  \longp{\graph \setminus \exec{j}}{a} - 1$.
  The maximum number of such steps is
  thus~$\longp{\graph \setminus \exec{0}}{a}$, and
  so~$B_l \leq P\cdot \longp{\graph \setminus \exec{0}}{a}$.
  Because~$\noancs{s} \supset \graph \setminus \exec{0}$,
  any path excluding vertices
  in~$\exec{0}$ is contained in~$\noancs{s}$, and
  $\longp{\graph \setminus \exec{0}}{a} \leq
  \longp{\noancs{s}}{a}$, so
  $B_l \leq P\cdot \longp{\noancs{s}}{a} = P\cdot \longp{\compwork{}{a}}{a}$.
\end{proof}

The above theorem not only bounds response time, but computation time as well.
%The above theorem is quite general and it is worthwhile to note several
%special cases.
%
Let~$a$ be the main thread, which is always at the bottommost priority.
The response time of the main thread is equal to the
computation time of the entire program.
Because prompt schedules are greedy, we expect to be able to bound this
time by~$\frac{W}{P} + S$, where~$W$ is the total number of operations in the
program and~$S$ is the length of the longest path in the
DAG~\citep{EagerZaLa89}.
Indeed,
%we have~$\frac{1}{\fc(\psnle{\bot})} = 1$, and
the priority work
and~$a-$span reduce to the overall work and span, respectively, so the
bound given by Theorem~\ref{thm:gen-brent} coincides with the expected
bound on computation time.

%% One might wonder about the purpose of the univerally quantified
%% priority~$\prio'$.
%% %
%% We illustrate its use by considering two extremal cases.
%% %
%% First, let~$\prio' = \prio$.
%% %
%% Doing so yields the bound that might intuitively be expected: the response time
%% of a thread at priority~$\prio$ depends on the work and span at priorities
%% not less than~$\prio$, inflated somewhat by the fairness criterion because
%% only some of the cycles are devoted to work at priorites higher than~$\prio$.
%% %
%% This bound is correct and frequently useful, but will diverge in
%% the case that~$\fc(\psnle{\prio}) = 0$.
%% %
%% In such cases, it may be worthwhile to look at a bound which considers cycles
%% donated from priorities lower than~$\prio$.
%% %
%% This increases the factor~$\frac{1}{\fc(\psnle{\prio})}$ at the cost of
%% having to consider more vertices as competitors.

%% In the extreme case, we set~$\prio' = \bot$.
%% %
%% As above, the multiplicative factor for the fairness criterion goes to 1,
%% but~$\prioworkof{\compwork{}{a}}{\psnle{\prio'}}$ becomes all of the work,
%% at any priority, that may happen in parallel with thread~$a$.
%% %
%% Intuitively, this says that if we ignore priority and simply run a greedy
%% schedule, thread~$a$ will complete eventually but may, in the worst case,
%% have to wait for all of the other work in the system to complete.
 
%%% Local Variables:
%%% mode: latex
%%% TeX-master: "main"
%%% End:

\subsection{Cost Semantics for {\calcname}}\label{sec:cost-sem}
\begin{figure}
\[
\begin{array}{c}

\infer[C-Val]
      {\strut}
      {\ceval[\dassn]{v}{v}{\ethread}}

\qquad

\infer[C-Let]
      {\ceval[\dassn]{e_1}{v_1}{\uthread_1}\\
        \ceval[\dassn]{[v_1/x]e_2}{v}{\uthread_2}\\
        u\fresh}
      {\ceval[\dassn]{\kwlet{x}{e_1}{e_2}}{v}
        {\uthread_1 \tscomp{} u \tscomp{} \uthread_2}}

\qquad

\infer[C-Ifz-NZ]
      {\ceval[\dassn]{[\kwnumeral{n}/x]e_2}{v}{\uthread}\\
        u \fresh}
      {\ceval[\dassn]{\kwifz{\kwnumeral{n+1}}{e_1}{x}{e_2}}{v}
        {u \tscomp{} \uthread}}

\\[4ex]

\infer[C-Ifz-Z]
      {\ceval[\dassn]{e_1}{v}{\uthread}\\
        u \fresh}
      {\ceval[\dassn]{\kwifz{\kwnumeral{0}}{e_1}{x}{e_2}}{v}
        {u \tscomp{} \uthread}}

\qquad

\infer[C-App]
      {\ceval[\dassn]{[v/x]e}{v'}{\uthread}\\
        u \fresh}
      {\ceval[\dassn]{\kwapply{(\kwfun{x}{e})}{v}}{v'}
        {u \tscomp{} \uthread}}

\qquad

\infer[C-Pair]
      {u \fresh}
      {\ceval[\dassn]{\kwepair{v_1}{v_2}}{\kwpair{v_1}{v_2}}
        {u}}

\\[4ex]

\infer[C-Fst]
      {u \fresh}
      {\ceval[\dassn]{\kwfst{\kwpair{v_1}{v_2}}}{v_1}{u}}

\qquad

\infer[C-Snd]
      {u \fresh}
      {\ceval[\dassn]{\kwsnd{\kwpair{v_1}{v_2}}}{v_2}{u}}

\qquad

\infer[C-InL]
      {u \fresh}
      {\ceval[\dassn]{\kweinl{v}}{\kwinl{v}}{u}}

\qquad

\infer[C-InR]
      {u \fresh}
      {\ceval[\dassn]{\kweinr{v}}{\kwinr{v}}{u}}

\\[4ex]

\infer[C-Case-L]
      {\ceval[\dassn]{[v/x]e_1}{v'}{\uthread}\\
        u \fresh}
      {\ceval[\dassn]{\kwcase{\kwinl{v}}{x}{e_1}{y}{e_2}}{v'}
        {u \tscomp{} \uthread}}

\qquad

\infer[C-Case-R]
      {\ceval[\dassn]{[v/y]e_2}{v'}{\uthread}\\
        u \fresh}
      {\ceval[\dassn]{\kwcase{\kwinr{v}}{x}{e_1}{y}{e_2}}{v'}
        {u \tscomp{} \uthread}}

\qquad

\infer[C-Output]
      {u \fresh}
      {\ceval[\dassn]{\kwoutput{v}}{\kwtriv}{u}}

\\[4ex]

\infer[C-Input]
      {%u_1, u_2 \fresh\\
        u\fresh
        %\dassn(\dvar) = \delay
      }
      {\ceval[\dassn]{\kwinput}{\kwn}
        %{u_1 \tscomp{\delay} u_2}}
        {u}}

\qquad

\infer[C-PrApp]
      {\ceval[\dassn]{[\prio/\vprio]e}{v}{\uthread}\\
        u \fresh}
      {\ceval[\dassn]{\kwapply{(\kwwlam{\vprio}{\cons}{e})}{\prio}}{v}
        {u \tscomp{} \uthread}}

\qquad

\infer[C-Fix]
      {\ceval[\dassn]{[v/x]e}{v'}{\uthread}\\
        u \fresh}
      {\ceval[\dassn]{\kwfix{x}{\tau}{e}}{v'}{u \tscomp{} \uthread}}

\\[8ex]

\infer[C-Bind]
      {\ceval[\dassn]{e}{\kwcmd{\prio}{\cmd_1}}{\uthread_1}\\
        \cceval{\tsig}{\sig}{a}{\prio}{\cmd_1}{v}{\graph_1}{\tsig_1}{\sig_1}\\
        u \fresh\\
        \cceval{\tsig_1}{\sig_1}{a}{\prio}{[v/x]\cmd_2}{v'}{\graph_2}{\tsig_2}{\sig_2}\\
        }
      {\cceval{\tsig}{\sig}{a}{\prio}{\kwbind{e}{x}{\cmd_2}}{v'}
        {\sgraph{\uthread_1} \scomp{a}
          \graph_1 \scomp{a} \sgraph{u} \scomp{a} \graph_2}{\tsig_2}{\sig_2}
      }

\\[4ex]

\infer[C-Spawn]
      {b \fresh\\
        \cceval{\tsig}{\sig}{b}{\prio'}{\cmd}{v}
               {\dagq{\dthreads}{\spawns}{\syncs}{\polls}}
               {\tsig,\tsig'}{\sig,\sig'}\\
        u \fresh
        }
      {\cceval{\tsig}{\sig}{a}{\prio}
        {\kwbspawn{\prio'}{\tau}{\cmd}}{\kwtid{b}}
        {\dagq{\cthread{a}{\prio}{u} \uplus \dthreads}
          {\spawns \cup \{(u, b)\}}{\syncs}{\polls}}
        {\tsig,\tsig',\sigcent{b}{v}{\sig'}}
        {\sig,\sigtype{b}{\tau}{\prio'}}
      }

\\[4ex]

\infer[C-Sync]
      {\ceval[\dassn]{e}{\kwtid{b}}{\uthread}\\
        u \fresh}
      {\cceval{\tsig,\sigcent{b}{v}{\sig'}}
        {\sig,\sigtype{b}{\tau}{\prio'}}{a}{\prio}
        {\kwbsync{e}}
        {v}
        {\dagq{\cthread{a}{\prio}{\uthread \tscomp{} u}}{\emptyset}
          {\{(b, u)\}}{\emptyset}}{\tsig,\sigcent{b}{v}{\sig'}}{\sig,\sigtype{b}{\tau}{\prio'}, \sig'}
      }

%% \qquad

%% \infer[C-Poll-Some]
%%       {u \fresh}
%%       {\cceval{a}{\prio}
%%         {\kwpoll{\kwtd{b}{v}}}
%%         {\kwinl{v}}
%%         {\dagq{\cthread{a}{u}}{\emptyset}{\emptyset}{\{(b, u)\}}}}

%% \qquad

%% \infer[C-Poll-None]
%%       {u \fresh}
%%       {\cceval{a}{\prio}
%%         {\kwpoll{\kwtd{b}{v}}}
%%         {\kwinr{\kwtriv}}
%%         {\dagq{\cthread{a}{u}}{\emptyset}{\emptyset}{\emptyset}}}
%%         {\dagq{\dthreads}{\spawns}{\syncs}{\polls}}}

\\[4ex]

\infer[C-Ret]
      {\ceval[\dassn]{e}{v}{\uthread}}
      {\cceval{\tsig}{\sig}{a}{\prio}{\kwret{e}}{v}
        {\dagq{\cthread{a}{\prio}{\uthread}}{\emptyset}{\emptyset}{\emptyset}}
        {\tsig}{\sig}
      }

\\[8ex]

\infer[CT-Thread]
      {\cceval{\tsig}{\sig}{a}{\prio}{\cmd}{v}{\graph}{\tsig'}{\sig'}}
      {\mceval[\dassn]{\tsig, \tsigent{a}{v}{\sig'}}{\sig}
        {\dthread{a}{\prio}{\cmd}}{\tsig'}{\graph}}

\qquad

\infer[CT-Extend]
      {
        \mceval[\dassn]{\tsig}{\sig, \sig'}{\mem}{\tsig'}{\graph}
      }
      {
        \mceval[\dassn]{\tsig}{\sig}{\mbconfig{\sig'}{\mem}}{\tsig'}{\graph}
      }

\\[4ex]

\infer[CT-Concat]
      {
        \mceval[\dassn]{\tsig}{\sig}{\mem}{\tsig_1}
               {\dagq{\dthreads}{\spawns}{\syncs}{\polls}}\\
        \mceval[\dassn]{\tsig}{\sig}{\mem'}{\tsig_2}
               {\dagq{\dthreads'}{\spawns'}{\syncs'}{\polls'}}
      }
      {
        \mceval[\dassn]{\tsig}{\sig}{\mem \mcp \mem'}{\tsig_1, \tsig_2}
               {\dagq{\dthreads \uplus \dthreads'}{\spawns \cup \spawns'}
                 {\syncs \cup \syncs'}{\polls \cup \polls'}}
      }
\end{array}
\]
\caption{Cost semantics of {\calcname}}
\label{fig:cost}
\end{figure}
%%       \infer[C-EmptyTP]
%%       {\strut}
%%       {\mceval[\dassn]{\emem}{\emptyset}}
%%
%% \qquad
%%
%% \infer[C-ConcatTP]
%%       {\cceval{\esig}{\esig}{a}{\prio}{\cmd}{v}
%%         {\dagq{\dthreads}{\spawns}{\syncs}{\polls}}{\tsig}{\sig}\\
%%         \mceval[\dassn]{\mem}{\dagq{\dthreads'}{\spawns'}{\syncs'}{\polls'}}
%%       }
%%       {\mceval[\dassn]{\dthread{a}{\prio}{\cmd} \uplus \mem}
%%         {\dagq{\dthreads \cup \dthreads'}{\spawns \cup \spawns'}
%%           {\syncs \cup \syncs'}{\polls \cup \polls'}}
%%       }
%%
%% 

We develop a cost semantics that evaluates a program, producing a value
and a DAG of the form described in Section~\ref{sec:dag}.
Unlike the
operational semantics of Section~\ref{sec:dyn}, this is an evaluation semantics
that does not fully specify the order in which threads are evaluated.
Figure~\ref{fig:cost} shows the cost semantics for {\calcname} using
three judgments.
The judgment for expressions is~$\ceval[\dassn]{e}{v}{\uthread}$, indicating
that expression~$e$ evaluates to value~$v$ and produces thread~$\uthread$.
The two-level syntax of {\calcname}  ensures that expressions
cannot produce spawn or join edges in the cost graph, and so the rules for
this judgment are quite straightforward: subexpressions are evaluated
to produce sequences of operations, which are then composed sequentially.
The judgment~$\cceval{\tsig}{\sig}{a}{\prio}{\cmd}{v}{\graph}{\tsig'}{\sig'}$
indicates that~$\cmd$ evaluates to~$\kwret{v}$ and produces the graph~$\graph$.
Because threads in our cost graphs are named and annotated with priorities,
the current thread's name and priority are included in the judgment.
The judgment also includes the ambient thread signature before~($\sig$) and
after~($\sig'$) evaluation of the command.
In addition, it includes a {\em thread record}~$\tsig$ (and~$\tsig'$).
The thread record maps a thread name~$a$ to a
pair~$(v_a,\sig_a)$ of the value to which thread~$a$ evaluates, and a
signature containing threads that are (transitively) spawned by~$a$.
%available to thread~$a$ (because they were spawned by~$a$ or its children).
%
The thread record is used by the rule \rulename{C-Sync} to capture the value
of the target thread~$b$, which must be returned by the sync operation. The
rule also captures the signature of threads transitively spawned by~$b$,
which it adds to the signature, indicating that future operations in thread~$a$
now ``know about'' these threads.
In showing the consistency of the cost
semantics later in the section, we will use the
judgment~$\tsigtyped{\sig}{\tsig}$ to indicate that the values in~$\tsig$ are
well-typed.
The following rules apply to the judgment:
\[
\begin{array}{c}
  \infer
      {
      }
      {
        \tsigtyped{\sig}{\esig}
      }
   \qquad
   \infer
       {
         \etyped{\sig,\sigtype{a}{\tau}{\prio},\sig'}{\ectx}{v}{\tau}\\
         \tsigtyped{\sig,\sigtype{a}{\tau}{\prio}}{\tsig}
       }
       {
         \tsigtyped{\sig,\sigtype{a}{\tau}{\prio}}
                   {\tsig,\sigcent{a}{v}{\sig'}}
       }
\end{array}
\]
The other rules are more straightforward.
Rule \rulename{C-Bind} composes the
graphs generated by the subexpressions using the sequential composition
operation defined as follows:
\[
\dagq{\dthread{a}{\prio}{\uthread} \uplus \dthreads}{\spawns}{\syncs}{\polls}
\scomp{a}
\dagq{\dthread{a}{\prio}{\uthread'} \uplus \dthreads'}{\spawns'}{\syncs'}
     {\polls'}
\defeq
\dagq{\dthread{a}{\prio}{\uthread \tscomp{} \uthread'} \uplus
  \dthreads \uplus \dthreads'}{\spawns \cup \spawns'}
     {\syncs \cup \syncs'}{\polls \cup \polls'}
\]
We use the notation~$\sgraph{\uthread}$ to indicate a graph consisting of a
single thread. The name and priority of the thread will generally be evident
from context, e.g. because~$\sgraph{\uthread}$ is immediately sequentially
composed with another graph at thread~$a$, so
\[
\sgraph{\uthread} \scomp{a} \dagq{\dthread{a}{\prio}{\uthread'} \uplus \dthreads}
       {\spawns}{\syncs}{\polls}
       \defeq
       \dagq{\dthread{a}{\prio}{\uthread \tscomp {} \uthread'} \uplus \dthreads}
         {\spawns}{\syncs}{\polls}
\]
Rule~\rulename{C-Spawn} evaluates the newly spawned thread to produce its
cost graph, and then adds it to the graph along with a single vertex~$u$ which
performs the spawn and the appropriate spawn edge.

Finally, the judgment~$\mceval[\dassn]{\tsig}{\sig}{\mem}{\tsig'}{\graph}$
evaluates the thread
pool~$\mem$ to a graph~$\graph$.
The judgment includes the ambient thread record~$\tsig$ and signature~$\sig$
so that when evaluating one thread, we have access to the records of the other
active threads.
A thread pool with a single thread~$\mthread[\delay]{a}{\prio}{\cmd}$
evaluates to the same graph as the command~$\cmd$.
Rule \rulename{CT-Concat} evaluates both parts of the thread pool and composes
the graphs, giving each access to the thread records of the other.

Lemma~\ref{lem:exp-cost-typed} shows that the evaluation judgment on
expressions preserves typing.
The equivalent property for commands will be shown as part of
Lemma~\ref{lem:typed-wf}.

%% \begin{defn}
%%   We say that a graph~$\graph = \dagq{\dthreads}{\spawns}{\syncs}{\polls}$
%%   is {\em well-formed except for}~$\sig$ if for all
%% $(a, u) \in \syncs$, if~$\cthread{a}{\prio_a}{\uthread},
%% \cthread{b}{\prio_b}{\uthread_1 \tscomp{} u \tscomp{} \uthread_2} \in \dthreads$,
%% we have that
%% \begin{enumerate}
%% \item $\ple{\prio_b}{\prio_a}$ and
%% \item If~$a \not\in \dom{\sig}$ and $(u', a) \in \spawns$
%%   then there exists a path from~$u'$ to~$u$,
%%   where the first edge is a thread edge.
%% \end{enumerate}
%% \end{defn}

%% By Lemma~\ref{lem:wf-alt}, if a DAG is well-formed except for~$\esig$, then
%% it is well-formed.

\begin{lemma}\label{lem:exp-cost-typed}
  If~$\etyped{\sig}{\ectx}{e}{\tau}$ and~$\ceval[\dassn]{e}{v}{\uthread}$,
  then~$\etyped{\sig}{\ectx}{v}{\tau}$.
\end{lemma}
\begin{proof}
  By induction on the derivation of~$\ceval[\dassn]{e}{v}{\uthread}$.
\end{proof}

%% \begin{lemma}
%%   \begin{enumerate}
%%   \item If~$\etyped{\sig}{\ectx}{e}{\tau}$
%%     and~$\ceval[\dassn]{e}{v}{\uthread}$, then~$\etyped{\sig}{\ectx}{v}{\tau}$.
%%   \item If~$\cmdtyped{\sig}{\ectx}{\cmd}{\tau}{\prio}$
%%     and~$\cceval{\sig}{a}{\prio}{\cmd}{v}{\graph}{\sig'}$
%%     then~$\etyped{\sig'}{\ectx}{v}{\tau}$.
%%   \end{enumerate}
%% \end{lemma}
%% \begin{proof}
%%   \begin{enumerate}
%%   \item By induction on the derivation of~$\ceval[\dassn]{e}{v}{\uthread}$.
%%   \item By induction on the derivation
%%     of~$\cceval{\sig}{a}{\prio}{\cmd}{v}{\graph}{\sig'}$.
%%     \begin{itemize}
%%     \item \rulename{C-Bind}.
%%       Then~$\cmd = \kwbind{e}{x}{\cmd_2}$
%%       and~$\ceval[\dassn]{e}{\kwcmd{\prio}{\cmd_1}}{\uthread}$
%%       and~$\cceval{\sig}{\cmd_1}{v}{\graph}{\sig'}$
%%       and~$\cceval{\sig'}{[v/x]\cmd_2}{v'}{\graph'}{\sig''}$.
%%       By inversion on the typing rules,~$\etyped{\sig}{\ectx}{e}
%%       {\kwcmdt{\tau}{\prio}}$
%%       and~$\cmdtyped{\sig}{\hastype{x}{\tau}}{\cmd_2}{\tau'}{\prio}$.
%%       By part 1,~$\etyped{\sig}{\ectx}{\kwcmd{\prio}{\cmd_1}}
%%       {\kwcmdt{\tau}{\prio}}$.
%%       By substitution,$\cmdtyped{\sig}{[v/x]\cmd_2}{\tau'}{\prio}$.
%%       By induction,~$\etyped{\sig'}{\ectx}{v}{\tau}$
%%       and~$\etyped{\sig''}{\ectx}{v'}{\tau'}{\prio}$.
%%     \item \rulename{C-Spawn}.
%%       Then~$\cmd = \kwbspawn{\prio'}{\tau}{\cmd_0}$
%%       and~$\cceval{\sig}{b}{\prio'}{\cmd_0}{v}
%%       By inversion,~$\cmdtyped{\sig}{\ectx}{\cmd_0}{
%%     \item \rulename{C-Sync}.

One more technical result we will need in Section~\ref{sec:resp-time} is that
entries in the thread record for threads that don't appear in a command or
thread pool are unnecessary for the purposes of the cost semantics.

%% \cceval{\tsig}{\sig}{a}{\prio}{\cmd}{v}{\graph}{\tsig'}{\sig'}

\begin{lemma}\label{lem:unused-tsig}
  \begin{enumerate}
  \item If~$\cmdtyped{\sig}{\ectx}{\cmd}{\tau}{\prio}$
    and~$\cceval{\tsig, \tsigent{c}{v_c}{\sig_c}}{\sig}{a}{\prio}{\cmd}
    {v}{\graph}{\tsig', \tsigent{c}{v_c}{\sig_c}}{\sig'}$
    and~$c \not\in \dom{\sig}$,
    then~$\cceval{\tsig}{\sig}{a}{\prio}{\cmd}{v}{\graph}{\tsig'}{\sig'}$.
  \item If~$\mtyped{\prios}{\sig}{\mem}{\sig'}$ and
    $\mceval[\dassn]{\tsig, \tsigent{c}{v_c}{\sig_c}}{\sig}
    {\mem}{\tsig', \tsigent{c}{v_c}{\sig_c}}{\graph}$
    and~$c \not\in \dom{\sig}$,
    then~$\mceval[\dassn]{\tsig}{\sig}{\mem}{\tsig'}{\graph}$.
  \end{enumerate}
\end{lemma}
\begin{proof}
  \begin{enumerate}
    \item
      By induction on the derivation of~$\cceval{\tsig, \tsigent{c}{v_c}{\sig_c}
        {\prio'}}{\sig}{a}{\prio}{\cmd}
      {v}{\graph}{\tsig',\tsigent{c}{v_c}{\sig_c}}{\sig'}$.
    \ifproofs The interesting case is~\rulename{C-Sync}.
    \[
    \infer[C-Sync]
      {\ceval[\dassn]{e}{\kwtid{b}}{\uthread}\\
        u \fresh}
      {\cceval{\tsig,\sigcent{b}{v}{\sig'}}
        {\sig,\sigtype{b}{\tau}{\prio'}}{a}{\prio}
        {\kwbsync{e}}
        {v}
        {\dagq{\cthread{a}{\prio}{\uthread \tscomp{} u}}{\emptyset}
          {\{(b, u)\}}{\emptyset}}{\tsig,\sigcent{b}{v}{\sig'}}{\sig,\sigtype{b}{\tau}{\prio'}, \sig'}
      }
      \]
      \[
      \begin{plines}
        \have{c \neq b}
        \cjust{$c \not\in\dom{\sig,\sigtype{b}{\tau}{\prio'}}$}\\
        \have{\cceval{\tsig,\tsigent{b}{v}{\sig'}}
          {\sig,\sigtype{b}{\tau}{\prio'}}{a}{\prio}{\kwbsync{e}\\ &}
          {v}{\tsig,\tsigent{b}{v}{\sig'}}
          {\sig,\sigtype{b}{\tau}{\prio'}, \sig'}
          {\dagq{\cthread{a}{\prio}{\uthread \tscomp{} u}}{\emptyset}
          {\{(b, u)\}}{\emptyset}}}
        \cjust{\rulename{C-Sync}}
      \end{plines}
      \]
      \fi
    \item By induction on the derivation of
      $\mceval[\dassn]{\tsig, \tsigent{b}{v'}{\prio'}}{\sig}
      {\mem}{\tsig', \tsigent{b}{v'}{\prio'}}{\graph}$.
      All cases follow from induction.
  \end{enumerate}
\end{proof}

We now show that well-typed programs produce strongly well-formed cost graphs.
%
%Recall the two conditions of strong well-formedness: 1) all join edges go from
%higher-priority threads to lower-priority threads and 2) if a path from~$u$
%to~$u'$ starts with a spawn edge and ends with a join edge, there is another
%path.
%
%The first condition is directly guaranteed by the type system.
%
%We show the second condition using the signatures that are threaded through
%the cost semantics.
%
We maintain the invariant that if~$b \in \dom{\sig}$ when an operation
corresponding to vertex~$u$ in thread~$a$ is typed, then the vertex that
spawned~$b$ must be an ancestor of~$u$.
We say that a graph for which this invariant holds is {\em compatible}
with~$\sig$ at~$a$.

\begin{defn}
  We say that a graph~$\graph=\dagq{\dthreads}{\spawns}{\syncs}{\polls}$
  is compatible with a signature~$\sig$ at~$a$ if
  \begin{enumerate}
  \item $\cthread{a}{\prio_a}{\uthread_a \tscomp{} t_a} \in \dthreads$
  \item for all~$b \in \dom{\sig}$, if~$(u, b) \in \spawns$, then~$\anc{u}{t_a}$.
    %we
  %have~$\cthread{b}{\prio_b}{\uthread_b \tscomp{} t_b} \in \dthreads$
  %and~$(u, b) \in \spawns$ for some~$\anc{u}{t_a}$.
  \end{enumerate}

  We say that a graph~$\graph$ is compatible with a thread record~$\tsig$
  if for all~$\sigcent{b}{v}{\sig'} \in \tsig$, it is the case that~$\graph$
  is compatible with~$\sig'$ at~$b$.
\end{defn}

\ifproofs
We show some facts about compatibility and strong well-formedness
that will be useful later:
\begin{lemma}\label{lem:compat-facts}
  \begin{enumerate}
  \item If~$\graph$ is compatible with a signature~$\sig$ at~$a$, then
    $\graph \scomp{a} \sthread{\uthread}$ is compatible with~$\sig$ at~$a$
    and $\sthread{\uthread} \scomp{a} \graph$ is compatible with~$\sig$ at~$a$.
  \item If~$\graph$ is compatible with~$\tsig$, then
    $\graph \scomp{a} \sthread{\uthread}$ is compatible with~$\tsig$
    and $\sthread{\uthread} \scomp{a} \graph$ is compatible with~$\tsig$.
  \item If~$\graph_1$ and~$\graph_2$ are compatible with~$\sig$ at~$a$, then
    $\graph_1 \scomp{a} \graph_2$ is compatible with~$\sig$ at~$a$.
  \item If~$\graph_1$ and~$\graph_2$ are compatible with~$\tsig$, then
    $\graph_1 \scomp{a} \graph_2$ is compatible with~$\tsig$.
  \item If~$\graph$ is strongly well-formed, then
    $\graph \scomp{a} \sthread{\uthread}$ is strongly well-formed
    and $\sthread{\uthread} \scomp{a} \graph$ is strongly well-formed.
  \end{enumerate}
\end{lemma}

\begin{proof}
  \begin{enumerate}
  \item Part (1) of compatibility is immediate from the definitions, as is
    part (2) for $\sthread{\uthread} \scomp{a} \graph$. To show
    part (2) on $\graph \scomp{a} \sthread{\uthread}$,
    let~$b \in \dom{\sig}$, suppose~$(u,b) \in \spawns$.
    By definition,~$\anc{u}{t_a}$, where~$t_a$ is the last vertex of~$a$
    in~$\graph$. We have~$\anc{t_a}{t_a'}$ where~$t_a'$ is the last vertex
    of~$a$ in~$\uthread$.
  \item Composing at~$a$ doesn't change the structure of any other
    thread.%, so compatibility with~$\tsig$ is preserved.
  \item Let~$\graph_1 = \dagq{\dthreads_1}{\spawns}{\syncs}{\polls}$
    and~$\graph_2 = \dagq{\dthreads_2}{\spawns'}{\syncs'}{\polls'}$,
    where~$t_1$ is the last vertex of~$a$ in~$\graph_1$ and~$t_2$ is the
    last vertex of~$a$ in~$\graph_2$.
    Part (1) of compatibility is immediate from the definitions. For part
    (2), let~$b \in \dom{\sig}$ and suppose~$(u,b) \in \spawns$.
    Then~$\anc{u}{\anc{t_1}{t_2}}$.
    Now suppose~$(u,b) \in \spawns'$. Then~$\anc{u}{t_2}$ immediately.
  \item Composing at~$a$ doesn't change the structure of any other
    thread.%, so compatibility with~$\tsig$ is preserved.
  \item No join edges are added in either case, so strong well-formedness is
    preserved by composition.
  \end{enumerate}
\end{proof}
\fi

Compatibility gives the final piece needed to show that a graph is
strongly well-formed: if a vertex~$u$ syncs on a thread~$b$, then~$b$
must be in the signature~$\sig$ used to type the sync operation~$u$,
and if the graph generated up to this point is compatible with~$\sig$, the
vertex that spawned~$b$ is an ancestor of~$u$.
At first glance, the phrase ``the graph generated up to this point'' seems
terribly non-compositional.
This would be worrisome, as we wish to be able to prove a large graph
well-formed by breaking it into subgraphs and showing the result by induction.
To do so, we posit the existence
of a graph~$\graph'$ which is well-formed and compatible with the current
signature and thread record.
This graph represents ``the graph generated up to this point''.
%and will be
%filled in appropriately when we compose subgraphs.

\begin{lemma}\label{lem:typed-wf}
  If~$\cmdtyped{\sig}{\ectx}{\cmd}{\tau}{\prio}$ and
  $\tsigtyped{\sig}{\tsig}$ and
  $\cceval{\tsig}{\sig}{a}{\prio}{\cmd}{v}{\graph}{\tsig'}{\sig'}$
  and there exists~$\graph'$ such that:
  \begin{enumerate}
  \item $\graph'$ is strongly well-formed
  \item $\graph'$ is compatible with~$\sig$ at~$a$ and
  \item $\graph'$ is compatible with~$\tsig$
  \end{enumerate}
  then
  \begin{enumerate}
  \item $\graph=\dagq{\cthread{a}{\prio}{\uthread}
    \uplus \dthreads}{\spawns}{\syncs}{\polls}$
  \item $\sig'$ extends~$\sig$
  \item $\graph' \scomp{a} \graph$ is strongly well-formed
    %    \item All vertices of~$\graph$ are accessible from~$u_1$
  \item $\graph' \scomp{a} \graph$ is compatible with~$\sig'$ at~$a$.
  \item $\graph' \scomp{a} \graph$ is compatible with~$\tsig'$.
  \item $\etyped{\sig'}{\ectx}{v}{\tau}$.
    %If~$d \in \dom{\sig'} \setminus \dom{\sig}$, we have
    %$(u'', d) \in \spawns \cup \spawns'$ and there is a path from~$u''$ to
    %$u_n$ in $\graph' \scomp{a} \graph$.
  \end{enumerate}
\end{lemma}
\begin{proof}
  By induction on the derivation of
  $\cceval{\tsig}{\sig}{a}{\prio}{\cmd}{v}{\graph}{\tsig'}{\sig'}$.
  \ifproofs
  \begin{itemize}
  \item[Case]
      \[\infer[C-Bind]
      {\ceval[\dassn]{e}{\kwcmd{\prio}{\cmd_1}}{\uthread_1}\\
        \cceval{\tsig}{\sig}{a}{\prio}{\cmd_1}{v}{\graph_1}{\tsig_1}{\sig_1}
               \\
        u \fresh\\
        \cceval{\tsig_1}{\sig_1}{a}{\prio}{[v/x]\cmd_2}{v'}{\graph_2}
               {\tsig_2}{\sig_2}
               \\
        }
      {\cceval{\tsig}{\sig}{a}{\prio}{\kwbind{e}{x}{\cmd_2}}{v'}
        {\sgraph{\uthread_1} \scomp{a}
          \graph_1 \scomp{a} \sgraph{u} \scomp{a} \graph_2}
        {\tsig_2}{\sig_2}
      }\]
      This case follows from the inductive hypothesis applied to the
      second subderivation if we can show that the conditions of the lemma
      hold for~$\graph' \scomp{a} \sgraph{\uthread_1} \scomp{a} \graph_1
      \scomp{a} \sgraph{u}$.
      These in turn hold from Lemma~\ref{lem:compat-facts} and the
      inductive hypothesis applied to the first subderivation if we can show
      that the conditions of the lemma hold
      for~$\graph' \scomp{a} \sgraph{u_1}$.
      This follows from Lemma~\ref{lem:compat-facts} and the assumptions.

    \item[Case]
      \[\infer[C-Spawn]
      {b \fresh\\
        \cceval{\tsig}{\sig}{b}{\prio}{\cmd}{v}
               {\dagq{\dthreads}{\spawns}{\syncs}{\polls}}
               {\tsig,\tsig'}{\sig,\sig'}
               \\
        u \fresh
        }
      {\cceval{\tsig}{\sig}{a}{\prio}{\kwbspawn{\prio'}{\tau}{\cmd}\\}
        {\kwtid{b}}{\dagq{\cthread{a}{\prio}{u} \uplus \dthreads}
          {\spawns \cup \{(u, b)\}}{\syncs}{\polls}}
        {\tsig,\tsig',\tsigent{b}{v}{\sig'}}
          {\sig,\sigtype{b}{\tau}{\prio'}}
      }\]
      Then $\graph = \dagq{\cthread{a}{\prio}{\sgraph{u}} \uplus \dthreads}
      {\spawns \cup \{(u, b)\}}{\syncs}{\polls}$. By induction,
      $\sig,\sig'$ extends~$\sig$ and
      $\graph' \scomp{a} \dagq{\dthreads}{\spawns}{\syncs}{\polls}$ is
      strongly well-formed and compatible with~$\sig,\sig'$ at~$b$ and
      is compatible with~$\tsig,\tsig'$ and
      $\etyped{\sig,\sig'}{\ectx}{v}{\tau}$.
      We have that~$\graph' \scomp{a} \graph$ is strongly well-formed as well,
      because this adds no join edges.
      Because~$\graph' \scomp{a} \graph$ is compatible with~$\sig,\sig'$ at~$b$
      and is compatible with~$\tsig,\tsig'$, we have that
      $\graph' \scomp{a} \graph$ is compatible
      with~$\tsig,\tsig',\tsigent{b}{v}{\sig'}$. It remains to show that
      $\graph' \scomp{a} \graph$ is compatible
      with~$\sig,\sigtype{b}{\tau}{\prio}$ at~$a$. This is the case
      because~$\graph' \scomp{a} \graph$ is compatible with~$\sig$ at~$a$
      and~$(u, b) \in \spawns \cup \{(u, b)\}$.

    \item[Case]
      \[
      \infer[C-Sync]
      {\ceval[\dassn]{e}{\kwtid{b}}{\uthread}\\
        u \fresh}
      {\cceval{\tsig,\tsigent{b}{v}{\sig'}}
        {\sig,\sigtype{b}{\tau}{\prio'}}{a}{\prio}{\kwbsync{e}\\}{v}
        {\dagq{\cthread{a}{\prio}{\uthread \tscomp{} u}}{\emptyset}
          {\{(b, u)\}}{\emptyset}}
        {\tsig,\tsigent{b}{v}{\sig'}}{\sig,\sigtype{b}{\tau}{\prio'}, \sig'}}
      \]
      Then $\graph = \dagq{\cthread{a}{\prio}{\uthread \tscomp{} \sgraph{u}}}
      {\emptyset}{\{(b, u)\}}{\emptyset}$.
      The only join edge added to form~$\graph' \scomp{a} \graph$ is~$(b, u)$.
      By inversion on the typing rule, we must have~$\ple{\prio}{\prio'}$.
      Because~$\graph'$ is compatible with~$\sig,\sigtype{b}{\tau}{\prio'}$
      at~$a$, if~$(u', b) \in \graph'$ then~$\anc{u'}{u}$, so
      $\graph' \scomp{a} \graph$ is strongly well-formed.
      In addition, it remains compatible with~$\tsig,\tsigent{b}{v}{\sig'}$.
      By inversion on~$\tsigtyped{\sig,\sigtype{b}{\tau}{\prio'}}
      {\tsig,\tsigent{b}{v}{\sig'}}$, we must
      have~$\etyped{\sig,\sigtype{b}{\tau}{\prio'},\sig'}{\ectx}{v}{\tau}$.
      It remains to show that~$\graph' \scomp{a} \graph$ is compatible
      with~$\sig,\sigtype{b}{\tau}{\prio'},\sig'$ at~$a$ and in particular
      that for all~$c \in \dom{\sig'}$,
      if~$(u'', c) \in \graph' \scomp{a}{\graph}$, then~$\anc{u''}{u}$.
      Because~$\graph'$ is compatible with~$\tsig,\tsigent{b}{v}{\sig'}$,
      we have that~$\graph'$ is compatible with~$\sig'$ at~$b$, so~$u''$
      is an ancestor of the last vertex of~$b$ in~$\graph'$ and is therefore
      an ancestor of~$u$ in~$\graph' \scomp{a}{\graph}$.

    \item[Case]
      \[\infer[C-Ret]
      {\ceval[\dassn]{e}{v}{\uthread}}
      {\cceval{\tsig}{\sig}{a}{\prio}{\kwret{e}}{v}
        {\dagq{\cthread{a}{\prio}{\uthread}}{\emptyset}{\emptyset}{\emptyset}}
        {\tsig}{\sig}
      }\]
      By Lemma~\ref{lem:exp-cost-typed}, we
      have~$\etyped{\sig}{\ectx}{v}{\tau}$.
      The other conditions follow from Lemma~\ref{lem:compat-facts}.
  \end{itemize}
  \fi
\end{proof}

In order to show that a full graph generated by a well-typed program is
strongly well-formed, we simply observe that ``the graph generated up to this
point'' is empty, and trivially satisfies the requirements of the lemma.

\begin{cor}
  If~$\cmdtyped{\esig}{\ectx}{\cmd}{\tau}{\prio}$ and
  $\cceval{\esig}{\esig}{a}{\prio}{\cmd}{v}{\graph}{\tsig}{\sig}$, then
  $\graph$ is well-formed.
\end{cor}
\begin{proof}
  Because~$\egraph$ is strongly well-formed and compatible with~$\esig$,
  Lemma~\ref{lem:typed-wf} shows that~$\graph$ is strongly well-formed, and
  is thus well-formed by Lemma~\ref{lem:wf-alt}.
\end{proof}

\subsection{Response Time Bound for Operational Semantics}\label{sec:resp-time}
Thus far in this section, we have developed a DAG-based cost model for
{\calcname} programs and showed an offline scheduling bound which holds for
DAGs derived from well-typed {\calcname} programs.
Although the DAGs are built upon our intuitions of how {\calcname}
programs execute, they are still abstract artifacts which must, in order to
be valuable, be shown to correspond to more concrete, runtime notions.
%
%Just as we justify the correctness of type systems by showing that
%``well-typed programs do not go wrong'', we should justify the correctness
%of our cost model by showing that it correctly predicts runtime costs.

Our goal in this section is to show that an execution of a {\calcname} program
using the dynamic semantics corresponds to a valid schedule of the DAG
generated from that program.
Because well-typed programs admit the cost bound
of Theorem~\ref{thm:gen-brent}, we may then directly appeal to that theorem
for cost bounds on programs.
The argument proceeds as follows:
\begin{enumerate}
\item Lemmas~\ref{lem:ready-threads} and~\ref{lem:global-ready-threads}
  show that a thread of a DAG is ready (i.e. its first unexecuted
  vertex is ready) if and only if the corresponding thread in the program
  may take a step.
\item Lemma~\ref{lem:step-cost} shows that stepping some set of threads in
  the dynamic semantics corresponds to executing the first vertex of those
  threads in a schedule of the DAG.
\item Lemma~\ref{lem:schedule} combines the above results to establish a
  correspondence between an execution of a {\calcname} program and a schedule
  of its cost graph.
\item Finally, we use Theorem~\ref{thm:gen-brent} to bound the length of
  the schedule and therefore the length of the execution in the dynamic
  semantics.
\end{enumerate}

The correspondence between ready DAG threads and active thread pool
threads requires intermediate results about expressions and commands.
Part (1) of Lemma~\ref{lem:ready-threads} states that
an expression produces an empty thread if and only if it is a value.
Part (2) states that
a command a) takes a silent step if and only if it produces a graph
with a ready first vertex, b) returns a value if and only if it produces
an empty graph and c) takes a sync step if and only if it produces a
graph with an incoming join edge.
Parts (3) and (4) extend part (2) to thread pools.
Part (4) in particular states that if the first vertex of a thread is ready
in a graph, the corresponding thread in the thread pool can take a silent step.
The key observation in proving part (4) from part (2) is that if a
vertex~$u$ has an incoming join edge~$(b, u)$ but thread~$b$ is empty, then
thread~$b$ must be returning a value and~$u$ can perform the sync, taking
a silent step with rule \rulename{D-Sync}.

\begin{lemma}\label{lem:ready-threads}
  \begin{enumerate}
  \item If $\etyped{\sig}{\ectx}{e}{\tau}$ and
    $\ceval[\dassn]{e}{v}{\uthread}$, then
    $e \estep{\sig} e'$ for some~$e'$ if and only if~$\uthread$ is nonempty.
  \item If $\cmdtyped{\sig}{\ectx}{\cmd}{\tau}{\prio}$ and
    $\cceval{\tsig}{\sig}{a}{\prio}{\cmd}{v}{\graph}{\tsig'}{\sig'}$, then
    $\graph = \dagq{\cthread{a}{\prio}{\uthread}
    \uplus \dthreads}{\spawns}{\syncs}{\polls}$,
    and~$\graph$ has no spawn edges to threads in~$\sig$ and has no join edges
    to active threads other than~$a$, and
    one of the following is true:
    \begin{enumerate}
    \item There exists~$\cmd'$ such that
      $\cmd \lstep{\sig}{\asil} \cmd'$ and
      $\uthread = u\tscomp{\delay} \uthread'$ and~$u$ is ready in $\graph$.
    \item There exists~$v$ such that~$v\val{\sig}$ and
      $\cmd = \kwret{v}$ and $\uthread = \ethread$.
    \item There exist~$v$ and~$\cmd'$ such that
      $\cmd \lstep{\sig}{\async{b}{v}} \cmd'$ and
      $\uthread = u\tscomp{\delay} \uthread'$ and
      there exists an edge~$(b,u) \in \graph$, which is the only in-edge of~$u$.
    \end{enumerate}
  \item If $\mtyped{\worlds}{\sig}{\mem}{\sig',\sigtype{a}{\tau}{\prio}}$ and
    $\mceval[\dassn]{\tsig}{\sig}{\mem}{\tsig'}{\graph}$ where
    $\graph =\dagq{\dthreads}{\spawns}{\syncs}{\polls}$
    and~$\mem \gstep{\sig}{\tact{a}{\act}} \mem'$,
    then~$\graph$ has no spawn or join edges to threads not in~$\dthreads$ and
    \begin{enumerate}
    \item If~$\act = \asil$, then
      $\cthread{a}{\prio}{u\tscomp{\delay} \uthread} \in \dthreads$ and~$u$
      is ready in~$\graph$.
    \item If~$\act = \asend{b}{v}$, then
      $\cthread{a}{\prio}{\ethread} \in \dthreads$.
    \item If~$\act = \async{b}{v}$, then
      $\cthread{a}{\prio}{u\tscomp{\delay} \uthread} \in \dthreads$ and there
      exists an edge~$(b,u) \in \graph$, which is the only in-edge of~$u$.
    \end{enumerate}
  \item If $\mtyped{\worlds}{\esig}{\mem}{\sig}$ and
    $\mceval[\dassn]{\tsig}{\sig}{\mem}{\tsig'}{\graph}$
    and the first vertex of~$a$ is ready in~$\graph$, then there
    exists~$\mem'$ such that $\mem \gstep{\esig}{\tact{a}{\asil}} \mem'$.
  \end{enumerate}

\end{lemma}
\ifproofs
\begin{proof}
  \begin{enumerate}
  \item By Theorem~\ref{thm:prog}, either $e \estep{\sig} e'$ or~$e\val{\sig}$.
    It remains to show that~$e\val{\sig}$ if and only if~$\uthread=\ethread$.
    Both directions are clear by inspection of the cost semantics.
  \item By Theorem~\ref{thm:prog} and inspection of the dynamic semantics,
    the three cases given are exhaustive. If~$\cmd = \kwret{v}$, then
    apply~\rulename{C-Val} and~\rulename{C-Ret}. Otherwise,
    proceed by induction on the derivation of $\cmd \lstep{\sig}{\act} \cmd'$.

    \begin{itemize}
    \item[Case]
      \[\infer[D-Bind1]
{
  e \estep{\sig} e'
}
{
  \kwbind{e}{x}{\cmd}
  \lstep{\sig}{\asil}
  \rpconfig{\esig}{\kwbind{e'}{x}{\cmd}}{\emem}
}
      \]
      \[
      \begin{plines}
        \have{\ceval[\dassn]{e}{v'}{\uthread'},~\uthread'~\text{nonempty}}
        \just{inversion on~\rulename{C-Bind}, part 1}\\
        \have{\graph=\dagq{\cthread{a}{\prio}{\uthread'\tscomp{}\uthread''}
            \uplus \dthreads}{\spawns}{\syncs}{\polls}}
        \just{inversion on~\rulename{C-Bind}}
      \end{plines}
      \]
    \item[Case]
      \[\infer[D-Bind2]
{
  \cmd_1 \lstep{\sig}{\act} \rpconfig{\sig'}{\cmd_1'}{\mem'}
}
{
  \kwbind{\kwcmd{\prio}{\cmd_1}}{x}{\cmd_2}
  \lstep{\sig}{\act}
  \rpconfig{\sig'}{\kwbind{\kwcmd{\prio}{\cmd_1'}}{x}{\cmd_2}}{\mem'}
}
              \]
              \[
      \begin{plines}
        \have{\cceval{\tsig}{\sig}{a}{\prio}{\cmd_1}{v}
          {\graph_1}{\tsig'}{\sig''}}
        \just{inversion on cost semantics}\\
        \have{\graph_1, \act~\text{meet condition 2(a) or 2(c) of the lemma}}
        \just{induction}\\
        \have{\graph, \act~\text{meet condition 2(a) or 2(c) of the lemma}}
        \just{\rulename{C-Bind}}
      \end{plines}
      \]
    \item[Case]
      \[\infer[D-Bind3]
{
  e \val{\sig}
}
{
  \kwbind{\kwcmd{\prio}{\kwret{e}}}{x}{\cmd}
  \lstep{\sig}{\asil}
  \rpconfig{\esig}{[e/x]\cmd}{\emem}
}
      \]
      \[
      \begin{plines}
        \have{\graph=\dagq{\cthread{a}{\prio}{u\tscomp{}\uthread} \uplus
            \dthreads}{\spawns}{\syncs}{\polls}}
        \just{inversion on~\rulename{C-Bind}}
      \end{plines}
      \]
    \item[Case]
      \[
      \infer[D-Spawn]
{
  b \fresh
}
{
  \kwbspawn{\prio}{\tau}{\cmd}
  \lstep{\sig}{\asil}
  \rpconfig{\sigtype{b}{\tau}{\prio}}{\kwret{\kwtid{b}}}{\dthread{b}{\prio}{\cmd}}
}
      \]
      \[
      \begin{plines}
        \have{\graph = \dagq{\cthread{a}{\prio}{u} \uplus \dthreads}
          {\spawns \cup \{(u, b)\}}{\syncs}{\polls}}
        \just{inversion on~\rulename{C-Spawn}}
      \end{plines}
      \]
    \item[Case]
      \[\infer[D-Sync1]
{
  e \estep{\sig} e'
}
{
  \kwbsync{e}
  \lstep{\sig}{\asil}
  \rpconfig{\esig}{\kwbsync{e'}}{\emem}
}
      \]
      \[
      \begin{plines}
        \have{\ceval[\dassn]{e}{v'}{\uthread'},~\uthread'~\text{nonempty}}
        \just{inversion on~\rulename{C-Sync}, part 1}\\
        \have{\graph=\dagq{\cthread{a}{\prio}{\uthread'\tscomp{}\uthread''}
            \uplus \dthreads}{\spawns}{\syncs}{\polls}}
        \just{inversion on~\rulename{C-Sync}}
      \end{plines}
      \]
    \item[Case]
      \[
      \infer[D-Sync2]
{
  v \val{\sig}%\\
  %\ple{\prio}{\prio'}
}
{
  \kwbsync{(\kwtid{b})}
  \lstep[\prio]{\sig}{\async{b}{v}}
  \rpconfig{\esig}{\kwret{v}}{\emem}
}
      \]
      \[
      \begin{plines}
        \have{\graph = \dagq{\cthread{a}{\prio}{\uthread' \tscomp{} u}}{\emptyset}
          {\{(b, u)\}}{\emptyset}}
        \just{inversion on~\rulename{C-Sync}}\\
        \have{\ceval[\dassn]{v}{v}{\uthread'}}
        \just{inversion on~\rulename{C-Sync}}\\
        \have{\uthread' = \ethread}
        \just{part 1}
      \end{plines}
      \]

    \item[Case]
      \[
      \infer[D-Ret]
{
  e \estep{\sig} e'
}
{
  \kwret{e}
  \lstep{\sig}{\asil}
  \rpconfig{\esig}{\kwret{e'}}{\emem}
}
      \]
      \[
      \begin{plines}
        \have{\ceval[\dassn]{e}{v}{\uthread'},\uthread'~\text{nonempty}}
        \just{inversion on~\rulename{C-Ret}, part 1}\\
        \have{\dagq{\cthread{a}{\prio}{\uthread'}}{\emptyset}{\emptyset}{\emptyset}}
        \just{inversion on~\rulename{C-Ret}}
      \end{plines}
      \]
    \end{itemize}

  \item
    By induction on the derivation of~$\mem \gstep{\sig}{\tact{a}{\act}} \mem'$.

        \begin{itemize}
    \item[Case]
      \[\infer[DT-Thread]
{
  \cmd \lstep[\prio]{\sig}{\act}
  \rpconfig{\sig'}{\cmd'}{\mem'}
}
{
  \dthread{a}{\prio}{\cmd}
  \gstep{\sigtype{a}{\tau}{\prio}, \sig}{\tact{a}{\act}}
  \mbconfig{\sig'}{\dthread{a}{\prio}{\cmd'} \mcp \mem'}
}
      \]
      \[
      \begin{plines}
        \have{\cceval{\tsig}{\sig}{a}{\prio}{\cmd}{v}{\graph}{\tsig'}{\sig'}}
        \just{inversion on~\rulename{CT-Thread}}\\
        \have{\graph = \dagq{\cthread{a}{\prio}{\uthread} \uplus \dthreads}
          {\spawns}{\syncs}{\polls}}
        \just{part 2}\\
        \have{\text{condition (a) or (c) holds on}~\graph}
        \just{part 2}
      \end{plines}
      \]
    \item[Case]
      \[\infer[DT-Ret]
{
  v \val{\sigtype{a}{\tau}{\prio}, \sig}
}
{
  \dthread{a}{\prio}{\kwret{v}}
  \gstep{\sigtype{a}{\tau}{\prio}, \sig}{\tact{a}{\asend{a}{v}}}
  \dthread{a}{\prio}{\kwret{v}}
}
      \]
      \[
      \begin{plines}
        \have{\cceval{\tsig}{\sig}{a}{\prio}{\kwret{v}}{v}{\graph}
          {\tsig'}{\sig'}}
        \just{inversion on~\rulename{CT-Ret}}\\
        \have{\graph = \tgraph{a}{\prio}{\ethread}}
        \just{part 2}
      \end{plines}
      \]
    \item[Case]
      \[\infer[DT-Sync]
{
  \sig = \sig', \sigtype{a}{\tau_a}{\prio_a}, \sigtype{b}{\tau_b}{\prio_b}\\
  %\ple{\prio_a}{\prio_b}\\
  \mem_1 \gstep{\sig}{\tact{a}{\async{b}{v}}} \mem_1'\\
  \mem_2 \gstep{\sig}{\tact{b}{\asend{b}{v}}} \mem_2
}
{
  \mem_1 \mcp \mem_2
  \gstep{\sig}{\tact{a}{\asil}}
  \mem_1' \mcp \mem_2
}
      \]
      \[\begin{plines}
      \have{\mceval[\dassn]{\tsig, \tsig_2}{\sig}{\mem_1}{\tsig_1}
        {\dagq{\cthread{a}{\prio_a}{u\tscomp{} \uthread_a} \uplus \dthreads}
          {\spawns}{\syncs}{\polls}}}
      \ljust{inversion on \rulename{CT-Concat}, induction}\\
      \have{\mceval[\dassn]{\tsig, \tsig_1}{\sig}{\mem_2}{\tsig_2}
        {\dagq{\cthread{b}{\prio_b}{\ethread} \uplus \dthreads'}
          {\spawns'}{\syncs'}{\polls'}}}
      \ljust{inversion on \rulename{CT-Concat}, induction}\\
      \have{\graph =
        \dagq{\cthread{a}{\prio_a}{u\tscomp{}\uthread_a} \uplus
          \cthread{b}{\prio_b}{\ethread} \uplus \dthreads \uplus \dthreads'}
             {\spawns \cup \spawns'}{\syncs \cup \syncs'}{\polls \cup \polls'}}
      \ljust{\rulename{CT-Concat}}\\
      \have{\text{no edges in}~\spawns', \syncs'~\text{target}~u}
      \cjust{induction}\\
      \have{u~\text{is ready in}~\graph}
      \cjust{$b$ is empty, so we may ignore the edge~$(b,u)$}
      \end{plines}
      \]
    \item[Case]
      \[\infer[DT-Concat]
{
  \mem_1 \gstep{\sig}{\tact{a}{\act}} \mem_1'
}
{
  \mem_1 \mcp \mem_2
  \gstep{\sig}{\tact{a}{\act}}
  \mem_1' \mcp \mem_2
}
      \]
      \[\begin{plines}
      \have{\mceval[\dassn]{\tsig,\tsig_2}{\sig}{\mem_1}{\tsig_1}{\graph_1},
        \graph_1 =
        {\dagq{\cthread{a}{\prio_a}{\uthread} \uplus \dthreads}
          {\spawns}{\syncs}{\polls}}}
      \just{induction}\\
      \have{\mceval[\dassn]{\tsig, \tsig_1}{\sig}{\mem_2}{\tsig_2}
        {\dagq{\dthreads'}{\spawns'}{\syncs'}{\polls'}}}
      \just{induction}\\
      \have{\graph =
        \dagq{\cthread{a}{\prio_a}{\uthread} \uplus
          \dthreads \uplus \dthreads'}
             {\spawns \cup \spawns'}{\syncs \cup \syncs'}{\polls \cup \polls'}}
      \just{\rulename{CT-Concat}}\\
      \pcase{\act = \asil}
        \have{\uthread = u \tscomp{} \uthread', u~\text{is ready in}~\graph_1}
        \just{induction}\\
        \have{\text{no edges in}~\spawns',\syncs'~\text{target}~u}
        \just{induction}\\
        \have{u~\text{is ready in}~\graph}\\
      \pcase{\act = \asend{a}{v}}
        \have{\uthread = \ethread}
        \just{induction}
      \pcase{\act = \async{b}{v}}
        \have{\uthread = u \tscomp{\delay} \uthread', \exists
          (b,u)~\text{the only in-edge of}~u~\text{in}~\graph_1}
        \just{induction}\\
        \have{\text{no edges in}~\spawns',\syncs'~\text{target}~u}
        \just{induction}\\
        \have{(b,u)~\text{is the only in-edge of}~u~\text{in}~\graph}
      \end{plines}
      \]
    \item[Case]
      \[\infer[DT-Extend]
{
  \mem \gstep{\sig, \sigtype{a}{\tau}{\prio}}{\tact{b}{\act}} \mem'
}
{
  \mbconfig{\sigtype{a}{\tau}{\prio}}{\mem}
  \gstep{\sig}{\tact{b}{\act}}
  \mbconfig{\sigtype{a}{\tau}{\prio}}{\mem'}
}
      \]
      \[\begin{plines}
      \have{\mceval[\dassn]{\tsig}{\sig,\sigtype{a}{\tau}{\prio}}{\mem}
        {\tsig'}{\graph}}
      \just{inversion on~\rulename{CT-Extend}}\\
      \have{\graph~\text{meets the conditions of the lemma}}
      \just{induction}
      \end{plines}
      \]
        \end{itemize}

      \item
    \[\begin{plines}
    \have{\mem \meq \mbconfig{\sig'}{\mthread{a}{\prio}{\cmd} \mcp \mem_0},
      \text{where}~\dom{\mem_0} \cup \{a\} = \dom{\sig'}}
    \cjust{Theorem~\ref{thm:prog}}\\
    \have{\mem \gstep{\sig,\sig'}{\tact{a}{\act}} \mem'~\text{and}~
      \atyped{\sig,\sig'}{\act}}
    \cjust{Theorem~\ref{thm:prog}}\\
    \have{\act \neq \asend{a}{v}}
    \cjust{part 3 would imply~$a$ has no vertices in~$\graph$, a contradiction}
    \\
    \pcase{\act = \asil}
      \have{\text{Conclusion holds trivially}}
    \pcase{\act = \async{b}{v}}
      \have{(b, u) \in \graph,~\text{where}~u~\text{is the first vertex of}~
        a~\text{and this is the lone in-edge of}~u}
      \ljust{part 3}\\
      \have{\sigtype{b}{\tau_b}{\prio_b} \in \sig,\sig'}
      \cjust{inversion on the static semantics for actions}\\
      \have{\mthread{b}{\prio_b}{\cmd_b} \in \mem_0}
      \cjust{assumption}\\
      \have{b~\text{is empty in}~\graph}
      \cjust{$u$~is ready in~$\graph$}\\
      \have{\cmd_b = \kwret{v}}
      \cjust{part 2}\\
      \have{\exists \mem''. \mem \gstep{\sig,\sig'}{\tact{a}{\asil}} \mem''}
      \cjust{\rulename{DT-Ret}, \rulename{DT-Sync}}
    \end{plines}
    \]
  \end{enumerate}
\end{proof}
\else
%The proofs of cases (1), (2) and (4) make the observation that, by
%Theorem~\ref{thm:prog}, a step must be possible or evaluation has completed.
%All cases then proceed by induction on the step taken, if any.
The full proof is available in the technical report~\citep{partial-prio-tr}.
\fi

Parts (3) and (4) of Lemma~\ref{lem:ready-threads} state that a thread can
take a silent step if and only if its first vertex is ready in the
corresponding graph.
However, this result still considers only sequential
execution: if threads~$a$ and~$b$ are both ready in the graph, it says nothing
about whether~$a$ and~$b$ can step {\em in parallel}.
Lemma~\ref{lem:global-ready-threads} extends the result to parallel steps.
It states that a set~${a_1, \dots, a_n}$ of threads that are ready in~$\graph$
may all step simultaneously, and that any set of threads that can take a
parallel step must be ready in~$\graph$.

\begin{lemma}\label{lem:global-ready-threads}
  Let~$R = \{a \mid \cthread{a}{\prio}{u \tscomp{\delay} \uthread} \in \graph, u
  \text{ is ready in } \graph\}$.
  If $\mtyped{\worlds}{\sig}{\mem}$ and
  $\mceval[\dassn]{\tsig}{\sig}{\mem}{\tsig'}{\graph}$, then
  \begin{enumerate}
    \item For any subset~$\{a_1, \dots, a_n\}$ of~$R$, we have
      $\mem \pgstep{\{a_1, \dots, a_n\}} \mem'$.
    \item If
      $\mem \pgstep{\{a_1, \dots, a_n\}} \mem'$,
      then $\{a_1, \dots, a_n\} \subset R$.
  \end{enumerate}
\end{lemma}
\begin{proof}
  \begin{enumerate}
  \item
    By Theorem~\ref{thm:safety}, we have
    $\mem \meq \mbconfig{\sig'}{\dthread{a_1}{\prio_1}{\cmd_1} \mcp \dots
      \mcp \dthread{a_m}{\prio_m}{\cmd_m}}$.
    For all~$a_i \in R$, we have that by Lemma~\ref{lem:ready-threads},
    $\mem \gstep{\sig}{\tact{a_i}{\asil}} \mem_i'$.
    A straightforward induction
    on $\mem \gstep{\sig}{\tact{a_i}{\asil}} \mem_i'$ shows that
    $\mem_i' \meq \mbconfig{\sig''}{\dthread{a_1}{\prio_1}{\cmd_1} \mcp \dots
      \mcp \dthread{a_i}{\prio_i}{\cmd_i'} \mcp \mem_i'' \mcp \dots
      \mcp \dthread{a_m}{\prio_m}{\cmd_m}}$.
    Applying this reasoning to all~$a_i \in \{a_1, \dots, a_n\}$
    allows us to apply rule~\rulename{DT-Par}.
  \item Let $i \in [1, n]$. By inversion on rule~\rulename{DT-Par},
    $\mem \gstep{\esig}{\tact{a_i}{\asil}} \mem_i'$.
    By Lemma~\ref{lem:ready-threads}, $a_i \in R$.
  \end{enumerate}
\end{proof}

We now move on to showing that a parallel transition corresponds to a step
of a schedule. At a more precise level, Lemma~\ref{lem:step-cost} shows that
if a thread pool~$\mem'$ produces a graph~$\graph'$
and~$\mem$ steps to~$\mem'$, then~$\mem$ produces a graph isomorphic
to~$\graph$ sequentially post-composed with one vertex for each thread that
was stepped.

Stating this formally requires us to define a new graph composition
operator~$\gscomp{a}$
which composes a thread with a graph~$g$ by adding outgoing edges from the
thread to {\em all} sources of~$g$, with the edge to~$a$ being a continuation
edge and all other edges being spawn edges
(as opposed to~$\scomp{a}$ which adds an edge only to thread~$a$).
\[
\begin{array}{l l}
& \sgraph{\uthread} \gscomp{a}
\dagq{\cthread{a}{\prio}{\uthread'} \uplus
  \cthread{a_1}{\prio_1}{\uthread_1}\dots \uplus
  \cthread{a_n}{\prio_n}{\uthread_n}}{\spawns}{\syncs}{\polls}\\
\defeq &
\dagq{\cthread{a}{\prio}{\uthread \tscomp{\delay} \uthread'} \uplus
  \cthread{a_1}{\prio_1}{\uthread_1}\dots \uplus
  \cthread{a_n}{\prio_n}{\uthread_n}}
     {\spawns \cup \{(u, a_1), \dots, (u, a_n)\}}{\syncs}{\polls}
\end{array}
\]

\begin{lemma}\label{lem:step-cost}
  \begin{enumerate}
  \item If $\ceval[\dassn]{e'}{v}{\uthread}$ and $e \estep{\sig} e'$, then
    $\ceval[\dassn]{e}{v}{u \tscomp{} \uthread}$.
  \item If $\mceval[\dassn]{\tsig}{\sig}{\dthread{a}{\prio}{\cmd'} \uplus \mem'}
    {\tsig''}{\graph}$ and
    $\cmd \lstep{\sig}{\act} \rpconfig{\sig'}{\cmd'}{\mem'}$, then
    $\cceval{\tsig}{\sig}{a}{\prio}{\cmd}{v}{\graph_0}{\tsig'}{\sig'}$,
    where~$\graph_0$ is isomorphic to~$\sgraph{u} \gscomp{a} \graph$.
  \item If $\mceval[\dassn]{\tsig}{\sig}{\mem'}{\tsig'}{\graph'}$ and
    $\mem \pgstep{\{\tact{a_i}{\asil}, \dots, \tact{a_n}{\asil}\}} \mem'$, then
    $\graph'$ can be decomposed into
    $\graph_0 \uplus \graph_1' \uplus \dots \uplus \graph_n'$, and
    $\mceval[\dassn]{\tsig}{\sig}{\mem}{\tsig'}{\graph}$,
    where~$\graph$ is isomorphic to
    $\graph_0 \uplus (\sgraph{u_1} \gscomp{a_1} \graph_1) \uplus \dots \uplus
    (\sgraph{u_n} \gscomp{a_n} \graph_n)$.
  \end{enumerate}
\end{lemma}
\ifproofs
\begin{proof}
  \begin{enumerate}
  \item By induction on the derivation of $e \estep{\sig} e'$.
    \begin{itemize}
    \item \rulename{D-Let-Step}. Then~$e = \kwlet{x}{e_1}{e_2}$ and
      $e' = \kwlet{x}{e_1'}{e_2}$ and $e_1 \estep{\sig} e_1'$. By
      inversion on \rulename{C-Let},
      $\ceval[\dassn]{e_1'}{v_1}{\uthread_1}$ and
      $\ceval[\dassn]{e'}{v}{\uthread_1 \tscomp{} u \tscomp{} \uthread_2}$.
      By induction, $\ceval[\dassn]{e_1}{v_1}{u' \tscomp{} \uthread_1}$.
      Apply \rulename{C-Let}.
    \item \rulename{D-Let}. Then~$e = \kwlet{x}{v_1}{e_2}$ and~$e' = [v_1/x]e_2$.
      Apply \rulename{C-Let}.
    \item \rulename{D-Ifz-NZ}. Then $e = \kwifz{\kwnumeral{n+1}}{e_1}{x}{e_2}$
      and $e' = [\kwnumeral{n}/x]e_2$.
      Apply \rulename{D-Ifz-NZ}.
    \item \rulename{D-Ifz-Z}. Then $e = \kwifz{\kwnumeral{0}}{e_1}{x}{e_2}$
      and $e' = e_1$.
      Apply \rulename{D-Ifz-Z}.
    \item \rulename{D-App}. Then $e = \kwapply{(\kwfun{x}{e_1})}{v}$ and
      $e' = [v/x]e_1$. Apply \rulename{C-App}.
    \item \rulename{D-Pair}. Then $e = \kwepair{v_1}{v_2}$ and
      $e' = \kwpair{v_1}{v_2}$. By inversion on the cost rules, we have
      $\uthread = \ethread$. Apply rule \rulename{C-Pair}.
    \item \rulename{D-Fst}. Then $e = \kwfst{\kwpair{v_1}{v_2}}$ and $e' = v_1$.
      By inversion on the cost rules, we have
      $\uthread = \ethread$. Apply rule \rulename{C-Fst}.
    \item \rulename{D-Snd}. Then $e = \kwfst{\kwpair{v_1}{v_2}}$ and $e' = v_2$.
      By inversion on the cost rules, we have
      $\uthread = \ethread$. Apply rule \rulename{C-Snd}.
    \item \rulename{D-InL}. Then $e = \kweinl{v}$ and $e' = \kwinl{v}$.
      By inversion on the cost rules, we have
      $\uthread = \ethread$. Apply rule \rulename{C-InL}.
    \item \rulename{D-InR}. Then $e = \kweinr{v}$ and $e' = \kwinr{v}$.
      By inversion on the cost rules, we have
      $\uthread = \ethread$. Apply rule \rulename{C-InR}.
    \item \rulename{D-Case-L}. Then $e = \kwcase{\kwinl{v}}{x}{e_1}{y}{e_2}$
      and $e' = [v/x]e_1$. Apply \rulename{C-Case-L}.
    \item \rulename{D-Case-R}. Then $e = \kwcase{\kwinr{v}}{x}{e_1}{y}{e_2}$
      and $e' = [v/x]e_2$. Apply \rulename{C-Case-R}.
    \item \rulename{D-Output}. Then $e = \kwoutput{\kwn}$ and $e' = \kwtriv$.
      By inversion on the cost rules, we have
      $\uthread = \ethread$. Apply rule \rulename{C-Output}.
    \item \rulename{D-Input}. Then $e = \kwinput$ and $e' = \kwnumeral{n}$.
      By inversion on the cost rules, we have
      $\uthread = \ethread$. Apply rule \rulename{C-Input}.
    \item \rulename{D-PrApp}.
      Then $e = \kwwapp{(\kwwlam{\vprio}{\cons}{e_1})}{\prio'}$ and
      $e' = [\prio'/\vprio]e_1$. Apply rule \rulename{C-PrApp}.
    \item \rulename{D-Fix}.
      Then $e = \kwfix{x}{\tau}{e}$ and $e' = [\kwfix{x}{\tau}{e}/x] e$.
      Apply \rulename{C-Fix}.
    \end{itemize}
  \item By induction on the derivation of
    $\cmd \lstep{\sig}{\act} \rpconfig{\sig'}{\cmd'}{\mem'}$.
    \begin{itemize}
    \item[Case]
      \[
      \infer[D-Bind1]
{
  e \estep{\sig} e'
}
{
  \kwbind{e}{x}{\cmd}
  \lstep{\sig}{\asil}
  \rpconfig{\esig}{\kwbind{e'}{x}{\cmd}}{\emem}
}
            \]
      \[
      \begin{plines}
        \have{\ceval[\dassn]{e'}{\kwcmd{\prio'}{\cmd_1}}{\uthread_1}}
        \cjust{inversion on \rulename{C-Bind}}\\
        \have{\cceval{\tsig}{\sig}{a}{\prio}{\cmd_1}{v_1}{\graph_1}
          {\tsig_1}{\sig_1}}
        \cjust{inversion on \rulename{C-Bind}}\\
        \have{\cceval{\tsig_1}{\sig_1}{a}{\prio}{[v_1/x]\cmd}{v}
          {\graph_2}{\tsig''}{\sig''}}
        \cjust{inversion on \rulename{C-Bind}}\\
        \have{\mceval[\dassn]{\tsig}{\sig}
          {\mthread[\delay]{a}{\prio}{\kwbind{e'}{x}{\cmd}}}{\tsig''}
          {\sgraph{\uthread_1}
          \scomp{a} \graph_1 \scomp{a} \sgraph{u} \scomp{a} \graph_2}}
        \ljust{inversion on \rulename{CT-Concat}, \rulename{CT-Thread},
          \rulename{C-Bind}}\\
        \have{\ceval[\dassn]{e}{\kwcmd{\prio'}{\cmd_1}}
          {u' \tscomp{} \uthread_1}}
        \cjust{part 1}\\
        \have{\cceval{\tsig}{\sig}{a}{\prio}{\kwbind{e}{x}{\cmd}}{v}
          {\sgraph{u' \tscomp{\delay + 1} \uthread_1}
            \scomp{a} \graph_1 \scomp{a} \sgraph{u} \scomp{a} \graph_2}
          {\tsig''}{\sig''}}
        \cjust{\rulename{C-Bind}}\\
        \have{\mceval[\dassn]{\tsig}{\sig}
          {\mthread{a}{\prio}{\kwbind{e}{x}{\cmd}}}{\tsig''}
          {\sgraph{u'} \tscomp{} \graph}}
        \cjust{\rulename{CT-Concat}}
      \end{plines}
      \]
    \item[Case]
      \[\infer[D-Bind2]
{
  \cmd_1 \lstep{\sig}{\act} \rpconfig{\sig'}{\cmd_1'}{\mem'}
}
{
  \kwbind{\kwcmd{\prio}{\cmd_1}}{x}{\cmd_2}
  \lstep{\sig}{\act}
  \rpconfig{\sig'}{\kwbind{\kwcmd{\prio}{\cmd_1'}}{x}{\cmd_2}}{\mem'}
}
\]
      \[
      \begin{plines}
        \have{\cceval{\tsig}{\sig}{a}{\prio}{\cmd_1'}{v_1}{\graph_1}
          {\tsig_1}{\sig_1}}
        \cjust{inversion on \rulename{C-Bind}}\\
        \have{\cceval{\tsig_1}{\sig_1}{a}{\prio}{[v_1/x]\cmd}
          {v}{\graph_2}{\tsig''}{\sig''}}
        \cjust{inversion on \rulename{C-Bind}}\\
        \have{\mceval[\dassn]{\tsig,\tsigent{a}{v_1}{\sig_1}}{\sig}
          {\mthread{a}{\prio}{\kwbind{\cmd_1'}{x}{\cmd_2}}}{\tsig''}
          {(\graph_1 \scomp{a} \sgraph{u} \scomp{a} \graph_2)}}
        \ljust{inversion on \rulename{CT-Concat}, \rulename{CT-Thread},
          \rulename{C-Bind}}\\
        \have{\mceval[\dassn]{\tsig,\tsigent{a}{v_1}{\sig_1}}{\sig}{\mem'}
          {\tsig_3}{\graph_3}}
        \cjust{inversion on \rulename{CT-Concat}}\\
        \have{\mceval[\dassn]{\tsig}{\sig}{\mem'}{\tsig_3}{\graph_3}}
        \cjust{inversion on \rulename{DT-Thread}, Theorem~\ref{thm:pres},
          Lemma~\ref{lem:unused-tsig}}\\
        \have{\mceval[\dassn]{\tsig}{\mthread{a}{\prio}{\cmd_1'}
            \mcp \mem'}{\tsig_1,\tsig_3}{\graph_1 \uplus \graph_3}}
        \cjust{\rulename{CT-Concat}}\\
        \have{\cceval{\tsig}{\sig}{a}{\prio}{\cmd_1}
          {v_1}{u' \gscomp{a} \graph_1}{\tsig_1,\tsig_3}{\sig_1,\sig_3}}
        \cjust{induction}\\
        \have{\cceval{\tsig}{\sig}{a}{\prio}{\kwbind{\cmd_1}{x}{\cmd_2}}
          {v}{\sgraph{u'} \gscomp{a}
            (\graph_1 \scomp{a} \sgraph{u} \scomp{a} \graph_2)}
          {\tsig'',\tsig_3}{\sig'',\sig_3}}
        \ljust{\rulename{C-Val}, \rulename{C-Bind}}\\
        \have{\mceval[\dassn]{\tsig,\tsigent{a}{v_1}{\sig_1}}{\sig}
          {\mthread{a}{\prio}{\kwbind{\cmd_1}{x}{\cmd_2}}}
          {\tsig'',\tsig_3}
          {\sgraph{u'} \gscomp{a} \graph}}
        \cjust{\rulename{CT-Concat}}
      \end{plines}
      \]
    \item[Case]
      \[\infer[D-Bind3]
{
  e \val{\sig}
}
{
  \kwbind{\kwcmd{\prio}{\kwret{e}}}{x}{\cmd}
  \lstep{\sig}{\asil}
  \rpconfig{\esig}{[e/x]\cmd}{\emem}
}
      \]
      \[
      \begin{plines}
        \have{\cceval{\tsig}{\sig}{a}{\prio}{[e/x]\cmd}{v}{\graph}
          {\tsig'}{\sig'}}
        \just{inversion on~\rulename{CT-Thread}}\\
        \have{\cceval{\tsig}{\sig}{a}{\prio}{[e/x]\cmd}{v}
          {\sgraph{u} \gscomp{a} \graph}{\tsig'}{\sig'}}
        \just{\rulename{C-Val}, \rulename{C-Ret}, \rulename{C-Bind}}
      \end{plines}
      \]
    \item[Case]
      \[
      \infer[D-Spawn]
{
  b \fresh
}
{
  \kwbspawn{\prio}{\tau}{\cmd}
  \lstep{\sig}{\asil}
  \rpconfig{\sigtype{b}{\tau}{\prio}}{\kwret{\kwtid{b}}}{\dthread{b}{\prio}{\cmd}}
}
      \]
      \[
      \begin{plines}
        \have{\cceval{\tsig}{\sig}{b}{\prio}{\cmd}{v'}{\graph}{\tsig''}{\sig,\sig''}}
        \cjust{inversion on \rulename{CT-Concat}, \rulename{CT-Thread},
          \rulename{C-Ret}}\\
        \have{\mceval[\dassn]{\tsig}{\sig}
          {\mthread{a}{\prio'}{\kwbspawn{\prio}{\tau}{\cmd}}}
          {\tsig'',\tsigent{b}{v'}{\sig''}}
          {\sgraph{u} \gscomp{a} \graph}}
        \ljust{\rulename{C-Spawn}}
      \end{plines}
      \]
    \item[Case]
      \[\infer[D-Sync1]
{
  e \estep{\sig} e'
}
{
  \kwbsync{e}
  \lstep{\sig}{\asil}
  \rpconfig{\esig}{\kwbsync{e'}}{\emem}
}
      \]
      \[
      \begin{plines}
        \have{\graph = \dagq{\cthread{a}{\prio}{\uthread \tscomp{} u}}
          {\emptyset}{\{(b, u)\}}{\emptyset},
          \ceval[\dassn]{e'}{\kwtid{b}}{\uthread}}
        \just{inversion on \rulename{C-Sync}}\\
        \have{\ceval[\dassn]{e}{\kwtid{b}}{u' \tscomp{\delay + 1} \uthread}}
        \just{part 1}\\
        \have{\cceval{\tsig,\tsigent{b}{v}{\sig_b}}{\sig}{a}{\prio}
          {\kwbsync{e}\\ &}{v}
          {\dagq{\cthread{a}{\prio}{u' \tscomp{} \uthread \tscomp{} u}}
            {\emptyset}{\{(b, u)\}}{\emptyset}}{\tsig'',\tsigent{b}{v}{\sig_b}}{\sig''}}
        \just{\rulename{C-Sync}}
      \end{plines}
      \]

    \item[Case]
      \[\infer[D-Sync2]
{
  v \val{\sig}%\\
  %\ple{\prio}{\prio'}
}
{
  \kwbsync{(\kwtid{b})}
  \lstep[\prio]{\sig}{\async{b}{v}}
  \rpconfig{\esig}{\kwret{v}}{\emem}
}
      \]
      \[
      \begin{plines}
        \have{\graph = \egraph}
        \cjust{inversion on \rulename{C-Ret}}\\
        \have{\cceval{\tsig, \tsigent{b}{v}{\sig_b}}{\sig}{a}{\prio}
          {\kwbsync{(\kwtid{b})}}{v}
          {\dagq{\cthread{a}{\prio}{\sthread{u}}}{\emptyset}
            {\{(b, u)\}}{\emptyset}}
          {\tsig''}{\sig''}}
        \cjust{\rulename{C-Sync}}
      \end{plines}
      \]
    \item[Case]
      \[
      \infer[D-Ret]
{
  e \estep{\sig} e'
}
{
  \kwret{e}
  \lstep{\sig}{\asil}
  \rpconfig{\esig}{\kwret{e'}}{\emem}
}
       \]
       \[
       \begin{plines}
         \have{\graph =
           \dagq{\cthread{a}{\prio}{\uthread}}{\emptyset}{\emptyset}
                {\emptyset}, \ceval[\dassn]{e'}{v}{\uthread}}
         \cjust{inversion on the cost semantics}\\
         \have{\ceval[\dassn]{e}{v}{u \tscomp{\delay + 1} \uthread}}
         \cjust{part 1}\\
         \have{\cceval{\tsig}{\sig}{a}{\prio}{\kwret{e}}{v}
           {\dagq{\cthread{a}{\prio}{u \tscomp{}\uthread}}{\emptyset}
             {\emptyset}{\emptyset}}{\tsig''}{\sig''}}
         \cjust{\rulename{C-Ret}}
       \end{plines}
       \]
    \end{itemize}

  \item
\[
\begin{plines}
  \have{\mem' \meq \mbconfig{\sig}{\mem_0 \mcp \dthread{a_1}{\prio_1}{\cmd_1'}
      \mcp \mem_1' \mcp \dots \mcp \dthread{a_n}{\prio_n}{\cmd_n'} \mcp \mem_n'}}
  \ljust{inversion on \rulename{DT-Par}.}\\
  \have{\graph' = \graph_0 \uplus \graph_1' \uplus \dots \uplus \graph_n',
        \forall i. \exists \tsig_i, \sig_i. \mceval[\dassn]{\tsig_i}{\sig_i}
                {\mthread[\delay_i]{a_i}{\prio_i}{\cmd_i'} \mcp \mem_i'}
                {\tsig_i'}
                {\graph_i'}}
  \ljust{inversion on the cost semantics}\\
  \have{\mthread[0]{a_i}{\prio_i}{\cmd_i} \lstep{\sig}{\act_i}
        \rpconfig{\sig_i'}{\cmd_i'}{\mem_i'}}
  \cjust{inspection of transition rules}\\
  \have{\cceval{\tsig_i}{\sig_i}{a_i}{\prio_i}{\cmd_i}{v}{\graph_i}
        {\tsig_i'}{\sig_i'}, \graph_i~\text{is isomorphic to}~
        \sgraph{u_i} \gscomp{a} \graph_i'}
      \cjust{part 2}\\
      \have{\mceval[\dassn]{\tsig}{\sig}{\mem}{\tsig'}{\graph},
        \graph~\text{is isomorphic to}\\ &~
        \graph_0 \uplus (\sgraph{u_1} \gscomp{a_1} \graph_1) \uplus \dots \uplus
        (\sgraph{u_n} \gscomp{a_n} \graph_n)}
      \cjust{\rulename{C-Concat}}
    \end{plines}
    \]

  \end{enumerate}
\end{proof}
\else
See the technical report~\citep{partial-prio-tr} for the proof, which is by induction on the
transition derivation.
\fi

We can now repeatedly apply the above results to show a step-by-step
correspondence between arbitrary executions of {\calcname} programs and
schedules of the corresponding DAG.
To be more precise, we show that, for any execution of a program, there exists
a cost graph~$\graph$ corresponding to the program, and a schedule of~$\graph$
that corresponds to the execution.
If the threads at each parallel transition are chosen in a ``prompt''
manner by stepping as many threads as possible and prioritizing high-priority
threads, then the corresponding schedule is prompt.
Specifying how to pick the threads in a parallel transition is out of the
scope of this paper, though we briefly discuss an appropriate scheduling
algorithm at an implementation level in Section~\ref{sec:impl}.

\begin{lemma}\label{lem:schedule}
  Suppose $\mtyped{\worlds}{\esig}{\mem}$ and
  $\mem \pgstep{}^* \mem'$ where
  $\mceval[\dassn]{\esig}{\esig}{\mem'}{\esig}{\emptyset}$
  and thread~$a$ is active for~$T$
  transitions
  and at each transition, threads are chosen in a prompt
  manner. Then $\mceval[\dassn]{\esig}{\esig}{\mem}{\esig}{\graph}$ and
  there exists a prompt schedule of~$\graph$ in
  which~$\resptimeof{a} = T$.
\end{lemma}
\begin{proof}
  By induction on the derivation of~$\mem \pgstep{}^* \mem'$.
  If~$\mem=\mem'$, then the result is clear.
  Suppose $\mem \pgstep{\{a_1, \dots, a_n\}} \mem''
  \pgstep{}^* \mem'$, and~$a$ is active for~$T$ transitions of the latter
  execution.
  By induction, $\mceval[\dassn]{\esig}{\esig}{\mem''}{\esig}{\graph''}$ and there exists a
  prompt schedule of~$\graph''$ where~$\resptimeof{a} = T$.
  By Lemma~\ref{lem:step-cost}, $\graph''$ is isomorphic to
  $\graph_0 \uplus \graph_1' \dots \graph_n'$ and
  $\mceval[\dassn]{\esig}{\esig}{\mem}{\esig}{\graph}$, where
  $\graph$
  is isomorphic to
  $\graph_0 \uplus (\sgraph{u_1} \gscomp{a_1} \graph_1') \uplus
  \dots \uplus (\sgraph{u_n} \gscomp{a_n} \graph_n')$.
  By Lemma~\ref{lem:global-ready-threads}, these threads are ready
  in~$\graph$, so the schedule that executes~$u_1, \dots, u_n$ in step 1 and
  then follows the schedule of~$\graph''$ is a valid schedule
  of~$\graph$. Because (also by Lemma~\ref{lem:global-ready-threads}), all
  threads that are ready in~$\graph$ are available to be executed and (by
  inspection of the cost semantics) thread priorities are preserved
  between~$\graph$ and~$\mem$, the
  schedule is also a prompt schedule of~$\graph$.
  If~$a \in \dom{\mem}$, then by
  Lemma~\ref{lem:global-ready-threads},
  $a$ is ready in~$\graph$ and the resulting schedule has
  $\resptimeof{a} = T+1$. Otherwise, the resulting schedule has
  $\resptimeof{a} = T$.
\end{proof}

Finally, we conclude by applying Theorem~\ref{thm:gen-brent}
to bound the response time of prompt schedules, and therefore of the
corresponding executions of the operational semantics.

\begin{thm}\label{thm:cost-bound}
  If $\cmdtyped{\esig}{\ectx}{\cmd}{\tau}{\prio}$ and
  $\dthread{a}{\prio}{\cmd} \pgstep{}^* \mem'$, where
  $\mceval[\dassn]{\esig}{\esig}{\mem'}{\esig}{\egraph}$
  and thread~$a$ is active for~$T$
  transitions and at each transition, threads are chosen in a prompt
  manner, then there exists a graph~$\graph$ such that
  $\cceval{\esig}{\esig}{a}{\prio}{\cmd}{v}{\graph}{\tsig}{\sig}$
  and  % for any~$\ple{\prio'}{\prio}$,
  \[
  E[T] \leq %\frac{1}{\fc(\psnle{\prio'})}
  \frac{\prioworkof{\compwork{}{a}}{\psnle{\prio}}}{P} +
      \longp{\compwork{}{a}}{a}
      \]
\end{thm}
\begin{proof}
  By Lemma~\ref{lem:schedule}, there exists such a~$\graph$ and a
  prompt schedule of~$\graph$
  where~$\resptimeof{a} = T$. By Lemma~\ref{lem:typed-wf},~$\graph$ is
  well-formed. Thus, the result follows from Theorem~\ref{thm:gen-brent}.
\end{proof}

%%% Local Variables:
%%% mode: latex
%%% TeX-master: "main"
%%% End:

%\input{elab}
\section{Starvation and Fairness}
Throughout this paper, we assume that higher-priority threads should
always be given priority over lower-priority ones.
This is the desired semantics in many applications, but not all:
sometimes, it is important to be {\em fair} and devote a certain
fraction of cycles to lower-priority work.
Fairness raises a number of interesting theoretical and practical
questions the full treatment of which are beyond the scope of this
paper.
We note, however, that fairness is largely orthogonal to our results
and it is not difficult to extend our results (e.g., those in
Section~\ref{sec:dag}) to devote a
fraction~$L$ of processor cycles to lower-priority work.
This simply inflates the response time bounds by a factor
of~$\frac{1}{1-L}$ to account for time not devoted to being prompt.
%
%% For notational simplicity, we have chosen not to include fairness in our
%% cost bounds in the body of the paper.
%% %
A discussion of cost bounds accounting for fairness can be found in
\ifappendix
Appendix~\ref{app:fair}.
\else
an appendix in the attached supplementary material.
\fi

\section{Implementation}\label{sec:impl}

We have developed a prototype implementation of {\langname}. Our
implementation compiles {\langname}
to~\texttt{mlton-parmem}~\citep{RMAB-mm-2016}, a parallel extension of
Standard ML which is derived from the work of~\citet{spoonhower:phd}.
We have also developed a parallel scheduler for {\langname} programs,
which plugs into the~\texttt{mlton-parmem} runtime.  The
implementation allows programmers to use almost all of the features of
Standard ML, including datatype declarations, higher-order functions,
pattern matching, and so on. While {\langname} itself does not have
a module system and expects all {\langname} code to be in one file
(a limitation we inherit from the compiler on whose
elaborator we build), our implementation is designed so that code may
freely interface with the Standard ML basis library
and SML modules defined elsewhere.

We will describe the two components of the implementation (compilation to
parallel ML and the scheduler) separately.

\subsection{Compilation to Parallel ML}
Our compiler modifies the parser and elaborator of ML5/pgh~\citep{murphy:phd},
which also extends Standard ML with modal constructs, although
for a quite different purpose. Elaboration converts the {\langname}
abstract syntax tree to a typed intermediate language, and type checks
the code in the process. At the same time, the elaborator collects the
priority and ordering declarations into a set of worlds and a set of
ordering constraints (raising a type error if inconsistent ordering
declarations ever cause a cycle in the ordering relation).

For our purposes, the elaboration pass is used only for
type checking. We generate the final ML code from the original AST (which is
closer to the surface syntax of ML), so as not to produce highly obfuscated
code. Before generating the code, the compiler passes over the
AST, converting {\langname} features into SML with the parallel extensions of
\code{mlton-parmem}. Priority names and variables are converted into ordinary
SML variables. Priority-polymorphic functions become ordinary functions, with
extra arguments for the priorities, and their instantiations become function
applications. Commands and instructions become SML expressions, with a sequence
of bound instructions becoming a let binding. Encapsulated commands become
thunks (so as to preserve the semantics that they are delayed). We compile
threads using Spoonhower's original implementation of parallel futures:
spawn commands spawn a future, and sync commands force the future.

The AST generated by the above process is then prefaced by a series of
declarations which register all of the priorities
and ordering constraints with the runtime, and bind the priority names
to the generated priorities. The compiler finally generates Standard ML code
from the AST, and passes it to~\code{mlton-parmem} for compilation to an
executable.

\subsection{Runtime and Scheduler}
The runtime for {\langname} is written in Standard ML as a scheduler for
\code{mlton-parmem}. As described above, before executing the program code,
{\langname} programs call into the runtime to register the necessary priorities
and orderings. The runtime then uses Warshall's transitive closure
algorithm to build the full partial order and
stores the result, so that checking the ordering on two priorities at runtime
is a constant-time operation. It then performs a topological sort on the
priorities to convert the partial order into a total order which is compatible
with all of the ordering constraints. Once this is complete, the program runs.

In our scheduling algorithm, each processor has a private
deque~\citep{AcarChRa13} of tasks for each priority, ordered by the total
order computed above. Each processor works on its highest-priority task (in
the total order, which guarantees it has no higher-priority task in the
partial order). A busy processor~$q_1$ will periodically preempt its
work and pick another ``target'' processor~$q_2$ at random. Processor~$q_1$
will send work to~$q_2$ at an arbitrarily chosen priority, if~$q_2$ has no
work at that priority. It will then start the process over by finding its
highest-priority task (which may have changed if another processor has sent it
work) and working on it.

\subsection{Examples}
We have implemented five sizable programs in {\langname}. These include the
email client of Section~\ref{sec:overview} and a bank example inspired
by an example used to justify partially-ordered
priorities~\citep{bms-formal-1993}.
We have also adapted the Fibonacci server, streaming music and
web server benchmarks of our prior work~\citep{MullerAcHa17}.
These originally used only two priorities; we generalized them with a more
complex priority structure, and implemented them in {\langname}.

\paragraph{Email Client}
We have implemented the ``email client'', portions of which appear in
Section~\ref{sec:overview}. The program parses emails stored locally,
and is able to sort them by sender, date or subject, as
requested by the user in an event loop at priority~\code{loop\_p} (which
currently just takes the commands at the terminal; we don't yet have a graphical
interface). The user can also issue commands to send an email (stored as a
file) or quit the program.

\paragraph{Bank Simulator}
\citet{bms-formal-1993}
give the example of a banking system that can perform operations {\em query},
{\em credit} and {\em debit}. To avoid the risk of spurious overdrafts, the
system prioritizes credit actions over debit actions, but does not restrict the
priority of query actions. We implement such a system, in which a foreground
loop (at a fourth priority, higher than all of the others), takes
query, credit and debit commands and spawns threads to perform the corresponding
operations on an array of ``accounts'' (stored as integer balances).
%The threads are prioritized using the three priorities described above.

\paragraph{Fibonacci Server}
The Fibonacci server runs a foreground loop at the highest
priority~\code{fg} which takes a number~$n$ from the user,
spawns a new thread to compute the~$n^{th}$ Fibonacci
number in parallel, adds the spawned thread to a list, and
repeats. The computation is run at one of three
priorities (in order of decreasing priority):
\code{smallfib}, \code{medfib} and \code{largefib},
depending on the size of the computation, so smaller computations will be
prioritized. When the user
indicates that entry is complete, the loop terminates, prints a message at
priority \code{alert} (which is higher than \code{smallfib} but incomparable
with \code{fg}), and returns the list of threads to the main thread, which
syncs with all of the running threads, waiting for the Fibonacci computations
to complete (these syncs can be done safely since the main thread runs at
the lowest priority \code{bot}).

\paragraph{Streaming Music}
We simulate a hastily-monetized music streaming service, with a server thread
that listens (at priority~\code{server\_p}) for network
connections from clients, who each request a music file.
For each client, the server spawns a new thread which loads the requested file
and streams the data over the network to the client. The priority of this thread
corresponds to the user's subscription (the free Standard service or the paid
Premium and Deluxe subscriptions). Standard is
lower-priority than both Premium and Deluxe. Due to boardroom in-fighting,
it was never decided whether Premium or Deluxe subscribers get a higher level
of service, and so while both are higher than Standard, the Premium and Deluxe
priorities are incomparable. Both are lower than~\code{server\_p}. This benchmark
is designed to test how the system handles multiple threads performing
interaction; apart from the
asynchronous threads handling requests, no parallel computation is performed.

\paragraph{Web Server}
Like the server of the music service, the web server listens for connections
in a loop at priority~\code{accept\_p} and spawns a thread (always at
priority~\code{serve\_p}) for each client to respond to HTTP requests. A
background thread (priority~\code{stat\_p}) periodically traverses the
request logs and analyzes them (currently, the
analysis consists of calculating the number of views per page, together with
a large Fibonacci computation to simulate a larger job). Both~\code{accept\_p}
and~\code{serve\_p} are higher-priority than~\code{stat\_p}, but the ordering
between them is unspecified.

\subsubsection{Evaluation}
While a performance evaluation is outside the scope of this paper, we have
completed a preliminary performance evaluation of the scheduler described above.
We have evaluated the performance of the web server benchmark described above,
as well as a number of smaller benchmarks which allow for more controlled
experiments and comparisons to prior work. In all cases, we have observed
good performance and scaling.
The web server, for example, scales easily to~50 processors and~50 concurrent
requests while keeping response times to~6.5 milliseconds.
%More performance results are available in
%\ifappendix
%Appendix~\ref{app:eval}.
%\else
%an appendix uploaded with the supplementary material.
%\fi

%\input{poll}
\section{Related Work}\label{sec:related}

%Multithreading, priorities, and modal type systems are large fields in
%their own right.
%
In this section, we review some of the most closely related papers
from fields such as multithreading and modal type systems, and
discuss their relationship with our work.

\paragraph{Multithreading and Priorities.}
Multithreaded programming goes back to the early days of computer science,
such as the work on Mesa~\cite{lr-mesa-1980}, Xerox's
STAR~\citep{Xerox-STAR-1982}, and Cedar~\citep{Cedar-1986}.
These systems allow the programmer to create (``fork'') threads, and
synchronize (``join'') with running threads.
The programmer can assign priorities, generally chosen from a fixed
set  of natural numbers (e.g.,~$7$ in Cedar), to threads,
allowing those that execute latency-sensitive computations to have a
greater share of resources such as the CPU.

Our notion of priorities is significantly richer than those considered
in prior work, because we allow the
programmer to create as many priorities as needed, and impose an
arbitrary partial order on them.
Several authors have observed that partial orders are more expressive
and more desirable for programming with priorities than total
orders~\cite{bms-formal-1993,fidge-formal-1993}.
There is little prior work on programming language support for
partially ordered priorities. The only one we know of is the occam
language, whose expressive power is limited, leaving the potential for
ambiguities in priorities~\citep{fidge-formal-1993}.

Some languages, such as Concurrent ML~\citep{cml99}, don't expose
priorities to the programmer, but give higher priority at runtime
to threads that perform certain (e.g. interactive) operations.

\paragraph{Priority Inversion.}
Priority inversion is a classic problem in
multithreading systems.
\citet{lr-mesa-1980} appear to be the first to observe it in their work
on Mesa.
Their original description of the problem uses three threads with
three different priorities, but the general problem can be restated
using just two threads at different
priorities (e.g. \citep{bms-formal-1993}).

%
%Since priority inversions can lead to undesirable program behavior,
%there has been much work on techniques for preventing them.
%
%The priority inheritance technique allows the priority of low-priority
%threads to be increased at run-time in order prevent higher priority
%threads from being delayed~\cite{lr-mesa-1980,srl-priority-1990}.
%
Babao{\u{g}}lu, Marzullo, and Schneider
%deem the complexity of the
%analysis of priority inheritance techniques to be a
%disadvantage~\cite{bms-formal-1993}.  They then
provide a formalization
of priority inversions and describe protocols for
preventing them in some settings, e.g. transactional systems.
%

%% Beyond its complexity, priority inheritance
%% can promote a low-priority thread, which can then
%% take a significant time to complete.
%% %
%% As a result, reasoning about the cost of a program becomes
%% challenging.
%% %
%% If instead, programmers are alerted statically to priority inversions,
%% they could fix them and guarantee the desired properties.
%% %
%% In this paper, this is the approach that we take.
%% %
%% Instead of run-time heuristics, we present a type system that rejects
%% programs with priority inversions, guaranteeing thus that type-safe
%% programs are ``well behaved.''

\paragraph{Parallel Computing.}
Although earlier work on multithreading was driven primarily by the
need to develop interactive systems~\citep{hauserjathwewe93}, multithreading
has also become an important paradigm for parallel computing. 
In principle, a multithreading system such as pthreads can be used to
perform parallel computations by creating a number of threads and
distributing the work of the computation among them.
This approach, sometimes called ``flat parallelism,'' has numerous
disadvantages and has therefore given way to a higher-level
approach, sometimes called ``implicit threading'', in which
the programmer indicates the computations that
can be performed in parallel using constructs such as ``fork'' and
``join''. The language runtime system creates and manages the
threads as needed.
In addition to the focus on throughput rather than responsiveness,
cooperative systems differ from competitive systems in that they typically
handle many more, lighter-weight threads.
The ideas of implicit and cooperative threading go back to early
parallel programming languages such as Id~\citep{id-78} and
Multilisp~\citep{halstead85}, and many languages and systems have
since been developed.

\paragraph{Cost Semantics}
Cost semantics, broadly used to reason about resource usage
\citep{Rosendahl89,Sands90}, have been
deployed in many domains. We build in particular on cost models that
use DAGs to reason about parallel
programs~\citep{BlellochGr95,BlellochGr96,SpoonhowerBlHaGi08}.
These models summarize the parallel structure of a computation in the cost
metrics of {\em work} and {\em span}, which can then be used to
bound computation time. While finding an optimal schedule of a DAG
is NP-hard~\citep{Ullman75}, \citet{Brent74} showed that a ``level-by-level''
schedule is within a factor of two of optimal. \citet{EagerZaLa89} extended
this result to all greedy schedules.

While these models have historically been applied to cooperatively
threaded programs, in recent work we have extended them to handle
latency-incurring operations~\citep{ma-latency-2016}, and presented a
DAG model which enables reasoning about responsiveness in addition to
computation time~\citep{MullerAcHa17}.  This prior work introduced the
idea of a prompt schedule, but considers only two priorities.  Our
cost semantics in this paper applies to programs with a partially
ordered set of priorities.

\paragraph{Modal and Placed Type Systems.}
A number of type systems have been based on various modal logics, many
of them deriving from the judgmental formulation of
\citet{pfenning+:judgmental}. While we did not strictly base our type system
on a particular logic, many of our ideas and notations are inspired by
S4 modal logic and prior type systems based on modal logics.
\citet{moody03} used a type system
based on S4 modal logic to model distributed computation, allowing
programs to refer to results obtained elsewhere (corresponding in the
logical interpretation to allowing proofs to refer to ``remote
hypotheses''). It is not made clear, however, what role the asymmetry of S4
plays in the logic or the computational interpretation. Later type systems
for distributed computation~\citep{JiaWa04, Murphy07} used an explicit worlds
formulation of S5, in which the ``possible worlds'' of the modal logic are
made explicit in typing judgment. Worlds are interpreted
as nodes in the distributed system, and an expression that is well-typed at
a world is a computation that may be run on that node. Both type systems also
include a ``hybrid'' connective~$A~\kw{at}~w$, expressing the truth of a
proposition~$A$ at a world~$w$. They interpret proofs of such a proposition
as encapsulated computations that may be sent
to~$w$ to be run. Our type system uses a form of both of these features;
priorities are explicit, and the
types~$\kwcmdt{\tau}{\prio}$ and~$\kwat{\tau}{\prio}$ assign
priorities to computations.
Unlike prior work, we give an interpretation to the asymmetry of
the accessibility relations of S4 modal logic, as a partial order of thread
priorities.

A different but related line of work concerns
type systems for staged computation, based on linear temporal logic (LTL)
(e.g.~\citep{davies96,faaf-2016}). In these systems, the ``next'' modality of
LTL is interpreted as a type of computations that may occur at the next stage
of computation. In prior work~\citep{MullerAcHa17} we adapted these ideas to a
type system for prioritized computation with two priorities:
background and foreground. In principle, a priority type system based on LTL
could be generalized to more than two priorities, but (because of the
``linear'' of LTL), such systems would be limited to totally ordered priorities.

Place-based systems~(e.g. \citep{Yelick98, x10-2005, Chandra08}),
like the modal type systems
for distributed computation, also interpret computation as located at a
particular ``place'' and use a type system to enforce locality of resource
access. These systems tend to be designed more for practical concerns rather
than correspondence with a logic.

\section{Conclusion}
We present techniques for writing parallel interactive programs where
threads can be assigned
partially ordered priorities.
A type system ensures proper usage of priorities by precluding
priority inversions and a cost model enables predicting the
responsiveness and completion time properties for programs.
We implement these techniques by extending the Standard ML language
and show a number of example programs.
Our experiments provide preliminary evidence that the proposed
techniques can be effective in practice.

%% %
%% It is our hope that this work can form
%% the basis of a new paradigm for writing multithreaded interactive programs.

%\iftr\else\begin{anonsuppress}\fi
\section*{Acknowledgements}
The authors would like to thank Frank Pfenning and Tom Murphy VII for their
helpful correspondence on related work.

This work was partially supported by the National Science Foundation under
grant number CCF-1629444.
%\iftr\else\end{anonsuppress}\fi

%\clearpage
%\bibliographystyle{plainnat}  % author-year citations, longer bibliography
%\bibliographystyle{ACM-Reference-Format}
\bibliography{bib-main-2017,bib-new,../papers/read/annot.bib,new,thisbib}

\ifappendix
\appendix
\section{Scheduling with Fairness}\label{app:fair}

\subsection{The Fair and Prompt Scheduling Principle.}
To avoid starvation, we introduce a notion of fairness, in the form of
a parameter we call the {\em fairness criterion}.
A fairness criterion $\fc$
is a discrete probability distribution over priorities: a mapping from
priorities to real numbers in the range $[0, 1]$, summing to 1.
When threads of {\em all priorities} are present in the system, the scheduler
should devote, on average, the fraction~$\fc(\prio)$ of
processor cycles to threads at priority~$\prio$.
When a particular priority is unavailable (i.e. has
no threads available in the system), it ``donates'' its cycles to the highest
available priority.
This policy is flexible enough to encode many
application-specific scheduling policies.
For example, if $\top$ is the highest priority, we can encode a form of
prompt scheduling in which as many processors as possible are devoted to the
highest-priority work available, followed by the next-highest, and so on.
The fairness criterion for such a policy would be~$\fc(\top) = 1$ and
$\fc(\prio) = 0$ for all $\prio \neq \top$.

A {\em fair and prompt schedule}, parameterized by a fairness criterion~$\fc$,
is a schedule that adheres to the principle described intuitively above.
At each time step, processors are assigned to priorities probabilistically
according to the distribution~$\fc$
\footnote{In the continuous limit as~$P$
  approaches~$\infty$, this is equivalent to simply dividing the processors
  according to~$\fc$.}.
Processors attempt to execute a ready vertex at their assigned priority.
Processors that are unable to execute a vertex at their assigned priority
default to the ``prompt'' policy and execute the highest-priority ready vertex.

One may think of a schedule as a form of pebbling: if~$P$ processors are
available, at each time step, place up to~$P$ pebbles on ready vertices until
all vertices have been pebbled.
In the pebbling analogy, executing a program using a fair and prompt schedule
may be seen as pebbling a graph by drawing up to~$P$ pebbles at a time at
random from an infinite bag of pebbles.
Each pebble is colored, and each color is associated with a priority.
The colored pebbles in the bag are distributed according to the fairness
criterion~$\fc$.
When pebbles are drawn at a time step, we attempt to place each one on a
vertex of the appropriate priority.
Any pebbles that are left are placed on the highest-priority vertices, then
the next-highest, and so on.

\subsection{Bounding Response Time}
We are now ready to bound the response time of threads in fair and prompt
schedules using cost metrics which we will now define.

The {\em priority work}~$\prioworkof{\graph}{\prio}$ of a graph~$\graph$ at
a priority~$\prio$ is defined as the number of vertices in the graph at
priority~$\prio$:

\[\prioworkof{\graph}{\prio} =
|\{u \in \graph \mid \uprio{\graph}{u} = \prio \}|\]

Recall that the response time should depend only on work at a higher,
equal or unrelated priority to the priority of the thread being observed.
Because this set of priorities will come up in many places while analyzing
response time, we use the notation~$\psnle{\prio}$ to refer to the set of
priorities not less than~$\prio$.
We use two convenient shorthands related to summing quantities over the set
$\psnle{\prio}$.
\[
\begin{array}{r c l}
  \priowork{\psnle{\prio}}{\graph} & \defeq &
  \sum_{\nple{\prio'}{\prio}} \priowork{\prio'}{\graph}\\
  \fc(\psnle{\prio}) & \defeq &
  \sum_{\nple{\prio'}{\prio}} \fc(\prio')
\end{array}
\]

Theorem~\ref{thm:gen-brent-fair} bounds the expected response time (because fair and
prompt schedules are defined probabilistically, the respones time can
only be bounded in expectation) of a thread based on these quantities
which depend only on the fairness criterion and the work and span of
high-priority vertices.

\begin{thm}\label{thm:gen-brent-fair}
  Let~$\graph$ be a well-formed DAG.
  For any fair and prompt schedule on~$P$ processors and
  any~$\ple{\prio'}{\prio}$,
  \[
  E[\resptimeof{a}] \leq \frac{1}{\fc(\psnle{\prio'})}
  \left(\frac{\prioworkof{\compwork{}{a}}{\psnle{\prio'}}}{P} +
      \longp{\compwork{}{a}}{a}\right)
      \]
\end{thm}
\begin{proof}
  Let~$s$ and~$t$ be the first and last vertices of~$a$, respectively.
  Consider the portion of the schedule from the step in which~$s$ is ready
  (exclusive) to the step in which~$t$ is executed (inclusive).
  At each step, let~$P_\fc$ be the number of processors attempting to work on
  vertices of priority in~$\psnle{\prio'}$.
  For each processor at each step, place a token in one of three buckets.
  If the processor is attempting to work at a priority not less than~$\prio'$,
  but is unable to, place a token in the ``low'' bucket~$B_l$.
  If it attempting to work at a priority not less than~$\prio'$ and succeeds,
  place a token in the ``high'' bucket~$B_h$.
  If it is attempting to work at a priority less than~$\prio'$, place a
  token in the ``fair'' bucket~$B_f$.
  Because~$P$ tokens are placed per step,
  we have $\resptimeof{a} = \frac{1}{P}(B_l + B_h + B_f)$,
  where~$B_l$, $B_h$
  and~$B_f$ are the number of tokens in the buckets after~$t$ is executed.

  Let~$\Sigma = \fc(\psnle{\prio'})$.
  By fairness, we have $B_l + B_h = \Sigma(B_l + B_h + B_f)$.
  Thus,
  \[\resptimeof{a} = \frac{1}{P\Sigma}(B_l + B_h)\]
  Each token in~$B_h$ corresponds to work done at priority not less
  than~$\prio'$, and thus~$B_h \leq \prioworkof{\graph}{\psnle{\prio'}}$, so
  \[\resptimeof{a} \leq \frac{1}{\Sigma}\left(
  \frac{\prioworkof{\graph}{\psnle{\prio'}}}{P}
  + \frac{B_l}{P}\right)\]
  We now need only bound~$B_l$ by~$P\cdot \longp{\compwork{}{a}}{a}$.

  Let step 0 be the step after~$s$ is ready, and let~$\exec{j}$ be the
  set of vertices that have been executed at the start of step~$j$.
  Consider a step~$j$ in which a token is added to~$B_l$.
  For any path ending at~$t$ consisting of vertices
  of~$\graph \setminus \exec{j}$, the path starts at a vertex that is ready
  at the beginning of step~$j$.
  By the definition of well-formedness, this
  vertex must have priority greater than~$\prio$ and is therefore executed
  in step~$j$ by the prompt principle.
  Thus, the length of the path decreases by 1 and so
  $\longp{\graph \setminus \exec{j+1}}{t} =
  \longp{\graph \setminus \exec{j}}{t} - 1$.
  The maximum number of such steps is
  thus~$\longp{\graph \setminus \exec{0}}{a}$, and
  so~$B_l \leq P\cdot \longp{\graph \setminus \exec{0}}{a}$.
  Because~$\noancs{s} \supset \graph \setminus \exec{0}$,
  any path excluding vertices
  in~$\exec{0}$ is contained in~$\noancs{a}$, and
  $\longp{\graph \setminus \exec{0}}{a} \leq
  \longp{\noancs{a}}{a}$, so
  $B_l \leq P\cdot \longp{\noancs{a}}{a} = P\cdot \longp{\compwork{}{a}}{a}$.
\end{proof}

One might wonder about the purpose of the univerally quantified
priority~$\prio'$.
We illustrate its use by considering two extremal cases.
First, let~$\prio' = \prio$.
Doing so yields the bound that might intuitively be expected: the response time
of a thread at priority~$\prio$ depends on the work and span at priorities
not less than~$\prio$, inflated somewhat by the fairness criterion because
only some of the cycles are devoted to work at priorites higher than~$\prio$.
This bound is correct and frequently useful, but will diverge in
the case that~$\fc(\psnle{\prio}) = 0$.
In such cases, it may be worthwhile to look at a bound which considers cycles
donated from priorities lower than~$\prio$.
This increases the factor~$\frac{1}{\fc(\psnle{\prio})}$ at the cost of
having to consider more vertices as competitors.

In the extreme case, we set~$\prio' = \bot$.
As above, the multiplicative factor for the fairness criterion goes to 1,
but~$\prioworkof{\compwork{}{a}}{\psnle{\prio'}}$ becomes all of the work,
at any priority, that may happen in parallel with thread~$a$.
Intuitively, this says that if we ignore priority and simply run a greedy
schedule, thread~$a$ will complete eventually but may, in the worst case,
have to wait for all of the other work in the system to complete.

\fi

\end{document}